\pdfoutput=1

\documentclass{article}
\usepackage{ijcai25}

\usepackage{times}
\usepackage{soul}
\usepackage{url}
\usepackage[hidelinks]{hyperref}
\usepackage[utf8]{inputenc}
\usepackage[small]{caption}
\usepackage{graphicx}
\usepackage{algorithmic}
\usepackage[switch]{lineno}

\usepackage{subcaption}
\usepackage{svg}
\usepackage{tikz}
\usetikzlibrary{automata,positioning,shapes,arrows,shadows,shadows.blur, fit, decorations.pathmorphing}

\usepackage{amsmath}
\usepackage{amsthm}
\usepackage{amssymb}
\usepackage{xspace} %
\usepackage[ruled, vlined, linesnumbered]{algorithm2e}
\usepackage{multirow}  %
\usepackage{booktabs}
\usepackage{diagbox}
\usepackage{enumitem} %
\usepackage{pifont} %
\usepackage{tcolorbox}  %
\usepackage{xcolor,colortbl}

\newtheorem{problem}{Problem}

\newtheorem{theorem}{Theorem}

\newtheorem{definition}{Definition}

\newtheorem{corollary}{Corollary}
\newtheorem{lemma}{Lemma}

\newtheorem{example}{Example}

\newcommand{\B}{\mathcal{B}}

\newcommand{\G}{\mathcal{G}}    %

\newcommand{\Tau}{\mathrm{T}}   %
\renewcommand{\phi}{\varphi}
\newcommand{\play}{P}

\DeclareMathOperator{\plays}{Plays} %
\DeclareMathOperator{\last}{last} %
\DeclareMathOperator{\Val}{Val} %
\DeclareMathOperator{\sVal}{sVal} %
\DeclareMathOperator{\aVal}{aVal}   %
\DeclareMathOperator{\cVal}{cVal}   %
\DeclareMathOperator{\acVal}{acVal}   %
\DeclareMathOperator{\aValues}{aValues}   %

\DeclareMathOperator*{\argmin}{arg\,min}

\newcommand{\U}{\mathcal{U}}

\newcommand{\Next}{\raisebox{-0.27ex}{\LARGE$\circ$}}
\newcommand{\Until}{\mathop{\U}}

\newcommand{\subtitle}[1]{%
 \posttitle{%
  \par\end{center}
  \begin{center}\large#1\end{center}
  \vskip0.5em}%
}

\newcommand{\ltl}{\textsc{ltl}\xspace}
\newcommand{\ltlf}{\textsc{ltl}_f\xspace}

\DeclareMathOperator{\wco}{\mathbf{WCO}}
\DeclareMathOperator{\scoop}{\mathbf{SC}}
\DeclareMathOperator{\sco}{\mathbf{SCO}}
\DeclareMathOperator{\modscoop}{\mathbf{mSC}}

\DeclareMathOperator{\coop}{\mathbf{Co-Op}}
\DeclareMathOperator{\wcoop}{\mathbf{WCo-Op}}

\newcommand{\cmark}{\ding{51}}%
\newcommand{\xmark}{\ding{55}}%

\newlist{todolist}{itemize}{2}
\setlist[todolist]{label=$\square$}

\definecolor{Gray}{gray}{0.85}
\definecolor{LightCyan}{rgb}{0.88,1,1}

\newcolumntype{a}{>{\columncolor{Gray}}c}

\newif\ifproof
\prooffalse %

\usepackage{eso-pic} %

\title{Beyond Winning Strategies: Admissible and Admissible Winning Strategies for Quantitative Reachability Games}

\author{
Karan Muvvala\and
Qi Heng Ho\And
Morteza Lahijanian
\affiliations
University of Colorado at Boulder, CO, USA\\
\emails
\{karan.muvvala, qi.ho, morteza.lahijanian\}@colorado.edu
}

\begin{document}

\maketitle

\begin{abstract}
    Classical reactive synthesis approaches aim to synthesize a reactive system that always satisfies a given specification. These approaches often reduce to playing a two-player zero-sum game where the goal is to synthesize a winning strategy. However, in many pragmatic domains, such as robotics, a winning strategy does not always exist, yet it is desirable for the system to make an effort to satisfy its requirements instead of ``giving up." 
    To this end, this paper investigates the notion of \emph{admissible} strategies, which formalize ``doing-your-best", in quantitative reachability games.
    We show that, unlike the qualitative case, memoryless strategies are not sufficient to capture \emph{all} admissible strategies, making synthesis a challenging task. In addition, we prove that admissible strategies always exist but may produce undesirable optimistic behaviors.
    To mitigate this, we propose \emph{admissible winning} strategies, which enforce the best possible outcome while being admissible. We show that both strategies always exist but are not memoryless. We provide necessary and sufficient conditions for the existence of both strategies and propose synthesis algorithms. Finally, we illustrate the strategies on gridworld and robot manipulator domains.
\end{abstract}

\section{Introduction}

Reactive Synthesis is the problem of automatically generating reactive systems from logical specifications, first proposed by \cite{Church1963}. Its applications span a wide range of domains, including robotics \cite{mcmahon2023reason,he2017reactive,kress2018synthesis}, program synthesis \cite{pnueli1989synthesis}, 
distributed systems \cite{Filippidis2016},
formal verification \cite{kupferman2001model}, and
security \cite{Zhou2003security}. 
Existing approaches to reactive synthesis usually boil down to computing a strategy over a game between a System (Sys) and an Environment (Env) player, 
with the goal of finding a \emph{winning} strategy, which guarantees the Sys player achieves its objectives regardless of the Env player's moves.
In quantitative settings, the game incorporates a payoff requirement as part of the Sys player's objectives.
In many scenarios, however, a winning strategy may not exist, resulting in a failure in the synthesis algorithm. 
Prior works relax this requirement 
using the notions of \emph{best effort} \cite{benjamin2021best} and \emph{admissibility} \cite{berwanger2007admissibility,brandenburger2008admissibility}.  Best effort is specific to qualitative games, while admissibility is recently studied in quantitative games under strong assumptions on the Env player, namely, rationality and known objectives \cite{brenguier2016admissibility}.  
This paper aims 
to study admissible strategies without these assumptions
in quantitative reachability games
with a particular focus on
applications in robotics.

\begin{figure}[t]
    \begin{subfigure}[t]{0.49\linewidth}
        \centering
        \includegraphics[width=0.99\linewidth]{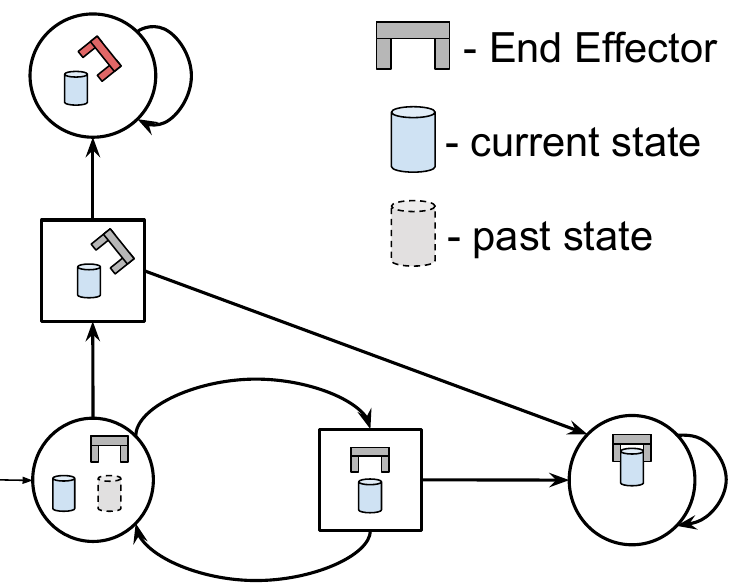}
        \caption{Manipulation Domain}
        \label{fig: manip_motivating_ex}
    \end{subfigure}%
    ~~
    \begin{subfigure}[t]{0.49\linewidth}
        \centering
        \resizebox{\columnwidth}{!}{
        \begin{tikzpicture}[->,>=stealth',shorten >=1pt,shorten <=1pt,auto,node distance=1cm,
            every loop/.style={looseness=6},
            initial text={},
            el/.style={font=\scriptsize},
            every fit/.style={draw,densely dotted,rectangle},
            inner sep=2mm,
            loopright/.style={loop,looseness=6,out=-45, in=45},
            loopleft/.style={loop,looseness=6,out=135, in=225},
            loopabove/.style={loop,looseness=6,out=45, in=135},
            loopbelow/.style={loop,looseness=6,out=-135, in=-45},]
        \tikzstyle{every state}=[node distance=1.4cm,minimum size=7mm, inner sep=1pt]; 
        \node[state, initial] (v0) {$v_0$};
        \node[state, rectangle, right of=v0] (v1) {$v_1$};
        \node[state, accepting, right of=v1] (v2) {$v_2$};
        \node[state, rectangle, above of=v0] (v3) {$v_3$};
        \node[state, above of=v3] (v4) {$v_4$};
        \node[state] at (3, 3) {Sys};
        \node[state, rectangle] at (3, 2) {Env};
        
        \path[->]
            (v0) edge[bend left] node {1} (v1)
            (v1) edge[bend left] node {0} (v0)
            (v1) edge node[below] {0} (v2)
            (v2) edge[loop] node[above] {0} (v2)
            (v0) edge node {1} (v3)
            (v3) edge node {0} (v2)
            (v3) edge node {0} (v4)
            (v4) edge[loop right] node[right] {0} (v4);
        \end{tikzpicture}
        }
        \caption{Two-player game abstraction}
        \label{fig: game_init_ex}
    \end{subfigure}%
    \caption{(a) Robotic manipulator in the presence of human. 
    (b) Game abstraction, where weights represent robot energy.
    }
    \label{fig: game_motivating_ex}
\end{figure}

Consider the example in Fig.~\ref{fig: game_motivating_ex} with a robot (Sys player) and a human (Env player) operating in shared workspace. The robot is tasked with grasping the bin. Since the human can intervene by moving the bin before the robot completes its grasp, there is no winning strategy that enforces completion of the task under the \emph{worst-case} Env strategy. However, in such cases it is still desirable for the robot to make an effort to satisfy its requirements instead of ``giving up."

To this end, this paper studies admissible strategies \cite{faella2009admissible,berwanger2007admissibility} in quantitative reachability games without assumptions on rationality and the objectives of the Env player. These games can model quantitative reactive synthesis for finite-behaviors expressed in, e.g., 
syntactically co-safe Linear Temporal Logic ($\ltl$) \cite{kupferman2001model} and $\ltl$ over finite behaviors ($\ltlf$) \cite{vardi2013ltlf}.
We show that admissible strategies relax requirement of winning strategies, and always exist. We prove that, unlike the qualitative setting, quantitative admissible strategies are generally history-dependent even for finite payoff functions. Then, we show that such strategies can produce overly optimistic behaviors, which may be undesirable for robotics applications. To mitigate this, we propose \emph{admissible winning} strategies, which have the desirable property of enforcing specification satisfaction when possible while being admissible. We prove that similar to admissible strategies, admissible winning strategies always exist and may require finite memory. Then, we provide necessary and sufficient conditions for both strategies 
and propose synthesis algorithms. 
Finally, we provide various robotic examples to show that admissible and admissible winning strategies provide desirable and flexible behaviors without a-priori knowledge of the objectives of other agents. 

Our contributions are fourfold: (i) analysis of admissible strategies in quantitative reachability games without assumptions on rationality and objective of the Env player, including proofs of their existence, finite memory requirement, and necessary and sufficient conditions, (ii) introduction and analysis of novel notion of admissible winning strategies to mitigate over-optimism in admissible strategies, (iii) synthesis algorithms for both strategies, and (iv) illustrative examples on gridworld and manipulation domains, showing emergent behaviors under these strategies.

\paragraph{Related Work.}

Several works explore alternatives to winning strategies.  Specifically, \citeauthor{faella2009admissible} investigates various concepts in qualitative games with reachability objectives within a zero-sum framework. They focus solely on the Sys player's objective without making any assumptions about the Env player.  They use admissibility to define the notion of \emph{best-eﬀort} (BE). \citeauthor{benjamin2021best,de2023symbolic} further examine the complexity of synthesizing BE strategies, showing that it can be reduced to standard algorithms, with memoryless strategies being sufficient.
In contrast, our work considers quantitative reachability games where the objective is to reach a goal state with minimal total cost. We show that memoryless strategies are insufficient in our context, and our synthesis approach does not reduce to standard algorithms.

The notion of admissiblity has also been explored in normal form games \cite{brandenburger2008admissibility,Apt_2011}. In qualitative games with logical specifications (extensive form), admissiblity has been investigated for n-player infinite games, where each player has their own objective and is assumed to play admissibly with respective to that objective -- referred to as \emph{assume admissible} (AA). \citeauthor{berwanger2007admissibility} was the first to formalize this notion of AA for games played on graphs. Subsequently, \citeauthor{brenguier2014complexity,brenguier2015assume} establish the complexity and give algorithms for $\omega$-regular objectives. In our settings, we consider reachability games that terminate in finite time. Notably, we make no assumptions about the Env player, i.e., we neither know Env player's objective nor require them to play admissibly. 

The work closest to ours is by \cite{brenguier2016admissibility}, who study admissibility in quantitative settings. They extend prior work \cite{brenguier2014complexity,brenguier2015assume} on infinite duration qualitative games to quantitative objectives. They give necessary and sufficient conditions for admissible strategies. Unlike their work, in our setting, we consider finite duration games and appropriately define our payoff over finite traces. We show our game is always determined and thus optimal worst-case and cooperative strategies always exist. 
While our analysis shares some conceptual similarities with theirs, addressing the finite play setting requires a distinct theoretical approach compared to infinite plays.
\citeauthor{brenguier2016admissibility} give a sketch of their algorithm based on parity games. We, however, present a detailed yet simpler synthesis algorithm. We also analyze emergent behavior under admissible strategies in robotics settings. We observe that these strategies can be overly optimistic. 
To address this, we identify the underlying cause and propose the concept of admissible winning strategies to mitigate such optimism.

\section{Problem Formulation}
\label{sec: prob_form}

The overarching goal of this work is quantitative reactive synthesis for $\ltlf$  or cosafe $\ltl$ specifications where satisfaction cannot necessarily be guaranteed. 
This problem reduces to reachability analysis in quantitative games played between the Sys and  Env players \cite{baier2008principles}. For the sake of generality, we focus on these games.

\subsection{2-Player Quantitative Games and Strategies}

\begin{definition}[2-player Quantitative Game]
    \label{def: game_abstraction}
 A two-player turn-based quantitative game is a tuple  $\G~=~(V, v_0, A_s, A_e, \delta, C, V_f)$, where
 \begin{itemize}
     \item $V = V_s \cup V_e$ is a finite
     set of states, where $V_s$ and $V_e$ are disjoint and belong to the Sys and Env player,
     \item $v_0 \in V$ is the initial state,
     \item $A_s$ and $A_e$ are the finite sets of  actions for the Sys and Env player, respectively,
     \item $\delta: V \times (A_s \cup A_e) \to V$ is the transition function
     such that,
     for $i,j \in \{s,e\}$ and $i\neq j$, given state $v\in V_i$ and action $a \in A_i$, the successor state is $\delta(v,a) \in V_j$,
     \item $C: V \times (A_s \cup A_e) \to \mathbb{N}^{0}$ is the cost (energy) function such that, for every $(v, a) \in V_s \times A_s$, $C(v,a) > 0$, otherwise $0$, and
     \item $V_f \subseteq V$ is a set of goal (final) states.
 \end{itemize}
\end{definition}

We assume that the game is non-blocking. There is at least one outgoing transition from every state, i.e., $\forall v \in V, \exists a \in A$ s.t. $\delta(v, a) \neq \emptyset$.
Note that the Env action cost is zero since we are solely interested in the Sys player objectives (action costs) and make no assumption about the objective of the Env player. Finally, we assume that our transition function is deterministic
and injective, i.e., $\delta(v, a) = \delta(v, a')$ iff $a = a'$.

The evolution of game $\G$ starts from $v_0$ and is played in turns between the Sys and Env player.
At state $v \in V_i$, where $i \in \{s,e\}$, Player $i$ picks an action $a \in A_i$ and incurs cost $C(v,a)$.
Then, the game evolves to the next state according to the transition function $\delta(v,a) \in V_j$, where $j \neq i$.  Then, Player $j$ picks an action, and the process repeats.  The game terminates if a goal state in $V_f$ is reached.
For the remainder of the paper, all definitions are provided with respect to the game $\G$. For brevity, we omit explicitly restating this context. 

The players choose actions according to a strategy. Formally, 

\begin{definition}[Strategy]
    A \emph{strategy} $\sigma$ ($\tau$) is a function that maps a finite sequence of states to a Sys (Env) action, 
    such that $\sigma: V^* \cdot V_s \to A_s$ and $\tau: V^* \cdot V_e \to A_e$, where $\cdot$ is the concatenation operator.
    We denote $\Sigma$ and $\Tau$ as the set of all strategies for the Sys and Env player, respectively.
    A strategy is called \emph{memoryless} or \emph{positional} if it only depends on the last state in the sequence.
    \label{def: str}
\end{definition}

Given strategies $\sigma$ and $\tau$, a unique sequence of states, called \emph{play} and denoted by $\play^{v_0}(\sigma, \tau)$, is induced from $v_0$. Note that the play is unique for terminating plays only. A play can be finite $\play^{v_0}(\sigma, \tau) := v_0 v_1 \dots v_n \in V^* $
or infinite $\play^{v_0}(\sigma, \tau) := v_0 v_1 \dots \in V^\omega$.
We denote by $\plays^{v} := \{P^v(\sigma, \tau) \mid \sigma \in \Sigma, \; \tau \in \Tau \}$ the set of plays starting from $v$ under every Sys and Env strategy. $\plays^{v}(\sigma)$ is the set of plays induced by a fixed strategy $\sigma$ and every Env strategy. We note that a finite play occurs iff a goal state is reached.

A finite prefix of a play is called the history $h$. We define $|h|$ the length of the history and $h_{j}$ for $0\leq j \leq |h|-1$ as the $(j+1)^{th}$ state in the sequence. The last vertex of a history $h$ is defined as $\last(h):= h_{|h| - 1}$. We denote the set of plays with common prefix $h$ as $\plays^{h} := \{\play^h = h \cdot \play \mid  \play \in \plays^{v}(\sigma,\tau), \sigma \in \Sigma, \tau \in \Tau, v = \delta(\last(h),\sigma(h)) \text{ if } v\in V_s, \text{ else } v = \delta(\last(h),\tau(h))  \}$.

In a qualitative reachability game, the objective of the Sys player is to choose $\sigma$ such that every play in $\plays^{v_0}(\sigma)$ reaches a state in $V_f$. In a quantitative reachability game, the Sys player has an additional objective of minimizing the total cost of its actions along the play, called the payoff. 
\begin{definition}[Total Payoff]
    Given strategies $\sigma \in \Sigma$ and $\tau \in \Tau$, \emph{total payoff} is defined as the sum of all the action costs given by $C$ along the induced play $ \play^{v_0}(\sigma, \tau) = v_0 v_1 \ldots v_n$ where $n \in \mathbb{N} \cup \{\infty\}$, i.e., 
    \begin{equation}
        \Val(\play^{v_0}(\sigma, \tau)) := \sum_{i=0}^{n-1} C(v_i, a_i),
    \end{equation}
    where $a_i = \sigma(v_0 \ldots v_{i})$ if  $v_{i} \in V_s$, else $a_i = \tau(v_0 \ldots v_{i})$.
\end{definition}

\noindent
Note that, for a play with infinite length $|\play^{v_0}(\sigma, \tau)| = \infty$,  $\Val(\play^{v_0}(\sigma, \tau)) = \infty$. 
We now define two notions of payoff for a game $\G$ that formalizes best-case and worst-case scenarios for the Sys player with respect to $\Val$.

\begin{definition}[Cooperative \& Adversarial Values]
    \label{def: cval_def}
    Given $h$, Sys ($\sigma$) and Env ($\tau$) strategies compatible with $h$, let $\play^{h}(\sigma, \tau)$ denote a play that extends history $h$.
    The \emph{cooperative value} $\cVal(h, \sigma)$ is the payoff of the play
    such that the Env player plays minimally, i.e., 
    \begin{equation}
        \cVal(h, \sigma) = \inf_{\tau \in \Tau}\, \Val(\play^{h}(\sigma, \tau)).
    \label{eq: cval}
    \end{equation}
    Similarly, the \emph{adversarial value} $\aVal(h, \sigma)$ is the payoff 
    where the Env player plays maximally, i.e.,
    \begin{equation}
        \aVal(h, \sigma) = \sup_{\tau \in \Tau}\,  \Val(\play^{h}(\sigma, \tau)).
    \label{eq: aval}
    \end{equation} 
    We denote by $\cVal(h)$ and $\aVal(h)$ the optimal cooperative and adversarial values for history $h$, respectively, i.e.,
    \begin{equation*}
        \cVal(h) = \inf_{\sigma \in \Sigma} \cVal(h, \sigma) \; \: \text{and} \; \: \aVal(h) = \inf_{\sigma \in \Sigma} \aVal(h, \sigma).
    \end{equation*}
\end{definition}

For a given $h$, we say Sys strategy $\sigma_{win} \in \Sigma$ is \emph{winning} 
if $\sigma_{win}$ is compatible with $h$ and $\aVal(h, \sigma_{win}) < \infty$.
If $\sigma_{win}$ exists, the Sys player can force a visit to $V_f$ for all Env strategies, guaranteeing its reachability objective.
In classical quantitative reactive synthesis, the interest is in a $\sigma_{win}$ that achieves the optimal $\aVal(v_0)$.

\subsection{Beyond Winning Strategies}

In many applications, when a winning strategy does not exist, it is desirable for the Sys player to adopt a strategy that ensures the best possible outcome. For game $\G$ in Fig.~\ref{fig: game_init_ex}, $\sigma_{win}$ does not exist but 
it is better for the Sys player to keep trying, i.e., move to $v_1$ or $v_3$ rather than giving up. To capture this intuition, we turn to the classical notion of dominance
\cite{Leyton-Brown2008}.

\begin{figure}[t!]
    \centering
    \resizebox{\linewidth}{!}{%
        \centering
        \begin{tikzpicture}[->,>=stealth',shorten >=1pt,shorten <=1pt,auto,node distance=1cm,
            every loop/.style={looseness=6},
            initial/.style = {initial below},
            initial text={},
            el/.style={font=\scriptsize},
            every fit/.style={draw,densely dotted,rectangle},
            inner sep=2mm,
            loopright/.style={loop,looseness=6,out=-45, in=45},
            loopleft/.style={loop,looseness=6,out=135, in=225},
            loopabove/.style={loop,looseness=6,out=45, in=135},
            loopbelow/.style={loop,looseness=6,out=-135, in=-45},]
        \tikzstyle{every state}=[node distance=1.4cm,minimum size=7mm, inner sep=1pt] ; 
        \node[state, initial] (v0) {$v_0$};
        \node[state, rectangle, right of=v0] (v1) {$v_1$};
        \node[state, rectangle, left of=v0] (v2) {$v_2$};
        \node[state, left of=v2] (v3) {$v_3$};
        \node[state, right of=v1] (v4) {$v_4$};
        \node[state, rectangle, right of=v4] (v5) {$v_5$};
        \node[state, accepting, right of=v5] (v6) {$v_6$};
        \node[state, rectangle, below of=v5] (v7) {$v_7$};
        \node[state, right of=v7] (v8) {$v_8$};
        \node[state, right of=v8] (v9) {$v_9$};
        \node[state, rectangle, right of=v6] (v10) {$v_{10}$};
        \node[state, node distance=1.4cm, below of=v3] (text1) {Sys};
        \node[state, rectangle, right of=text1] {Env};
        \path[->]
            (v0) edge node {1} (v1)
            (v1) edge node {0} (v4)
            (v4) edge[thick] node {9} (v5)
            (v4) edge[very thick] node[left] {1} (v7)
            (v3) edge[bend left] node[above] {1} (v2)
            (v0) edge[ultra thick] node {1} (v2)
            (v5) edge node[below] {0} (v6)
            (v2) edge[bend left] node {0} (v3)
            (v2) edge[bend left] node {0} (v6)
            (v6) edge[loop] node[above] {0} (v6)
            (v7) edge node[above] {0} (v8)
            (v7) edge[bend right] node[below] {0} (v9)
            (v8) edge node[left] {8} (v10)
            (v9) edge node[right] {1} (v10)
            (v10) edge node[above] {0} (v6);
        
            \begin{scope}[shift={(0,-3.5)}, font=\scriptsize] %
                \draw[thick] (1,2.3) -- (1.75,2.3) node[right] {$\sigma_1$};
                \draw[very thick] (1,2) -- (1.75,2) node[right] {$\sigma_2$};
                \draw[ultra thick] (1,1.7) -- (1.75,1.7) node[right] {$\sigma_3$};
            \end{scope}
        
        \end{tikzpicture}
        }
    \caption{Illustrative game. $v_0$ is initial and 
    $v_6$ is goal state. 
    }
    \label{fig: cex_adm_val_not_preserving}
\end{figure}

\begin{definition}[Dominance \cite{brenguier2016admissibility,berwanger2007admissibility}]
    Given two Sys strategies $\sigma,\sigma' \in \Sigma$ and initial state $v \in V$, we say 
    \begin{itemize}
        \item $\sigma$ \emph{very weakly dominates} $\sigma'$, denoted by $\sigma \succeq \sigma'$, if $\sigma$ does at least as well as $\sigma'$:
            \begin{equation}
                \label{eq: w_dom_str_eq1}
               \Val(\play^v(\sigma,\tau)) \leq \Val(\play^v(\sigma', \tau)) \quad \ \forall \tau \in \Tau.
            \end{equation}
        \item $\sigma$ \emph{weakly dominates} $\sigma'$, denoted by $\sigma \succ \sigma'$, if $\sigma \succeq \sigma'$ and $\sigma$ sometimes does better than $\sigma':$
            \begin{equation}
                  \Val(\play^v(\sigma,\tau)) < \Val(\play^v(\sigma', \tau)) \quad \ \exists \tau \in \Tau.
                \label{eq: w_dom_str_eq2}
            \end{equation}
    \end{itemize}
    \label{def: dom_str}
\end{definition}

Dominance induces a partial order on Sys  strategies, whose maximal elements are called admissible strategies. Strategies that are not dominated are always admissible.

\begin{definition}[Admissible Strategy]
    \label{def: quant_admissible}
    A strategy $\sigma$ is called \emph{admissible} 
    if it is not \emph{weakly dominated} by any other Sys player strategy, i.e., $\nexists \sigma' \in \Sigma$ \ s.t. \ $ \sigma' \succ \sigma$.
\end{definition}

\begin{example}
    In the game in Fig.~\ref{fig: cex_adm_val_not_preserving},
    let $\sigma_1: (v_0 \to v_1), (v_4 \to v_5)$, $\sigma_2: (v_0 \to v_1), (v_4 \to v_7)$, and $\sigma_3: (v_0 \to v_2), (v_3 \to v_2)$. 
    As $\sigma_2 \succ \sigma_1 $, $\sigma_1$ is not admissible. But, $\sigma_3$ is not weakly dominated by $\sigma_2$ as there exists a play under which $\sigma_3$ does strictly better than $\sigma_2$ for every play in $\plays^{v_0}(\sigma_2)$. Hence, $\sigma_3$ and $\sigma_2$ are both admissible. This example demonstrates that (i) admissible strategies can be overly optimistic, and (ii) they differ from winning strategies.
    \label{ex: adm_qual_qaunt_diff}
\end{example} 

\citeauthor{brandenburger2008admissibility} rationalize a player playing admissibly to be doing their \emph{best}. Thus, the first problem we consider is the synthesis of admissible strategies.

\begin{problem}
    \label{prob: problem_1}
    Given a 2-player quantitative reachability game $\G$ and energy budget $\B \in \mathbb{N^{+}}$, synthesize the set of all admissible strategies $\Sigma_{adm}$ such that, for every $\sigma \in \Sigma_{adm}$, there exists $\tau \in \Tau$ under which  $\Val(\play^{v_0}(\sigma, \tau)) \leq \B$. 
\end{problem}

Intuitively, $\sigma \in \Sigma_{adm}$ allows the Sys player to do its best without any assumptions about the Env player. 
In qualitative settings, winning strategies are always admissible \cite[Lemma 8]{brenguier2015assume}. 
In quantitative settings, however, winning strategies are not necessarily admissible (e.g., $\sigma_1$ in Fig. \ref{fig: cex_adm_val_not_preserving}).
Thus, to compute enforceable admissible strategies, we introduce the notion of an admissible winning strategy.

\begin{definition}[Admissible Winning Strategy]
\label{def: adm_win}
Strategy $\sigma$ is called \emph{admissible winning} for Sys player iff it is admissible and $\forall h \in \plays^{v_0}(\sigma)$ if $\aVal(h) < \infty$ then $\aVal(h, \sigma) < \infty$.
\end{definition}

For the game in
Fig.~\ref{fig: cex_adm_val_not_preserving},
$\sigma_3$ is admissible but not enforcing as it cannot ensure always reaching the goal state. $\sigma_1$ is winning but not admissible. $\sigma_2$ is admissible winning, which is more desirable than the other two.  
In this work, we also consider the synthesis problem of such strategies.

\begin{problem}
    \label{prob: problem_2}
    Given a 2-player quantitative reachability game $\G$ and energy budget $\B \in \mathbb{N}^+$, synthesize the set of all admissible winning strategies $\Sigma^{win}_{adm}$
    such that, for every $\sigma \in \Sigma^{win}_{adm}$ and $\forall h \in \plays^{v_0}(\sigma)$, if $\aVal(h) < \infty$ then $\aVal(h, \sigma) \leq \B$. 
\end{problem}

In Section \ref{sec: adm_str}, we show how to solve Problem \ref{prob: problem_1} by providing necessary and sufficient conditions for a strategy to be admissible. In Section \ref{sec: adm_win_str}, we identify the class of admissible strategies that are admissible winning and give an algorithm to solve Problem \ref{prob: problem_2}. Due to space constraints, the proofs of all theoretical claims are provided in Appendix.

\section{Admissible Strategies}
\label{sec: adm_str}

To have a sound and complete algorithm for the synthesis of $\Sigma_{adm}$, we need to first understand the characteristics and properties of admissible strategies. Prior works in qualitative reachability games show that synthesis can be reduced to strategies with \emph{value-preserving} property (defined below). We show that in our quantitative setting, this property does \emph{not} hold. Thus, we investigate the appropriate conditions that characterize admissible strategies. We identify two classes of strategies that are not only sufficient but also necessary for a strategy to be admissible. Finally, we show how memoryless strategies are not sufficient for admissibility and provide a synthesis algorithm. 

\subsection{Admissible Strategies are not Value-Preserving}
\label{ssec: charc_regions}

The reachability objective in $\G$ naturally partitions the set of states $V$ into three subsets: the set of states from which the Sys player
(i)  can force a visit to $V_f$ under every Env strategy,
(ii) cannot reach $V_f$ under any Env strategy, and
(iii) may reach $V_f$  only under some Env strategies.
We can formalize these sets using $\cVal(v)$ and $\aVal(v)$: 
\begin{align*}
    &\emph{winning region:} \; V_{win} =\{v \in V \mid \aVal(v) < \infty\},\\
    &\emph{losing region:}  \; V_{los} =\{v \in V \mid \aVal(v) = \cVal(v) = \infty \}, \\
    &\emph{pending region:} \; V_{pen} =\{v \in V \mid \aVal(v) = \infty, \\
    & \hspace{6cm} \cVal(v) < \infty \}.
\end{align*}
Note $V_{win}$, $V_{los}$, and $V_{pen}$ define a partition for $V$, i.e., their union is $V$ and their pair-wise intersection is the empty set. Based on these sets, we characterize value-preserving strategies according to the notion of value for each state. Let $\sVal: V \to \{-1,0,1\}$ be a state-value function such that $\sVal(v) = 1$ if $v \in V_{win}$, $0$ if $v \in V_{pen}$, and $-1$ if $v \in V_{los}$. 

\begin{definition}[Value-Preserving]
    \label{def: val_preserving}
    We say history $h$ is value-preserving if $\sVal(h_j) \leq \sVal(h_{j+1})$ for all $0 \leq j < |h| - 1 $. Strategy $\sigma$ is value preserving if every $h \in \plays^{v_0}(\sigma)$ is value preserving. 
\end{definition}

Let us now look at two classical notions for strategies defined for quantitative games and discuss their value-preserving property. We say $\sigma$ is a \emph{worst-case optimal} strategy ($\wco$) if $\aVal(h, \sigma) = \aVal(h)$. If at the current state $v$, $\sVal(v) = 1$, then an optimal winning strategy $\sigma_{win}$ exists such that all plays in $\plays^{v}(\sigma_{win})$ are value-preserving. Since $\sigma_{win}$ is $\wco$, every $\wco$ strategy in $V_{win}$ is also value preserving. 
If $\sVal(v) \neq -1$, a \emph{cooperatively-optimal} ($\coop$) 
strategy $\sigma$ exists such that $\cVal(h, \sigma) = \cVal(h)$. 
Unlike $\wco$, $\coop$ strategies are not value-preserving.
In Fig.~\ref{fig: cex_adm_val_not_preserving},
$\sigma_1$ and $\sigma_2$ are $\wco$ as they ensure the lowest payoff of 10 if Env is adversarial while $\sigma_3$ is $\coop$ as the corresponding payoff of 1 is the lowest for a cooperative Env. Further, notice that $\sigma_3$, while admissible, is not value preserving. The following lemma formalizes this observation.

\begin{lemma}
     Admissible strategies are \emph{not} always value preserving.
    \label{lem: quant_adm_not_val_pres}
\end{lemma}
\ifproof
\begin{proof}
    For the proof, it is sufficient to show a counterexample.
    In Example \ref{ex: adm_qual_qaunt_diff}, for $\G$ in Figure \ref{fig: cex_adm_val_not_preserving}, all states except for $v_2$ and $v_3$ belong to $V_{win}$ and hence $\sVal(v) = 1$. States $v_2$ and $v_3$ belong to $V_{pen}$ thus $\sVal(v_2) = \sVal(v_3) = 0$. Both, $\sigma_2$ and $\sigma_3$ are admissible strategies as there exists a play that does strictly better than all plays under strategy $\sigma_1$. But, under $\sigma_3$, the possible plays are, $v_0(v_2 v_1)^{\omega}$, $v_0(v_2 v_3)^* v_6$, and $v_0 v_2 v_6$. Thus, for $\sigma_3$, there exists a play that starts in the winning region and does \emph{not} stay in the winning region. For plays induced by $\sigma_2$, all states in all the plays belong to the winning region. Hence, there exist admissible strategies that are not value preserving.
\end{proof}
\fi

Unlike the qualitative setting \cite{faella2009admissible,aminof2020synthesizing}, Lemma~\ref{lem: quant_adm_not_val_pres} shows that we cannot characterize admissible strategies solely on the basis of $\wco$ strategies as they are not value-preserving in our quantitative setting. 
Below, 
we derive two new categories of strategies from 
$\wco$ and $\coop$ that are always admissible. We then discuss their properties and show that they are also necessary conditions for admissibility. 
Table~\ref{tab: str_summary} 
summarizes properties of all new strategies we define hereafter.

\begin{table}[t]
    \centering
    \resizebox{1\linewidth}{!}{%
    \begin{tabular}{l c c c c c}
        \toprule
         & \multirow{2}{*}{$\scoop$}  & \multirow{2}{*}{$\wcoop$}  & \multirow{2}{*}{$\modscoop$} & \multirow{2}{*}{Adm.} & Adm. \\ 
         & & & & & winning \\
         \midrule
         Value-preserving & \xmark & \cmark & \cmark & \xmark & \cmark \\
         Winning  & \xmark  & \cmark & \cmark & \xmark & \cmark \\
         Memoryless & \xmark & \cmark & \xmark & \xmark & \xmark \\
         Algorithm & Sec. \ref{sec: adm_str} & Sec. \ref{sec: adm_str} & Sec. \ref{sec: adm_win_str} & Sec. \ref{sec: adm_str} & Sec. \ref{sec: adm_win_str}\\
        \bottomrule
    \end{tabular}
    }
    \caption{Properties of new strategies defined for admissibility.}
    \label{tab: str_summary}
\end{table}

\subsection{Characterization of Admissible Strategies}

Note that, for every history $h$, the strategies that are cooperative optimal have the least payoff.
Thus, every $\sigma$ that is $\coop$ is admissible as there does not exist $\sigma'$ that weakly dominates it. We now define \emph{strongly cooperative} condition ($\scoop$) which generalizes $\coop$. Intuitively, strategies that are $\scoop$ have a lower payoff than the worst-case optimal payoff at $h$. In case, a lower payoff cannot be obtained, $\scoop$ are worst-case optimal. 

\begin{definition}[$\scoop$]
    \label{def: scoop}
    Strategy $\sigma$ is \emph{Strongly Cooperative} $(\scoop)$ if for every $h \in \plays^{v_0}(\sigma)$
    one of the following two conditions holds:
    (i) if $\cVal(h) < \aVal(h)$ then
    $\cVal(h, \sigma) < \aVal(h)$, or (ii) if $\cVal(h) = \aVal(h)$ then $\aVal(h, \sigma) = \cVal(h, \sigma) = \aVal(h)$.
\end{definition}
\noindent
In the game in 
Fig.~\ref{fig: cex_adm_val_not_preserving},
both $\sigma_2$ and $\sigma_3$ are $\scoop$ strategies.

Let $\sigma'$ be a strategy that is not $\scoop$. If $\cVal(h, \sigma') > \aVal(h)$ then $\sigma'$ always has a payoff worse than a $\wco$ strategy. If $\cVal(h, \sigma') = \aVal(h)$ then $\sigma'$ does as well as a $\wco$ strategy but never better. Thus, $\sigma'$ does not weakly dominate a $\scoop$ strategy, resulting in the following lemma.
\begin{lemma}
    All $\scoop$ strategies are admissible.
    \label{lem: scoop_proof}
\end{lemma}
\ifproof
\begin{proof}
    \begin{figure}[h]
    \centering
    \resizebox{0.45\linewidth}{!}{%
        \centering
        \begin{tikzpicture}
            [->,>=stealth',shorten >=1pt,shorten <=1pt,auto,node distance=1cm,
            every loop/.style={looseness=6},
            initial text={},
            el/.style={font=\scriptsize},
            every fit/.style={draw,densely dotted,rectangle},
            inner sep=2mm,
            decoration = {snake,pre length=3pt,post length=7pt},
            loopright/.style={loop,looseness=6,out=-45, in=45},
            loopleft/.style={loop,looseness=6,out=135, in=225},
            loopabove/.style={loop,looseness=6,out=45, in=135},
            loopbelow/.style={loop,looseness=6,out=-135, in=-45},]
        \tikzstyle{every state}=[node distance=1.4cm,minimum size=7mm, inner sep=1pt];
        \node[state] at (0,0) (v0){$v_s$};
        \node[state, left of=v0, node distance=2.4cm,] (v4) {$v_0$};
        \node[state, rectangle, node distance=2.4cm, above right of=v0] (v1){$\sigma(h)$};
        \node[state, rectangle, node distance=2.4cm, below right of=v0] (v2) {$\sigma'(h)$};
        \path[-latex'] 
          (v4) edge[decorate] node[above] {$h$} (v0)
          (v0) edge[decorate] node {$\sigma$ is $\scoop$}(v1)
          (v0) edge[decorate] node[left] {$\sigma'$ is not $\scoop$}(v2);
        \end{tikzpicture}
        }
    \caption{$\scoop$ proof example}
    \label{fig: sc_proof}
\end{figure}

    Let $\sigma$ be $\scoop$. Assume there exists $\sigma' \neq \sigma$ that is compatible with history $h$, $\last(h) = v_s$, and ``splits" at $v_s$ as shown in Figure \ref{fig: sc_proof}. Thus, $\sigma(h) \neq \sigma'(h)$. We note that only two cases for a strategy are possible, i.e., it is either $\scoop$ or not. 
    Further, let's assume that $\sigma'$ weakly dominates $\sigma$. We show that $\sigma'$ can not weakly dominate $\sigma$. 

    \paragraph{Case I}$\cVal(h) < \aVal(h)$:  As $\sigma'$ is not $\scoop$ this implies $\cVal(h, \sigma') \geq \aVal(h)$ and as $\cVal(h, \sigma) \leq \aVal(h, \sigma)$ for any $\sigma \in \Sigma$, we get 
    $$\aVal(h, \sigma') \geq \cVal(h, \sigma') \geq \aVal(h) > \cVal(h, \sigma).$$

    On simplifying, we get $\aVal(h, \sigma') \geq \aVal(h) > \cVal(h, \sigma)$. This statement implies that there exists $\tau \in \Tau$ such that $\Val(h \cdot \play^{h}(\sigma', \tau)) > \Val(h \cdot \play^{h}(\sigma, \tau))$. Note that since, $\Val(h \cdot 
    \play^{h}(\sigma, \tau)) = \Val(h) + \Val(\play^{v_s}(\sigma, \tau))$ we get $\Val(\play^{v_s}(\sigma', \tau)) > \Val(\play^{v_s}(\sigma, \tau))$. This contradicts the assumption that $\sigma'$ dominates $\sigma$ as $\sigma'$ should always have a payoff that is equal to or lower than $\sigma$. 

    \paragraph{Case II}$\cVal(h) = \aVal(h)$: For this case we get,
    $$\aVal(h, \sigma') \geq \cVal(h, \sigma') \geq \aVal(h) = \cVal(h) = \cVal(h, \sigma).$$ 

    This implies that $\cVal(h, \sigma') \geq \cVal(h)$. But, since $\sigma'$ dominates $\sigma$, there should exist a strategy $\tau \in \Tau$ under which $\sigma'$ does strictly better than $\sigma$. Since, $\cVal(h, \sigma') \geq \cVal(h) \forall \tau \in \Tau$, it implies that there does not exists a payoff $\Val(\play^{v_s}(\sigma', \tau))$ that has a payoff strictly less than $\cVal(h, \sigma) = \cVal(h)$. Thus, $\sigma'$ does not dominate $\sigma$. 

     We can repeat this for all histories $h$ in $\plays^{v_0}(\sigma)$.
    Hence, every $\sigma$ that is $\scoop$ is admissible.
\end{proof}
\fi

 From Lemma \ref{lem: scoop_proof}, it suffices for us to show that $\scoop$ strategies always exist to prove that admissible strategies always exist. Unfortunately, $\scoop$ strategies are history-dependent, i.e., we need to reason over every state along a history to check for admissibility. This is formalized as follows.

\begin{theorem}
    Memoryless strategies are \emph{not} sufficient for $\scoop$ strategies.
    \label{thm: scoop_memory}
\end{theorem}
\ifproof
\begin{proof}
    Consider the game $\G$ in Figure \ref{fig: local_conds_not_sufficient}. Let us consider strategies: $\sigma_1:(v_0 \to v_1), (v_3 \to v_6), (v_7 \to v_8)$, $\sigma_2:(v_0 \to v_1), (v_3 \to v_6),(v_7 \to v_9)$, $\sigma_3: (v_0 \to v_2), (v_4 \to v_6),(v_7 \to v_9)$, and $\sigma_4: (v_0 \to v_2), (v_4 \to v_6),(v_7 \to v_8)$. Note, $\sigma_1, \sigma_2$ and $\sigma_3, \sigma_4$ only differ at state $v_7$. Further, $\G$ is a tree-like arena where we define payoff as follows, if $\last(h)$ is a leaf node, then $\Val(\last(h)) > 0$ else, it is 0. Thus, the payoff function is history-independent, i.e., $\Val(h) = \Val(\last(h))$.
    Hence, it is sufficient to look at the last state along $h$ to compute optimal adversarial and cooperative values, i.e., $\aVal(h) = \aVal(\last(h))$ and $\cVal(h) = \cVal(\last(h))$.
    
    For the play induced by $\sigma_1$, $\cVal(v_0, \sigma_1) < \aVal(v_0)$, $\cVal(v_0 v_1 v_3, \sigma_1) < \min\{\aVal(v_0), \aVal(v_0 v_1 v_3)\}, \ldots \cVal(v_0 v_1 v_3 v_6 v_7, \sigma_1) < \min\{\aVal(v_0), \aVal(v_0 v_1 v_3), \aVal(v_0 v_1 v_3 v_6 v_7)$\}. Thus, $\cVal(h, \sigma_1) < \aVal(h)$ for all the histories $h$ compatible with $\sigma_1$. For the play induced by $\sigma_2$,  $\cVal(v_0, \sigma_2) < \aVal(v_0)$, but $\cVal(v_0 v_1 v_3, \sigma_2) > \aVal(v_0 v_1 v_3)$. Notice that $\cVal(v_0 v_1 v_3 v_6, \sigma_2) < \aVal(v_0 v_1 v_3 v_6)$ and thus we need to check admissibility for all histories $h$ compatible with $\sigma$. Hence, for history $h := v_0 v_1 v_3$, $\sigma_2$ is not $\scoop$. If $h = v_0 v_2 v_4 v_6$, then strategy $\sigma_3$ and $\sigma_4$ are both admissible. Thus, $\sigma_3$ and $\sigma_4$ are both $\scoop$ strategies. This proves that $\scoop$ strategies are history-dependent.
\end{proof}
\fi
We now look at another interesting class of strategy that is always admissible. For every history $h$, $\wco$ strategies always guarantee the worst-case payoff. A notable subset of $\wco$ strategies are those that are also $\coop$.
We call such strategies \emph{Worst-case Cooperative Optimal} ($\wcoop$). 
\begin{definition}[$\wcoop$]
    \label{def: wcoop}
    Strategy $\sigma$ is \emph{Worst-case Cooperative Optimal} $(\wcoop)$ if, for all  $h \in \plays^{v_0}(\sigma),$
    $$\aVal(h, \sigma) = \aVal(h) \; \text{ and } \; \cVal(h, \sigma) = \acVal(h),$$
    where $\acVal(h) :=  \min \{\cVal(h, \sigma) \mid \sigma \in \Sigma, \aVal(h, \sigma) \leq \aVal(h)\}$ is the optimal adversarial-cooperative value of $h$.
\end{definition}
In this definition, $\acVal$ is a new notion that characterizes the minimum payoff that Sys player can obtain from the set of worst-case optimal strategies.
In the game in Fig.~\ref{fig: cex_adm_val_not_preserving}, only $\sigma_2$ is $\wcoop$ strategy. That is because action $v_4 \to v_7$ belongs to an admissible strategy as it has the minimum cooperative value while ensuring the worst-case optimal payoff. Here $\aVal(v_4, \sigma_2) = \aVal(v_4) = 9$; $\cVal(v_4, \sigma_2) = 2$. 

\ifproof
\begin{figure}[t!]
    \centering
    \resizebox{1.2\linewidth}{!}{%
        \centering
        \begin{tikzpicture}
            [->,>=stealth',shorten >=1pt,shorten <=1pt,auto,node distance=1cm,
            every loop/.style={looseness=6},
            initial text={},
            el/.style={font=\scriptsize},
            every fit/.style={draw,densely dotted,rectangle},
            inner sep=2mm,
            loopright/.style={loop,looseness=6,out=-45, in=45},
            loopleft/.style={loop,looseness=6,out=135, in=225},
            loopabove/.style={loop,looseness=6,out=45, in=135},
            loopbelow/.style={loop,looseness=6,out=-135, in=-45},]
        \tikzstyle{every state}=[node distance=1.4cm,minimum size=7mm, inner sep=1pt];
        \node[state, initial] at (0,0) (v0){$v_0$};
        \node[state, rectangle, right of=v0] (v1){$v_1$};
        \node[state, rectangle, below of=v1] (v2) {$v_2$};
        \node[state, node distance=1.4cm, right of=v1] (v3) {$v_3$};
        \node[state, right of=v2] (v4) {$v_4$};
        \node[state, rectangle, node distance=1.4cm, right of=v3] (v6) {$v_6$};
        \node[state, rectangle, node distance=1.4cm, above right of=v3] (v5) {$v_5$};
        \node[state, node distance=1.4cm, right of=v6] (v7) {$v_7$};
        \node[state, rectangle, node distance=1.4cm, right of=v7] (v9) {$v_9$};
        \node[state, rectangle, node distance=1.4cm, above right of=v7] (v8) {$v_8$};
        \node[state, rectangle, node distance=2.4cm, above of=v0] (v10) {$v_{10}$};
        \node[state, node distance=1.4cm, above right of=v10] (v11) {$v_{11}$};
        \node[state, node distance=1.4cm, below of=v2] (v12) {$v_{12}$};
        \node[state, node distance=1.4cm, above of=v1] (v13) {$v_{13}$};
        \node[state, accepting, node distance=1.4cm, inner sep=1pt, above left of=v5](val3) {};%
        \node[state, accepting, node distance=1.4cm, inner sep=1pt, above right of=v5](val4) {}; %
        \node[state, accepting, inner sep=1pt, above of=v8](val5) {};%
        \node[state, accepting, inner sep=1pt, right of=v8](val6) {};%
        \node[state, accepting, inner sep=1pt, below of=v9](val7) {};%
        \node[state, accepting, inner sep=1pt, right of=v9](val8) {};%
        \node[above left of=v0](lab0) {$(\textcolor{blue}{5}, \textcolor{red}{\infty})$};
        \node[node distance=0.6cm, below of=v1](lab1) {$(\textcolor{blue}{4}, \textcolor{red}{\infty})$};
        \node[left of=v2](lab2) {$(\textcolor{blue}{4}, \textcolor{red}{\infty})$};
         \node[above left of=v10](lab9) {$(\textcolor{blue}{\infty}, \textcolor{red}{\infty})$};
        \node[node distance=0.7cm, above left of=v3](lab3) {$(\textcolor{blue}{4}, \textcolor{red}{6})$};
        \node[below right of=v4](lab4) {$(\textcolor{blue}{4}, \textcolor{red}{11})$};
        \node[node distance=0.8cm, above of=v5](lab5) {$(\textcolor{blue}{4}, \textcolor{red}{5})$};
        \node[node distance=0.8cm, below of=v6](lab6) {$(\textcolor{blue}{3}, \textcolor{red}{10})$};
        \node[node distance=0.8cm, below of=v7](lab7) {$(\textcolor{blue}{3}, \textcolor{red}{10})$};
        \node[above left of=v8](lab8) {$(\textcolor{blue}{2}, \textcolor{red}{9})$};
        \node[node distance=0.8cm, below right of=v9](lab9) {$(\textcolor{blue}{7}, \textcolor{red}{11})$};
        \node[right of=v11](lab11) {$(\textcolor{blue}{\infty}, \textcolor{red}{\infty})$};
        \node[node distance=0.6cm, below of=v12](lab12) {$(\textcolor{blue}{\infty}, \textcolor{red}{\infty})$};
        \node[node distance=0.6cm, above of=v13](lab13) {$(\textcolor{blue}{\infty}, \textcolor{red}{\infty})$};
        
        \path[-latex'] (v0) edge node {1} (v1)
        (v0) edge node[above] {\xmark} (v10)
        (v0) edge node[below] {1} (v2)
        (v1) edge (v3)
        (v2) edge (v4)
        (v3) edge node {1} (v5)
        (v3) edge node {1} (v6)
        (v4) edge node[below] {1} (v6)
        (v6) edge (v7)
        (v7) edge node {1} (v8)
        (v7) edge[dashed] node[below] {?} node[above] {1} (v9)
        (v5) edge node[below left] {4} (val3)
        (v5) edge node[below right] {5} (val4)
        (v8) edge node[right] {2} (val5)
        (v8) edge node {9} (val6)
        (v9) edge node[left] {11} (val7)
        (v9) edge node {7} (val8)
        (v10) edge[bend left] (v11)
        (v11) edge[bend left] node {1} (v10)
        (v2) edge[bend left] (v12)
        (v12) edge[bend left] node {1} (v2)
        (v1) edge[bend left] (v13)
        (v13) edge[bend left] node[left] {1} (v1)
        (val3) edge[loop] (val3)
        (val4) edge[loop] (val4)
        (val5) edge[loop] (val5)
        (val6) edge[loop] (val6)
        (val7) edge[loopbelow] (val7)
        (val8) edge[loopright] (val8);
        \end{tikzpicture}
        }
    \caption{Illustrative example $\G$: for all $h \in \plays^{v_0}$ we define $\Val < \infty$ if a play reaches a goal (double circle around the state) state. The cost of the Env action is zero, while the cost of the Sys actions are shown along the edges.
    The values in blue and red are $\cVal(v)$ and $\aVal(v)$.
    }
    \label{fig: local_conds_not_sufficient}
\end{figure}
\fi

Let $\sigma'$ be a strategy that is not $\wcoop$ and $\wco$. If $\sigma$ is $\wco$, then the worst-case payoff of $\sigma'$ is greater than $\sigma$'s worst-case payoff. As there exists a play under $\sigma'$ that does strictly worse than all plays under $\sigma$, it cannot dominate $\sigma$. Thus, $\sigma$ is admissible. 

\begin{lemma}
    All $\wcoop$ strategies are admissible.
    \label{lem: wcoop_proof}
\end{lemma}
\ifproof
\begin{proof}
    \begin{figure}[h]
    \centering
    \resizebox{0.45\linewidth}{!}{%
        \centering
        \begin{tikzpicture}
            [->,>=stealth',shorten >=1pt,shorten <=1pt,auto,node distance=1cm,
            every loop/.style={looseness=6},
            initial text={},
            el/.style={font=\scriptsize},
            every fit/.style={draw,densely dotted,rectangle},
            inner sep=2mm,
            decoration = {snake,pre length=3pt,post length=7pt},
            loopright/.style={loop,looseness=6,out=-45, in=45},
            loopleft/.style={loop,looseness=6,out=135, in=225},
            loopabove/.style={loop,looseness=6,out=45, in=135},
            loopbelow/.style={loop,looseness=6,out=-135, in=-45},]
        \tikzstyle{every state}=[node distance=1.4cm,minimum size=7mm, inner sep=1pt];
        \node[state] at (0,0) (v0){$v_s$};
        \node[state, left of=v0, node distance=2.4cm,] (v4) {$v_0$};
        \node[state, rectangle, node distance=2.4cm, above right of=v0] (v1){$\sigma(h)$};
        \node[state, rectangle, node distance=2.4cm, below right of=v0] (v2) {$\sigma'(h)$};
        \path[-latex'] 
          (v4) edge[decorate] node[above] {$h$} (v0)
          (v0) edge[decorate] node {$\sigma$ is $\wcoop$}(v1)
          (v0) edge[decorate] node[left] {$\sigma'$ is not $\wcoop$}(v2);
        \end{tikzpicture}
        }
    \caption{$\wcoop$ proof example}
    \label{fig: wcoop_proof}
\end{figure}

    Let $\sigma$ be a $\wcoop$. Assume there exists $\sigma' \neq \sigma$ that is compatible with history $h$, $\last(h) = v_s$, and ``splits" at $v_s$ as shown in Figure \ref{fig: wcoop_proof}. Thus, $\sigma(h) \neq \sigma'(h)$. We note that only two cases for a strategy are possible, i.e., it is either $\wcoop$ or not.
    Further, let's say that $\sigma'$ weakly dominates $\sigma$. We show that $\sigma'$ does not weakly dominate $\sigma$.

    We first note that, by definition, we have $\aVal(h, \sigma) \geq \aVal(h)$ for all $\sigma \in \Sigma$. If $\aVal(h, \sigma) = \aVal(h)$ then it is $\wco$ else it is not. For strategy $\sigma'$ compatible with $h$, $\neg \wcoop \implies \left( \aVal(h, \sigma') > \aVal(h) \vee \cVal(h, \sigma') \neq \acVal(h) \right)$.
    
    \paragraph{Case I}$\sigma'$ is not $\wco$: This implies that $\aVal(h) \neq \aVal(h, \sigma')$. As $\sigma'$ is not $\wco$, it implies that there exists an adversarial strategy $\tau \in \Tau$ for which the payoff $\Val(h \cdot \play^{v_s}(\sigma', \tau)) > \Val(h \cdot \play^{v_s}(\sigma, \tau))$. Since, $\Val(h \cdot \play^{v_s}(\sigma, \tau)) = \Val(h) + \Val(\play^{v_s}(\sigma, \tau))$, we get $\Val(\play^{v_s}(\sigma', \tau)) > \Val(\play^{v_s}(\sigma, \tau))$. This contradicts our assumption that $\sigma' \succ \sigma$ as $\sigma'$ should never do worse than $\sigma$.
    
    \paragraph{Case II}$\sigma'$ is $\wco$ but $\cVal(h, \sigma') \neq \acVal(h)$: For this case we have $\aVal(h) =\aVal(h, \sigma') \geq \cVal(h, \sigma') \wedge \cVal(h, \sigma') \neq \acVal(h)$. This implies that there exists a play whose payoff $\Val(\play^{v_s}(\sigma, \tau)) < \Val(\play^{v_s}(\sigma', \tau))$. This contradicts our statement as $\sigma \succ \sigma'$. 
    
    Thus, every strategy that is $\wcoop$ is admissible. 
\end{proof}
\fi

From Lemma \ref{lem: wcoop_proof}, it suffices for us to show that $\wcoop$ strategies always exist for admissible strategies to always exist. Interestingly, unlike prior work \cite{brenguier2016admissibility}, in our case, 
$\wcoop$ strategies always exist and at least one is memoryless.
The latter is desirable 
because
memoryless strategies can be computed efficiently using
fixed-point-based algorithms \cite{baier2008principles}.

\begin{theorem}
    $\wcoop$ strategies always exist and at least one is memoryless.
    \label{thm: wcoop_memoryless}
\end{theorem}
\ifproof
\begin{proof}
    We begin by showing that memoryless strategies are sufficient for optimal $\cVal$ and $\aVal$.
    For $\coop$ strategies, the game can be viewed as a single-player as both players are playing cooperatively. Thus, synthesizing a witnessing strategy $\sigma$ for $\coop$ reduces to the classical shortest path problem in a weighted graph. This can be solved in polynomial time using Dijkstra's and Flyod Warshall's algorithms when the weights are non-negative and arbitrary \cite{floyd1962algorithm, Mehlhorn2008algo}. 
    
    In the adversarial setting,   \citeauthor{khachiyan2008short} show that memoryless strategies are sufficient for two-player, non-negative weights scenario. Thus, we have that $\cVal(h, \sigma) = \cVal(\last(h), \sigma)$ for $\coop$ and $\aVal(h, \sigma) = \aVal(\last(h), \sigma)$ for $\wco$. 
    
    Next, we show that for a given history $h$, memoryless strategies are sufficient to be optimal adversarial-cooperative, i.e., it satisfies the $\cVal(\last(h), \sigma) = \acVal(h)$.  Given history $h$, we define $\aVal(\G, h)$ and $\cVal(\G, h)$ to be the optional adversarial and cooperative value in $\G$. We extend the definition to $\acVal(\G, h)$ accordingly. We denote by $\Sigma(\G)$ and $\Tau(\G)$   the set of all valid strategies in $\G$.
    
    Given $\G$, history $h$, we define $\G'$ to be the subgame such that the initial state in $\G'$ is $\last(h)$ and every state in $\G'$ satisfies condition $\aVal(\G, v) \leq \aVal(\G, h) = \aVal(\G, \last(h)) \; \forall v \in V$. We note that $\aVal(\G, v_f) = \cVal(\G, v_f) = 0$ and $\aVal(\G, v) \geq \cVal(\G, v) \neq 0 \; \forall v \notin V_f$. Thus, every subgame $\G'$ will include the goal states $v \in V_f$.
    As $\G'$ is the subgame of $\G$ the strategies in $\G'$ can be uniquely mapped to $\G$. Notice that $\Sigma(\G') \subseteq \Sigma(\G)$ where $\Sigma(\G')$ is the set of all valid  strategies in $\G'$. The weights of the edges in $\G'$ will be the same as $\G$.
    
    First, we will show that every state that is reachable in $\G'$ has at least one outgoing edge. Next, we will show that the .cooperative value of $\G'$ at the initial state $\last(h)$ is exactly the optimal adversarial-cooperative value of $h$ in $\G$, i.e., 
    $\cVal(\G', \last(h)) = \acVal(\G', \last(h)) = \acVal(\G, h)$. As memoryless strategies are sufficient for optimal cooperative value and strategies in $\G'$ can be uniquely mapped to $\G$, we will conclude that memoryless strategies are sufficient to be optimal $\wcoop$ strategy.

    \paragraph{Every reachable vertex in $\G'$ has at least one outgoing edge:} 
    By Thm. \ref{thm: adm_exist}, given $\G$, history $h$, there always exists a witnessing strategy $\sigma$ such that $\aVal(\G, h, \sigma) \leq \aVal(\G, h)$. By construction, we have that $\aVal(\G', v) \leq \aVal(\G, h)$. If $v$ is a Sys player state in $\G'$, then there always exists an action that corresponds to $\sigma$ such that $\sigma(v)$ is a valid edge in $\G'$. If $v$ is a state $\G'$ and belongs to $V_e$ then, by definition, for all the valid actions $a_e$,  $\aVal(\G', v') \leq \aVal(\G', v)$ where $v' = \delta(v, a_e)$. Thus, all edges from $v \in V_e$ will be present in $\G'$. Thus, any vertex reachable from $\last(h)$ has at least one outgoing edge in $\G'$.
    \paragraph{For any history $h$, $\cVal(\G',\last(h)) = \acVal(\G, h):$} We first observe that any state whose $\aVal(\G, v) > \aVal(\G, h)$, will not be present in $\G'$. Thus, for any prefix $h'$ of the plays starting from $\last(h)$ in $\G'$, will have their $\aVal(\G', \last(h')) \leq \aVal(\G, h)$. Further, all the plays that satisfy the $\acVal(\G, h)$ condition will be captured in $\G'$ too. Thus, $\acVal(\G', \last(h)) = \acVal(\G, h)$.
    Now, consider alternate strategy $\sigma' \neq \sigma$ where $\sigma', \sigma \in \Sigma(\G')$. We note that both $\sigma$ and $\sigma'$ are valid strategies in $\G$ too.
    Here $\sigma$ is a strategy such that $\aVal(\G', v, \sigma) \leq \aVal(\G, h)$.
    $\sigma'$ can be $\coop$ optimal in $\G$, i.e., $\cVal(\G, h, \sigma') < \acVal(\G, h)$, but $\aVal(\G, h, \sigma') > \aVal(\G, h)$ by definition of adversarial-cooperative value. As, $\aVal(\G, h, \sigma') > \aVal(\G, h)$, the state $\sigma'(h)$ will not be present in $\G'$. Inductively, for every history $h'$, starting from $\last(h)$ in $\G'$, we have that $\sigma'(h')$ will not be present in $\G'$. As every state that belongs $\sigma'(h')$ does not exists in $\G'$, $\cVal(\G', \last(h')) \not<  \acVal(\G, h)$. Further, $\cVal(\G', \last(h')) \not>  \acVal(\G, h)$ as $\cVal(\G', \last(h')) > \acVal(\G, h) \implies \aVal(\G', \last(h')) > \acVal(\G, h)$. This contradicts $\G'$ by construction. Thus, $\cVal(\G',\last(h)) = \acVal(\G, h)$. Since, $\acVal(\G', \last(h)) = \acVal(\G, h)$ we get $\cVal(\G',\last(h)) = \acVal(\G', \last(h)) = \acVal(\G, h)$. As memoryless strategies are sufficient for witnessing optimal cooperative value, this implies that memoryless strategies are sufficient for a strategy to be adversarial-cooperative optimal.
    
    Since memoryless strategies are sufficient for $\wco$ conditions and adversarial-cooperative optimal conditions, this implies that memoryless strategies are sufficient for $\wcoop$ strategies.
\end{proof}
\fi
The proof relies on the fact that memoryless strategies are sufficient for $\aVal$ and $\cVal$, and proceeds by constructing a subgame $\bar\G$ for a given history $h$ and showing that the cooperative value of the initial state in $\bar\G$ equals the $\acVal(h)$ of $\G$.
A consequence of Theorem~\ref{thm: wcoop_memoryless} is that a subset of admissible strategies, precisely $\wcoop$ strategies, are history-independent even for a payoff that is history-dependent. Thus, for $\sigma$ to be admissible, it is sufficient to be $\scoop$ or $\wcoop$, i.e.,
\begin{subequations}
    \begin{align}
    & \big (\cVal(h, \sigma) < \aVal(h) \big) \; 
    \vee \label{eq: adm_eq_1}\\
    & \big(\aVal(h) = \aVal(h, \sigma) \wedge \cVal(h, \sigma) = \acVal(h) \big) \label{eq: adm_eq_2}
\end{align}
\label{eq: inter_adm_eq}
\end{subequations}
\noindent
As shown below (Theorem~\ref{thm: adm_str}),
$\scoop$ and $\wcoop$ 
are also necessary conditions for admissibility. This becomes useful for synthesizing the set of all admissible strategies.

\subsection{Existence of Admissible Strategies} 

By simplifying Eq.~\eqref{eq: inter_adm_eq}, we get the following theorem.
\begin{theorem}
    \label{thm: adm_str}
    A strategy $\sigma$ is admissible if, and only if, $\forall h \in \plays^{v_0}(\sigma)$ with $\last(h) \in V_s$, the following holds
    \begin{subequations}
        \label{eq: adm_eq}
        \begin{align}
            & \big(\cVal(h, \sigma) < \aVal(h) \big) \;\; \vee \\
            & \big( \aVal(h) = \aVal(h, \sigma) = \cVal(h, \sigma) = \acVal(h) \big).
        \end{align}
    \end{subequations}
\end{theorem}
The proof uses Lemma \ref{lem: scoop_proof} and \ref{lem: wcoop_proof} for the sufficient conditions, and for the necessary conditions, it shows that $\neg \text{Eq.} \eqref{eq: adm_eq} \implies \sigma \notin \Sigma_{adm}$. We now establish the existence of admissible strategies. 

\ifproof
\begin{proof}
    We need to prove that $\neg$\eqref{eq: adm_eq} $\equiv$ \eqref{eq: dom_eq}. To prove this, we prove that $\neg$\eqref{eq: dom_eq} $\equiv$ \eqref{eq: adm_eq}.

    In Eq. \eqref{eq: dom_eq}, let us denote the term before disjunction as Term I and the latter as Term II. Thus, $\neg(\text{I} \vee \text{II}) = \neg \text{I} \wedge \neg \text{II}.$
    \begin{equation*}
        \neg \text{I} : \cVal(h, \sigma) < \aVal(h) \vee \aVal(h, \sigma) \leq \aVal(h)
    \end{equation*}

    By rewriting the equation, we get 
    \begin{multline*}
    \cVal(h, \sigma) < \aVal(h) \\
    \vee \bigl( \cVal(h, \sigma) \geq \aVal(h) \wedge \aVal(h, \sigma) \leq \aVal(h) \bigr)     
    \end{multline*}

    Rearranging the term inside the bracket, we get, $\aVal(h, \sigma) \leq \aVal(h) \leq \cVal(h, \sigma)$. Since, $\aVal(h, \sigma) \geq \cVal(h, \sigma)$ for any $h$, we get $\aVal(h, \sigma) = \aVal(h) = \cVal(h, \sigma)$. Thus, $\neg(\text{I})$ is
    \begin{equation}
          \cVal(h, \sigma) < \aVal(h) \vee \bigl(\aVal(h, \sigma) = \aVal(h) = \cVal(h, \sigma)\bigr)
          \label{eq: inter_step_1}
    \end{equation}

    Now, let us take the negation of term II after the disjunction. Thus,
    \begin{multline*}
        \neg \left[ \aVal(h) = \aVal(h, \sigma) = \cVal(h, \sigma) \right] \\
        \vee \neg \left[\acVal(h) < \aVal(h) \right]
    \end{multline*}

    On simplifying, $\neg(\text{II})$ is
    \begin{align}
        \left[\aVal(h) = \aVal(h, \sigma) = \cVal(h, \sigma) \right] & \notag \\
        \implies & \neg \left[\acVal(h) < \aVal(h) \right] \notag \\
        \aVal(h) = \aVal(h, \sigma) = \cVal(h, \sigma) & \implies \acVal(h) \geq \aVal(h)
        \label{eq: inter_step_2}
    \end{align}

    By substituting values for $\neg \text{I} \wedge \neg \text{II}$, we get
    \begin{align*}
         \cVal(h, \sigma) < \aVal(h) \vee
         \bigl(\aVal(h, \sigma) = \aVal(h) = \cVal(h, \sigma)\bigr) \\
         \wedge \aVal(h) = \aVal(h, \sigma) = \cVal(h, \sigma) \implies \acVal(h) \geq \aVal(h)
    \end{align*}
    \begin{multline*}
        \cVal(h, \sigma) < \aVal(h) \vee \biggl( \aVal(h, \sigma) = \aVal(h) = \cVal(h, \sigma) \\
        \wedge \; \aVal(h) = \aVal(h, \sigma) = \cVal(h, \sigma) \\ \implies \acVal(h) \geq \aVal(h) \biggr)
    \end{multline*}
    On simplifying, 
    \begin{multline*}
        \cVal(h, \sigma) < \aVal(h) \vee \biggl(\aVal(h) = \aVal(h, \sigma) = \cVal(h, \sigma) \\
        \implies \acVal(h) \geq \aVal(h) \biggr)
    \end{multline*}
    By definition, $\acVal(h) \leq \aVal(h)$. Thus, we can simplify the equation to be 
    \begin{align*}
        \cVal(h, \sigma) & < \aVal(h) \\
        \vee \; \aVal(h) = \aVal(h, \sigma) & = \cVal(h, \sigma) = \acVal(h) 
    \end{align*}
    Hence, we get that $\neg \eqref{eq: dom_eq} = \eqref{eq: adm_eq}$.
\end{proof}
\fi

\ifproof
\begin{proof}[Proof of Thm. \ref{thm: adm_str}]
    Given an admissible strategy $\sigma$, and prefix $h \in \plays^{v_0}(\sigma)$, we want to prove that $\sigma \iff \scoop \vee \wcoop$, i.e., $\sigma$ is the witnessing strategies that satisfies $\scoop$ or $\wcoop$ conditions. 

We will prove the statement using the following tautology 
\begin{equation*}
(A \iff B) \iff \big[(B \implies A) \wedge (\neg B \implies \neg A) \big]    
\end{equation*}

where $A$ is the logical statement that $\sigma$ is admissible and $B$ is the logical statement that implies that either $\cVal(h, \sigma) < \aVal(h)$ or $\aVal(h) = \aVal(h, \sigma)  = \cVal(h, \sigma) = \acVal(h, \sigma)$ is true. Note that all the plays in the game start from $v_0$, and thus, all prefixes start from $v_0$.

\paragraph{Case I: $(\neg B \implies \neg A)$} $\neg B$ is equivalent to Eq. \eqref{eq: dom_eq}. Let Eq.~\eqref{eq: dom_eq_1} hold, i.e., $\cVal(h, \sigma) \geq \aVal(h) \wedge \aVal(h, \sigma) > \aVal(h)$. Let us assume that there exists $\sigma' \neq \sigma$, which is compatible with prefix $h$ and is $\wco$ afterward. That is, 
$\sigma'$ is the witnessing strategy for $\wco$ after prefix $h$ and thus $\aVal(h, \sigma') = \aVal(h)$. We now claim that $\sigma'$ weakly dominates $\sigma$.

For $\sigma'$ to weakly dominate $\sigma$, it needs to have a payoff always lower than or equal to the payoff associated with $\sigma$. Since $\aVal(h, \sigma) > \aVal(h) = \aVal(h, \sigma')$ that means there exists a $\tau \in \Tau$ for which $\Val(h \cdot \play^{h}(\sigma', \tau)) < \Val(h \cdot \play^{h}(\sigma, \tau))$. Further, we have that $\cVal(h, \sigma) \geq \aVal(h, \sigma')$ which implies that for any $\tau \in \Tau$ which is compatible with prefix $h$, the payoff $\Val(h \cdot \play^{h}(\sigma, \tau)) \geq \Val(h \cdot \play^{h}(\sigma', \tau))$. Since, $\Val(h \cdot \play^{h}(\sigma, \tau)) = \Val(h) + \Val(\play^{h}(\sigma, \tau))$, we get $\Val(\play^{h}(\sigma, \tau)) \geq \Val(\play^{h}(\sigma', \tau))$ for all $\tau \in \Tau$. Thus, for all Env player strategies, the payoff associated with $\sigma'$ is less than or equal to $\sigma$. Thus, $\sigma'$ weakly dominates $\sigma$.

Now, let Eq. \eqref{eq: dom_eq_2} hold, i.e., $\aVal(h) = \aVal(h, \sigma) = \cVal(h, \sigma) \wedge \acVal(h) < \aVal(h)$. By definition of $\acVal(h)$, there exists $\sigma'$ that is $\wcoop$ and $\sigma' \neq \sigma$ such that $\cVal(h, \sigma') = \acVal(h)$ and $\cVal(h, \sigma') < \aVal(h)$ and $\sigma'$ is compatible with prefix $h$. This implies that for some $\tau \in \Tau$ which is compatible with prefix $h$, the payoff $\Val(h \cdot \play^{h}(\sigma', \tau)) < \Val(h \cdot \play^{h}(\sigma, \tau)) \implies \Val(\play^{h}(\sigma', \tau)) < \Val(\play^{h}(\sigma, \tau)) $. Further, $\aVal(h, \sigma') \leq \aVal(h) = \cVal(h, \sigma)$. This implies that for all $\tau \in \Tau$ compatible with prefix $h$, $\Val(h \cdot \play^{h}(\sigma', \tau)) \leq \Val(h \cdot \play^{h}(\sigma, \tau)) \implies \Val(\play^{h}(\sigma', \tau)) \leq \Val(\play^{h}(\sigma, \tau))$. Hence, $\sigma'$ weakly dominates $\sigma$. 

\paragraph{Case II: $(B \implies A)$} Assume that for all prefixes $h$ of $\plays^{v_0}(\sigma)$ we have that Eq. \eqref{eq: adm_eq} holds such that $\sigma$ is the witnessing strategy for $\scoop$ or $\wcoop$. Now, let $\sigma'$ be another strategy which is compatible with prefix $h$ but ``splits" at $h$, i.e., $\sigma(h) \neq \sigma'(h)$. Let us assume that $\sigma'$ weakly dominates $\sigma$. We will show a contradiction. 

 Let Eq. \eqref{eq: adm_eq_1} hold, i.e., $\sigma$ is $\scoop$ and $\sigma'$ is not. We can use Lemma \ref{lem: scoop_proof} to show that $\sigma'$ does not weakly dominate $\sigma$. Similarly, if Eq. \eqref{eq: adm_eq_2} holds, then $\sigma$ is $\wcoop$. As Eq. \eqref{eq: adm_eq_1} does not hold, therefore Eq. \eqref{eq: adm_eq_2} must hold. We can use Lemma \ref{lem: wcoop_proof} to show that $\sigma'$ does not weakly dominate $\sigma$. Thus, $\sigma'$ does not weakly dominate $\sigma$ and hence $\sigma$ is admissible. 

 Since we have shown both $(\neg B \implies \neg A)$ and $(B \implies A)$, hence, for all prefixes starting from $v_0$, Eqs. \eqref{eq: adm_eq_1} and \eqref{eq: adm_eq_2} are necessary and sufficient conditions for strategy $\sigma$ to be admissible.

\end{proof}
\fi

\begin{theorem}
    \label{thm: adm_exist}
    There always exists an admissible strategy $\sigma_{adm}$ in a 2-player, turn-based, total-payoff, reachability game $\G$, and $\sigma_{adm}$ solves Problem \ref{prob: problem_1} if $\cVal(v_0, \sigma_{adm}) \leq \B$.
\end{theorem}
\ifproof
\begin{proof}
    The result follows from Thm. \ref{thm: adm_str} that states that $\sco$ or $\wcoop$ are necessary and sufficient conditions for a strategy to be admissible. Using the results of \citeauthor[Thm. 1]{brihaye2017pseudopolynomial}, we show that the value iteration algorithm on $\G$ will converge after finite iterations to a fixed point. Then, using \citeauthor[Corollary 18 and Proposition 19]{brihaye2017pseudopolynomial}, we prove that witnessing strategies for $\wco$ and $\coop$ always exists. 
 
    Further, from Thm. \ref{thm: wcoop_memoryless}, we know that for any history $h$, $\acVal(h)$ is same as the optimal cooperative value in the sub-game starting from $\last(h)$. Thus, witnessing strategies for $\wco$ and $\coop$ are sufficient conditions for $\wcoop$ always exist. 

    Let us not define strongly cooperative optimal strategy ($\sco$). A strategy that is $\sco$ is also $\scoop$ by definition. Thus, it suffices for us to show that $\sco$ always exists to prove that $\sc$ always exists. For any history $h$, $\sco$ is a strategy such that if $\cVal(h) < \aVal(h)$ then $\cVal(h, \sigma) = \cVal(h)$ and if $\cVal(h) = \aVal(h)$ then $\aVal(h, \sigma) = \aVal(h)$. We observe that $\cVal(h, \sigma) = \cVal(h)$ is definition for $\coop$ and by definition $\cVal(h) \leq \aVal(h)$ and thus $\cVal(h, \sigma) \leq \aVal(h)$. As witnessing strategy for  $\coop$ always exists, $\sco$ strategies always exist. Thus. $\scoop$ always exists.
    
    Hence, admissible strategies always exist.
\end{proof}
\fi

The proof follows from Lemma \ref{lem: wcoop_proof} and Theorem~\ref{thm: wcoop_memoryless}. Observe that Theorem~\ref{thm: adm_str} characterizes the set of \emph{all} admissible strategies and Theorem~\ref{thm: adm_exist} establishes their existence in full generality, i.e., independent of $\B$, for 2-player turn-based games with reachability objectives. For computability considerations, we bound the payoffs associated with plays to a given budget $\B$ so that the set of all admissible strategies is finite. 
This is a reasonable assumption  
that is often used in, e.g., energy games with fixed initial credit and robotics application with finite resources \cite{chakrabarti2003resource,bouyer2008infinite,muvvala2023symbolic,filiot2010iterated}.

\subsection{Admissible Strategy Synthesis}
\label{ssec: algo_adm_str}
Given $\G$ and budget $\B$, we first construct a game tree arena that captures all plays with payoff less than or equal to $\B$. Next, we show that the payoff function on this tree is history-independent, which allows us to modify Theorem~\ref{thm: adm_str} and compute a finite set of $\aVal$s. We conclude the following from Lemma~\ref{lem: scoop_proof} and Theorem~\ref{thm: scoop_memory}.

\begin{algorithm}[t]
    \caption{Admissible Strategy Synthesis}
    \label{algo: pseudo_naive_adm}
    \SetKwInOut{Input}{Input}\SetKwInOut{Output}{Output}
    \SetKwComment{Comment}{/* }{ */}
    \Input{Game $\G$, Budget $\B$}
    \Output{Strategy $\Sigma_{adm}$}
    \SetKwFunction{ValueIteration}{ValueIteration}
    \SetKwFunction{push}{push}
    \SetKwFunction{pop}{pop}
    \SetKwFunction{nxt}{next}
    \SetKwFunction{iter}{iter}
    \SetKwProg{try}{try}{:}{}
    \SetKwProg{catch}{catch}{:}{end}
    $\G' \gets$ Unroll $\mathcal{G}$ up until payoff $\B$ \\
    $\aVal; \cVal \gets$ \ValueIteration($\G'$) \\
    \lForAll{$v$ in $\G'$}{ $\acVal(v) \gets$ as per Def. \ref{def: wcoop}}
    \lIf{$\B <\cVal(v_0)$}{\KwRet{$\G$}}
    $h.\push\big((v_0, \{\delta(v_0, a_s)\})\big)  \;\; \text{\# let $h$ be a stack}$ \\
    \While{$h \neq \emptyset$}{
        $v, \{v'\} \gets h[-1]$ \\
        \try{$v' \gets \nxt(\iter(\{v'\}))$}{
             \If{$v \in V_s$ and $(\eqref{eq: pre_indep_adm_eq_1} \vee \eqref{eq: pre_indep_adm_eq_2})$ holds}{
             $h.\push\big((v', \{\delta(v', a_e)\})\big)$ \\
            $\Sigma_{adm}: h \to v'$ \text{\# Add only states in $h$}\\
             }
             \lIf{$v \in V_e$}{$h.\push\big((v', \{\delta(v', a_s)\})\big)$}
        }
        \lcatch{StopIteration}{$h.\pop()$}
    }
    \KwRet{$\Sigma_{adm}$}
\end{algorithm}

\begin{corollary}
    \label{cor: adm_memoryless}
    Memoryless strategies are \emph{not} sufficient for admissible strategies.
\end{corollary}
\ifproof
\begin{proof} 
    We show this by counterexample. Consider the game in Figure \ref{fig: local_conds_not_sufficient}. 
    Let us consider three strategies: $\sigma_1: (v_3 \to v_5)$, $\sigma_2: (v_3 \to v_6), (v_7 \to v_8)$ and $\sigma_3: (v_3 \to v_6),(v_7 \to v_9)$.
    Let $h := v_0 v_2 v_4 v_6$, then strategy $\sigma_2$ and $\sigma_3$ are admissible. Notice that these strategies are $\sco$.
    If $h := v_0 v_1 v_3$, then only $\sigma_1$ and $\sigma_2$ are admissible. Notice for $h = v_0 v_1 v_3$, $\cVal(h, \sigma_3) > \aVal(h)$. Thus, only $\sigma_1$ and $\sigma_2$ are $\sco$ and admissible. Thus, the admissibility of a strategy is history-dependent.
\end{proof}
\fi

A consequence of Corollary~\ref{cor: adm_memoryless} is that we cannot use a backward induction-based algorithm to compute these strategies. Instead, we use a forward search algorithm that starts from the initial state and recursively checks whether the admissibility constraints are satisfied.

\paragraph{Game Tree Arena.}
The algorithm is outlined in Alg. \ref{algo: pseudo_naive_adm}. Given game $\G$ and budget $\B$, we construct a tree of plays $\G'$ by unrolling $\G$ 
until the payoff associated with a play exceeds $\B$ or a goal state is reached. 
Every play in $\G$ corresponds to a branch in $\G'$. Every play that reaches a goal state with a payoff $b \leq \B$ in $\G$ is a play that ends in a leaf node in $\G'$, which is marked as a goal state for Sys player. The leaf nodes in $\G'$ that correspond to the plays with payoff $b > \B$ in $\G$ are assigned a payoff of $+\infty$.
Further, in $\G'$, the weights along all the edges are zero, and the payoffs are strictly positive only when a play reaches a leaf node, otherwise it is zero. 
By construction, the payoff function $\Val$ is history-independent in $\G'$. 
Then, Theorem~\ref{thm: adm_str} can be restated as the following lemma.

 \begin{lemma}
    Given $\G'$, strategy $\sigma$ is admissible if and only if $\forall h \in \plays^{v_0}(\sigma)$ with $\last(h) \in V_s$ and $v' = \delta(\last(h), \sigma(h))$, the following holds,
    \begin{subequations}
        \begin{align}
            & \big(\cVal(v') < \min\{\aValues\} \big) \vee \label{eq: pre_indep_adm_eq_1} \\
            & \big(\aVal(v^{\dag}) = \aVal(v') = \cVal(v') = \acVal(v^{\dag})\big) \label{eq: pre_indep_adm_eq_2}
        \end{align}    
    \end{subequations}
    where $\aValues := \{\aVal(v) \mid v \in h\}$ is the set of adversarial values along history $h$ and $v^{\dag} = \last(h)$.
    \label{lem: pre_indep_adm}
 \end{lemma}
 \ifproof
 \begin{proof}
     We first show that $\Val$ is indeed history-independent on $\G'$.      For all $h$ in $\G'$, such that $\last(h)$ is a not leaf node., 
     as all the weights in $\G'$ are $0$, the payoff value for any prefix of history $h$ is exactly $0$. For all histories $h$, in $\G'$, such that $\last(h)$ is a leaf node, by construction, the payoff is associated with the last state, i.e., $\Val(h) = \Val(\last(h))$. Specifically, $\Val(h) = \Val(\last(h)) \leq \B$ if the leaf node is a goal state for the Sys player else $\Val(h) = \Val(\last(h)) = \infty$. Thus, 
     \begin{align*}
         \cVal(h) = \min_{\sigma \in \Sigma} \min_{\tau \in \Tau} \Val(\last(\play^{v_0}(\sigma, \tau))) \\
         \aVal(h) = \min_{\sigma \in \Sigma} \max_{\tau \in \Tau} \Val(\last(\play^{v_0}(\sigma, \tau)))
     \end{align*}

     From above, we have that $\aVal(h) = \aVal(\last(h))$ and $\cVal(h) = \cVal(\last(h))$, and $\acVal(h) = \min\{\cVal(\last(h), \sigma) | \sigma \in \Sigma, \; \aVal(\last(h), \sigma) \leq \aVal(\last(h))\}$. 
     Thus, for $v':= \delta(\last(h), \sigma(h))$, we can rewrite the equation after disjunction in Eq. \eqref{eq: adm_eq} and get Eq. \eqref{eq: pre_indep_adm_eq_2}.
     
    While it is sufficient to look only at the last state along a history $h$ to compute $\aVal$,  for $\scoop$ condition, $\cVal(v') < \aVal(\last(h))$ does not imply that $\cVal(v') < \aVal(h_{\leq j}) \; \forall j \in \mathbb{N}_{|h|}$. This is also evident from Figure \ref{fig: local_conds_not_sufficient} where the adversarial values for Sys player states along $h:= v_0 v_1 v_3 v_6 v_7$ and $v' := v_9$ is $\aVal(v_0) = \infty; \aVal(v_3) = 4; \aVal(v_7) = 9$. While $\cVal(v_9) < \aVal(v_7)$,  $\cVal(v_9) > \aVal(v_3)$. 
    
    Thus, 
    we need to check if $\cVal(v') < \aVal(v) (\iff \cVal(v') < \min\{\aVal(v)\})$ for all Sys player states along $h$ to check for admissibility of strategy $\sigma$ compatable with $h$. As $\aVal$ is history-independent, there are finitely many adversarial values for a given graph $\G'$, and thus $\{\aVal(v)\}$ is finite in size.
    
    For $\wcoop$ condition, it is sufficient to evaluate $\acVal$, $\aVal$, and $\cVal$ at $\last(h)$ to compute the corresponding optimal values for payoff independent functions. Thus, to check if $\aVal(h) = \aVal(h, \sigma)$ it sufficies to check $\aVal(\last(h)) = \aVal(\last(h), \sigma)$. We can make the same argument for $\cVal(\last(h), \sigma) = \acVal(\last(h))$. Hence, we get the equations in Lemma \ref{lem: pre_indep_adm}. 
 \end{proof}
 \fi
After constructing $\G'$, Alg.~\ref{algo: pseudo_naive_adm} computes $\aVal$ and $\cVal$ values for each state in $\G'$ using the Value Iteration algorithm \cite{brihaye2017pseudopolynomial}. If $\B < \cVal(v_0)$, then there does not exist a play that reaches a goal state in $\G'$. Thus, all strategies in $\G$ are admissible. To compute admissible strategies on $\G'$, we use a DFS algorithm to traverse every play and check if the admissibility criteria from Lemma \ref{lem: pre_indep_adm} is satisfied. If yes, we add the history and the successor state to $\Sigma_{adm}$. We repeat this until every state in $\G'$ is explored. This algorithm is sound and complete with polynomial time complexity.

\begin{theorem}[Sound and Complete]
    Given $\G$ and a budget $\B$, Alg.~\ref{algo: pseudo_naive_adm} returns the set of all admissible strategies $\Sigma_{adm}$. The algorithm runs in polynomial time when $\B$  is fixed and in pseudo-polynomial time when $\B$ is arbitrary.
    \label{thm: sound_correct_naive_adm}
\end{theorem}
\ifproof
\begin{proof}
    Let us define $\aVal(\G, v)$ and $\cVal(\G, v)$ to be the optional adversarial and cooperative value for state $v$ in $\G$. We extend the definition to $\acVal(\G, v)$ accordingly.
    We begin our proof by first noting that $\G'$ is a finite tree arena where the leaf nodes are partitioned into goal states for the Sys player and sink states. By construction, every state that is not a leaf node in $\G'$ has at least one outgoing edge. Only plays that reach a goal state in $\G'$ have a finite payoff value. Plays that fail to reach a goal state have a payoff of $+\infty$. Thus, using \citeauthor{khachiyan2008short}'s algorithm, we can compute optimal $\aVal(\G', v)$ and $\cVal(\G', v)$ and the corresponding witnessing strategies. Thus, game $\G'$ is well-formed, which implies that an admissible strategy always exists in $\G'$. 
    
    Next, using Lemma \ref{lem: pre_indep_adm}, we see that the payoff function for $\G'$ is history-independent. Thus, checking for admissibility reduces to checking for the adversarial and cooperative values at each state along a play $\G'$. If $\B < \cVal(\G, v_0)$, then there does not exist a play that reaches a goal state in $\G'$. Hence, $\aVal(\G', v) = \cVal(\G', v) = +\infty \; \forall v \in V$ and every valid action from every state is part of a strategy that is $\wcoop$. The Sys player may choose $\sigma$ indifferently in such cases. If $\B \geq \cVal(\G, v_0)$, there exists at least one play for which the payoff is finite. Hence, there exist states along such plays for which $\cVal(\G', v) \neq \infty$ and for all states $\aVal(\G', v) \leq \infty$ for every play. 

    In Thm. \ref{thm: sco_memory}, we show that even for a tree-like arena whose payoff function is also history-independent, the admissibility of an action depends on the history. 
    As a consequence of Corollary \ref{cor: adm_memoryless}, we use Depth First Search (DFS) based approach to explore nodes along each play until all the nodes are explored. At each state, we check if the $\sco$ or $\wcoop$ conditions are satisfied. If yes, then we add that history and the corresponding action to the transducer that represents finite memory strategy for admissible strategies. Algo. \ref{algo: pseudo_naive_adm} returns such a strategy, and hence, the transducer returned is correct.
\end{proof}
\fi

\section{Admissible Winning Strategies}
\label{sec: adm_win_str}

Although value preservation is a desirable attribute, Lemma~\ref{lem: quant_adm_not_val_pres} shows that admissible strategies do not ensure this. In contrast, an admissible winning strategy is value-preserving and enforces reaching a goal state from the winning region. 
For instance, for the game in Fig. \ref{fig: cex_adm_val_not_preserving}, an admissible winning strategy commits to $v_1$ from $v_0$, which is value-preserving. This is desirable as it ensures reaching $v_6$, while the other admissible strategies do not.

Here, we identify the subset of admissible strategies that are not value-preserving (aka, optimistic strategy), and prove that if they exist, they must be $\scoop$. Next, we propose $\modscoop$ that are admissible winning and show that $\modscoop$ and $\wcoop$ are admissible winning. Finally, we show that they always exist and give our synthesis algorithm.

\subsection{Optimistic Strategy}

In Sec. \ref{ssec: charc_regions}, we show that $\scoop$ strategies are willing to risk a higher payoff $(\aVal(\delta(\last(h), \sigma(h))) > \cVal(h, \sigma))$ in the hopes that the Env will cooperate, i.e., they are optimistic. 

\begin{definition}[Optimistic Strategy]
    \label{def: optmistic_str}
    Strategy $\sigma$ is an optimistic strategy if, and only if, $\sigma$ is admissible but not value preserving. 
\end{definition}

We now show that if optimistic strategies exist, then they must be $\scoop$ strategies. This implies, $\wcoop$ strategies are never optimistic strategies.

\begin{lemma}
    If an optimistic strategy exists, it must be $\scoop$ strategy. 
    \label{lem: optmistic_str}
\end{lemma}
\ifproof
\begin{proof}
    Given $\G$, the initial state of the game can belong to three regions. Thus, $v_0$ is either part of the losing, pending, or winning region. From Thm. \ref{thm: adm_str}, a strategy $\sigma$ is admissible if, and only if, it is $\sco$ or $\wcoop$.

    \paragraph{Case I: $v_0 \in V_{los}$} By definition, there does not exist a play that reaches the goal state from $v_0$. Thus, for all histories $h$, $\cVal(h) = \aVal(h) = \infty$. If $\sigma$ is $\sco$ then $\cVal(h, \sigma) = \aVal(h)$ which is true for every strategy in $\G$.    
    Thus, if $v_0 \in V_{los}$ then every strategy is $\sco$ strategy. If $\sigma$ is $\wcoop$, then $\aVal(h) = \aVal(h, \sigma)$ and $\cVal(h, \sigma) = \acVal(h)$. Since, for all histories $h$, $\aVal(h) = \cVal(h) = \infty$ we have that $\acVal(h, \sigma) = \infty$. Thus, every action from every state in $\G$ belongs to an admissible strategy $\sigma$ that is $\wcoop$.

    \paragraph{Case II: $v_0 \in V_{pen}$} If $v_0 \in V_{pen}$, then there exists a play that reaches the goal state in $\G$. If $\sigma$ is $\sco$ then it must satisfy $\cVal(h, \sigma) < \aVal(h)$. Let $\last(h) \in V_{pen}$ and $\sigma(h) \in V_{los}$, then $\cVal(h, \sigma) = \infty \not< \aVal(h)$ thus $\sigma$ is not $\sco$ strategy. For any prefix $h$ that ends in Sys player state in the pending region, if $\sigma$ is $\sco$, then $\sigma(h) \in V_{win} \cup V_{pen}$. Thus, $\sVal(v) \in \{0, 1\}$ for all plays induced by strategy $\sigma$ which is $\sco$. 

    Let $\sigma$ be $\wcoop$ strategy. Then, by definition, if $\sigma(h) \in V_{los}$ then $\cVal(h, \sigma) \neq \acVal(h)$ and $\cVal(h, \sigma) = \infty$ for all prefixes $h$. Thus, if $\sigma$ is $\wcoop$ strategy then $\sigma(h) \in V_{pen} \cup V_{los}$. Thus, $\sVal(v) \in \{0, 1\}$ for all plays induced by strategy $\sigma$ which is $\wcoop$. Hence, every strategy that is $\sco$ is value preserving if $\last(h) \in V_{pen}$ or is $\wcoop$.

    \paragraph{Case III: $v_0 \in V_{win}$} For this case, there exists a winning strategy $\sigma_{win}$ from the initial state. Thus, $\sVal(v) = 1$, for all the states in plays induced by the winning strategy as they always stay in the winning region. As every winning strategy is worst-case optimal, it satisfies $\aVal(h) = \aVal(h, \sigma)$. Let $\Sigma_{win}$ and $\Sigma_{\wcoop}$ be the set of all winning and worst-case cooperative optimal strategies. Then, by definition, we have $\Sigma_{\wcoop} \subseteq \Sigma_{win}$ and every $\wcoop$ strategy is winning strategy. Hence, $\wcoop$ strategies are value preserving.

    Now, let $\sigma$ be $\sco$ such that $\last(h) \in V_{win}$. Then if $\sigma(h) \in V_{los}$ then $\cVal(h, \sigma) > \aVal(h)$. If $\sigma(h) \in \{V_{pen} \cup V_{los}\}$ then $\cVal(h, \sigma)$ can be less than $\aVal(h)$. Thus, for all prefixes $h$ such that $\last(h) \in V_{win}$, if $\sigma$ is $\sco$ then it is not value preserving. 

    From above, we have that, for any history $h$, $\wcoop$ are value preserving strategies and hence are not optimistic strategies. For histories $h$ where $\last(h) \in V_{win}$, if $\sigma$ is $\sco$ then it may not be value preserving. Thus, if an optimistic strategy exists, then it is a $\sco$ strategy. 
\end{proof}
\fi

Thus, to enforce value-preserving, we modify Def. \ref{def: scoop}.
\begin{definition}[$\modscoop$]
    For all $h \in \plays^{v_0}(\sigma)$, strategy $\sigma$ is $\modscoop$, if $\sigma$ is $\scoop$ and value-preserving, i.e., if $\sVal(\last(h)) = 1$ then $\sVal(\delta(\last(h), \sigma(h))) = 1$.
\end{definition}

We say that every strategy that is either $\wcoop$ or $\modscoop$ is admissible winning as they are value preserving and admissible. 
In 
Fig.~\ref{fig: cex_adm_val_not_preserving},
both $\sigma_2$ and $\sigma_3$ are admissible, but $\sigma_3$ is not $\modscoop$ as it is not value preserving. Thus, only $\sigma_2$ is admissible winning strategy.
\begin{theorem}
    \label{thm: adm_win_str}
    A strategy $\sigma$ is \emph{admissible winning} if, and only if, $\forall h \in \plays^{v_0}(\sigma)$ with $\last(h) \in V_s$, the following holds
    \begin{subequations}
        \begin{align}
           & \Bigl( \big(\cVal(h, \sigma) < \aVal(h) \big) \; \wedge \label{eq: adm_win_eq_1} \\
           & \big(\sVal(\last(h)) = 1 \implies \sVal(\delta(\last(h), \sigma(h))) = 1 \big) \Bigr) \; \vee \notag \\
           & \bigl ( \aVal(h) = \aVal(h, \sigma) = \cVal(h, \sigma) = \acVal(h) \bigr) \label{eq: adm_win_eq_2}
       \end{align}%
    \end{subequations}
\end{theorem}

The additional term in Eq. \eqref{eq: adm_win_eq_1}
in comparison to Eq. \eqref{eq: adm_eq_1}
constrains a strategy $\sigma$ that satisfies $\scoop$ condition in the winning region to action(s) such that $\delta(\last(h), \sigma(h)) \in V_{win}$. For all $h$ with $\last(h) \notin V_{win}$, the condition is the same as the $\scoop$ condition in Eq. \eqref{eq: adm_eq_1}. 

\begin{lemma}
    \label{lem: adm_win_exist}
    There always exists an admissible winning $\sigma^{win}_{adm}$ strategy in a 2-player, turn-based, total-payoff, reachability games $\G$, and $\sigma^{win}_{adm}$ solves Problem \ref{prob: problem_2} if $\aVal(v_0, \sigma^{win}_{adm}) \leq \B$.
\end{lemma}
\ifproof
\begin{proof}
    We use the proof from Lemma \ref{lem: optmistic_str} that shows that $\sco$ and $\wcoop$ always exist. Then, we use Thm. \ref{thm: adm_exist} to show that a witnessing strategy always exists for optimal $\aVal$ and $\cVal$ in $\G$. As $\wco$ and $\coop$ always exist, $\wcoop$ strategies always exist.  
    Hence, admissible winning strategies always exist.  
\end{proof}
\fi

\subsection{Admissible Winning Strategy Synthesis}

Here, we give an algorithm to solve Problem \ref{prob: problem_2}. Similar to admissible strategies, admissible winning strategies are history-dependent as shown by the next theorem.

\begin{theorem}
    \label{thm: adm_win_memory}
    Memoryless strategies are \emph{not} sufficient for admissible winning strategies.
\end{theorem}
\ifproof
\begin{proof}
    Write proof here
\end{proof}
\fi

The algorithm is similar to Alg. \ref{algo: pseudo_naive_adm} except for the admissibility checking criteria. 
The tree construction, syntax, and semantics are as outlined in Sec. \ref{ssec: algo_adm_str}. We modify admissibility checking criteria in Line 9 of Alg. \ref{algo: pseudo_naive_adm} to $\cVal(v') < \min\{\aValues\}$ and $\neg (v \in V_{win}) \vee (v' \in V_{win})$
This modification expresses the value-preserving term from Eq. \eqref{eq: adm_win_eq_1}.

\begin{theorem}[Sound and Complete]
    \label{thm: sound_correct_adm_win}
    Given $\G$ and budget $\B$, the algorithm described returns the set of all admissible winning strategies $\Sigma^{win}_{adm}$ and has same time complexity as Alg.~\ref{algo: pseudo_naive_adm}.
\end{theorem}
\ifproof
\begin{proof}
    Add the outline of the proof here. 
\end{proof}
\fi

\section{Illustrative Examples}
We now discuss the emergent behavior under admissible and admissible winning strategies in two settings: (i) a gridworld domain, where two agents take actions in turns, and (ii) a manipulator domain, where a robotic arm and a human operate in a shared workspace.
For both domains, the task $\phi$ is specified using $\ltlf$ formulas. We first construct a game abstraction $\G$. 
Next, we construct a Deterministic Finite Automaton for the task \cite{fuggittiltlf2dfa} and take the product to construct the product game, where the objective for the Sys player is to reach a set of goal (accepting) states. Additional experiments and empirical validations of PTIME complexity of Alg.~1 for fixed $\B$ are provided in the extended version \cite{muvvala2025admissiblearxiv}. 
Code is available on Github \cite{admtool}. 
\begin{figure}[t]
    \centering
    \begin{subfigure}[t]{0.24\linewidth}
        \includegraphics[width=0.99\linewidth]{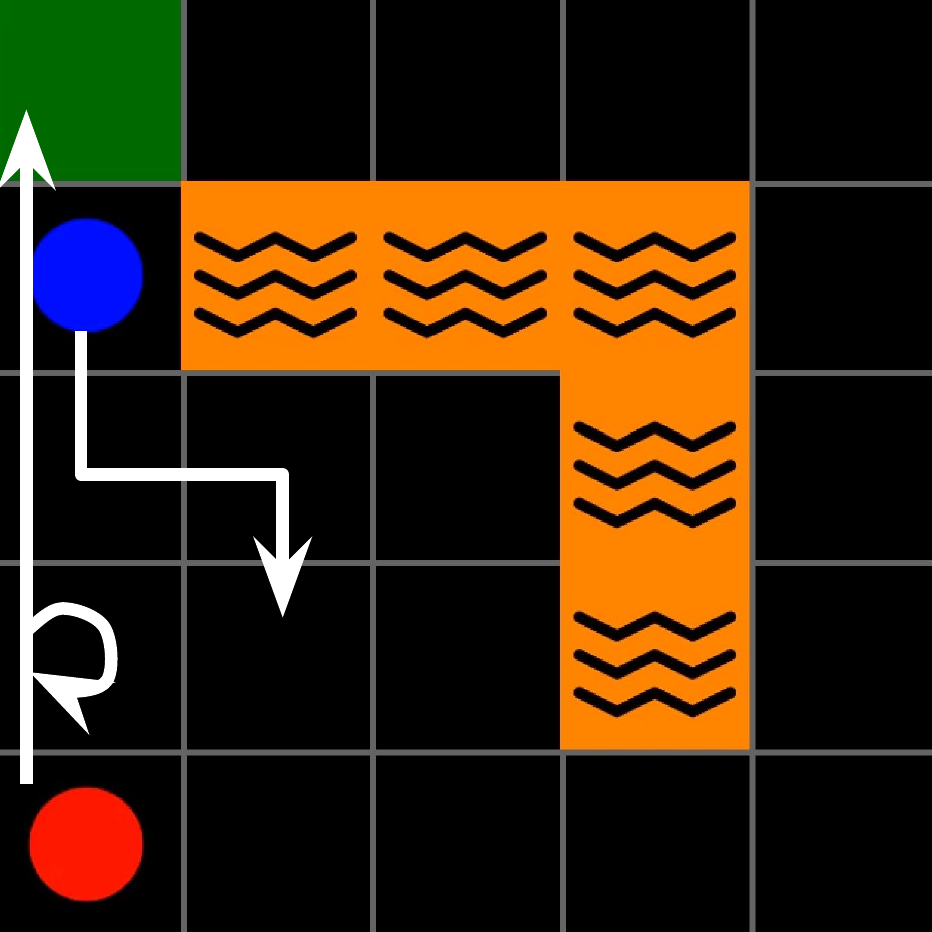}
        \caption{$\sigma_{adm},5$}
        \label{fig: gridworld1}
    \end{subfigure}%
    ~
    \begin{subfigure}[t]{0.24\linewidth}
        \includegraphics[width=0.99\linewidth]{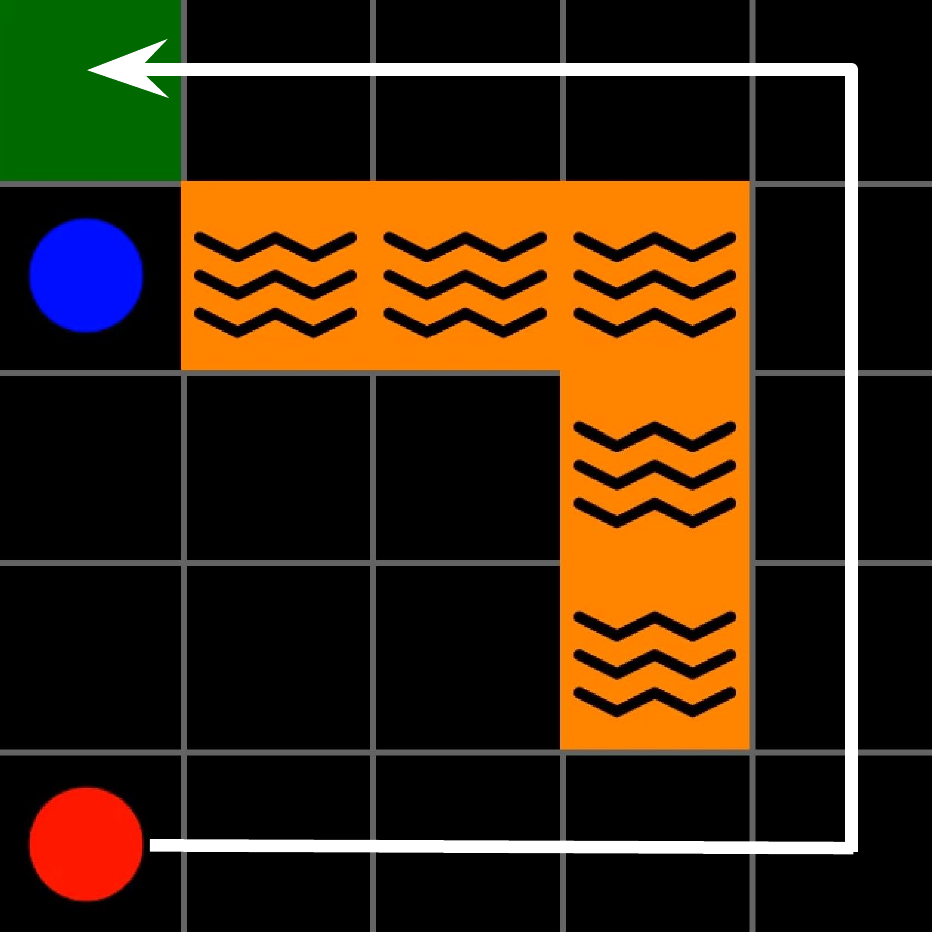}
        \caption{$\sigma^{win}_{adm}, 12$}
        \label{fig: gridworld2}
    \end{subfigure}%
    ~
    \begin{subfigure}[t]{0.24\linewidth}
        \includegraphics[width=0.99\linewidth]{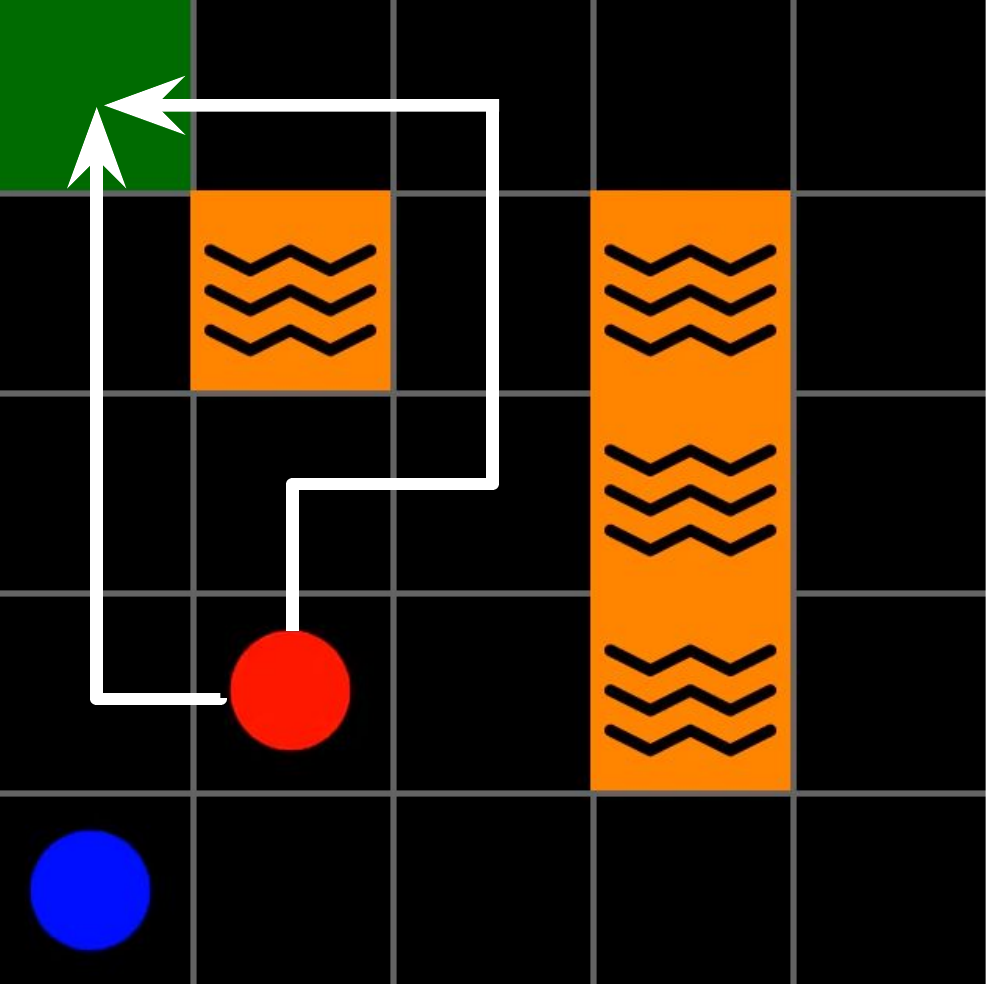}
        \caption{$\sigma_{adm},8$}
        \label{fig: gridworld3}
    \end{subfigure}%
    ~
    \begin{subfigure}[t]{0.24\linewidth}
        \includegraphics[width=0.99\linewidth]{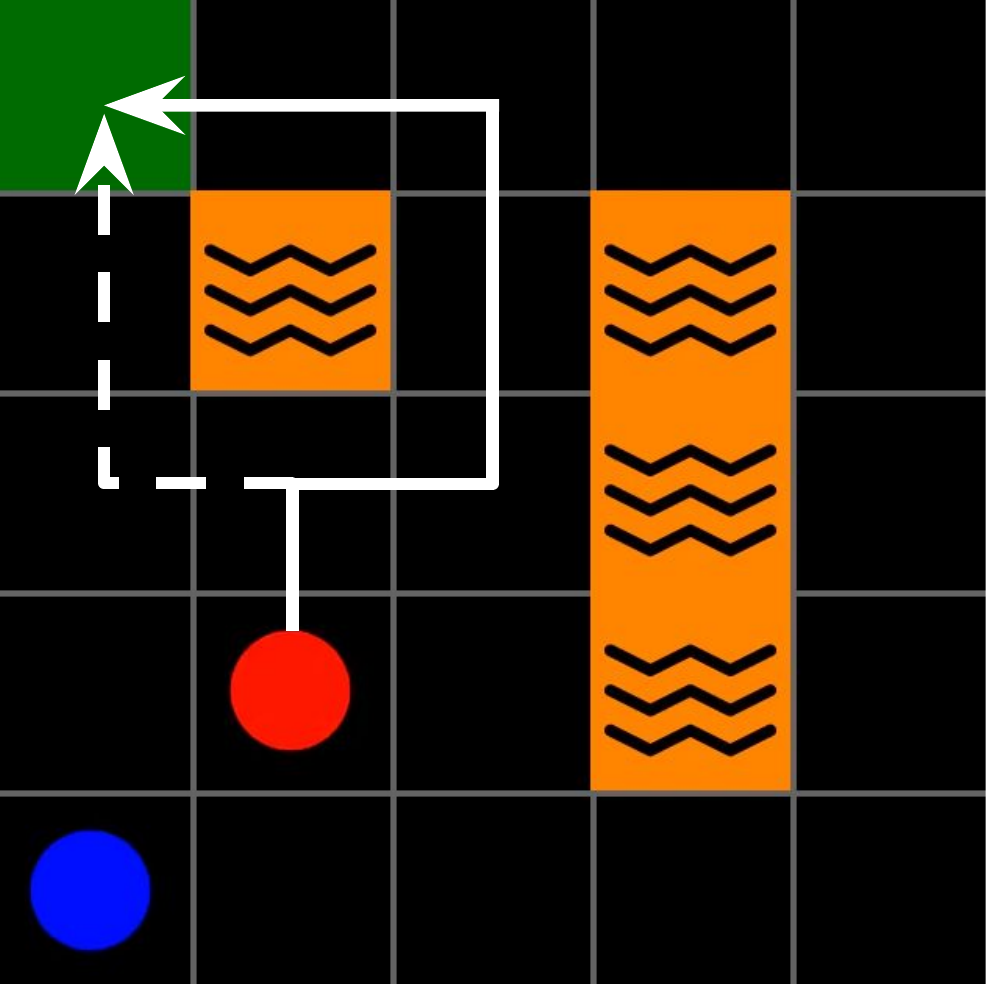}
        \caption{$\sigma^{win}_{adm},8$}
        \label{fig: gridworld4}
    \end{subfigure}%
    \caption{Gridworld Example. Sub-captions show strategy and $\B$.}
    \label{fig: gridworld_rollout}
\end{figure}

\paragraph{Gridworld.}

Fig. \ref{fig: gridworld_rollout} illustrates a gridworld domain with Sys (red), Env (blue), goal (green), and lava (orange) states. The objective for Sys is to reach the goal state while avoiding the Env player. The players should not enter lava and Env player cannot traverse through goal state. Action cost is 1 for all Sys player actions. The game is played in turns starting with Sys where both players take a step in each cardinal direction.

Figs.~\ref{fig: gridworld1} and \ref{fig: gridworld3} illustrate the initial position of both players for two different scenarios. 
Figs. \ref{fig: gridworld1}-\ref{fig: gridworld2} and \ref{fig: gridworld3}-\ref{fig: gridworld4} illustrate $\sigma_{adm}$ and $\sigma^{win}_{adm}$ for different budgets.
In Figs.~\ref{fig: gridworld1}-\ref{fig: gridworld2}, a winning strategy exists if $\B \geq 12$. Hence, for $\B=5$ (Fig.~\ref{fig: gridworld1}), only $\sigma_{adm}$ exists, which relies on Env's cooperation to reach goal state. For $B =12$ (Fig.~\ref{fig: gridworld2}), $\sigma^{win}_{adm}$ exists and commits to go around and reach the goal state. Note that $\sigma_{adm}$ (not shown) also exists but does not commit to go around. 

In Fig.~\ref{fig: gridworld3}, we start in the winning region. Despite $v_0 \in V_{win}$, $\sigma_{adm}$ 
chooses to go west (in $V_{pen}$) or north (in $V_{win}$). In contrast, $\sigma^{win}_{adm}$ in Fig.~\ref{fig: gridworld4} stays in the winning region (goes north) as it is the sole winning strategy that is also admissible, with a potentially shorter path (dashed) if Env cooperates.

\textit{Comparison against Best-Effort (BE): }
Here we illustrate differences between BE and admissible strategies. 
Besides qualitative vs quantitative, in the winning region, a major difference is that 
every $\sigma_{win}$ is also BE, but not every $\sigma_{win}$ is admissible winning. In Fig. \ref{fig: gridworld4}, if Env player stays in the left corner (cooperative), $\sigma^{win}_{adm}$ \emph{commits} to dashed-line strategy as the solid-line strategy is not $\coop$, whereas BE does not distinguish between the two strategies as both are winning.

\begin{figure}[t!]
    \centering
    \begin{subfigure}[b]{0.43\columnwidth}
        \includegraphics[width=0.99\linewidth]{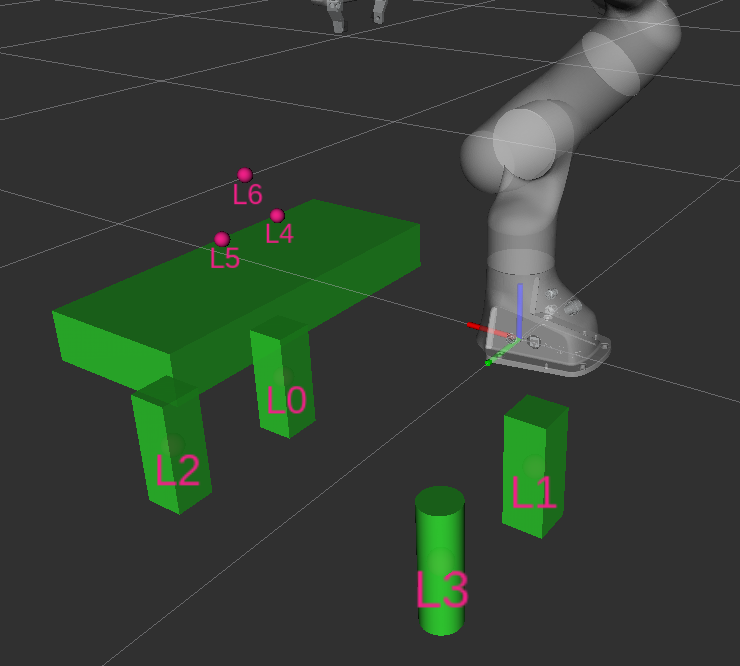}
        \label{fig: manip1}
    \end{subfigure}%
    ~
    \begin{subfigure}[b]{0.43\columnwidth}
        \includegraphics[width=0.99\linewidth]{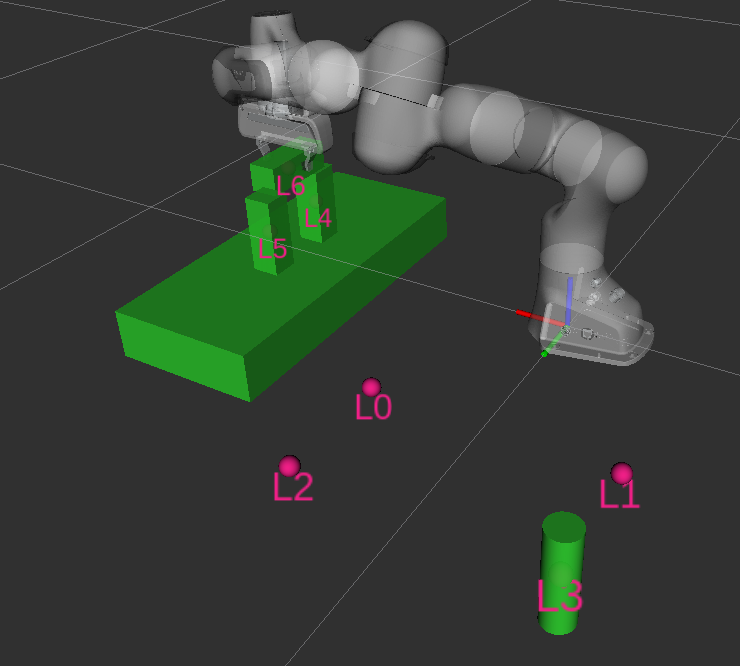}
        \label{fig: manip2}
    \end{subfigure}%
    \caption{
    Left: Initial setup.
    Right: $\sigma_{adm}$ for $\B = 13$. Task: $\phi=\bigl(p_{04} \wedge p_{15} \wedge \Next(p_{26} \vee p_{36})\bigr)$ where $p_{ij}$ is block $i$ at loc $j$. Blocks 0 to 2 are boxes and block 3 is a cylinder.
    Supplementary Video: \url{https://youtu.be/t0TMOC_PrNk}
    }
    \label{fig: new_manip_exp}
\end{figure}

\paragraph{Manipulator Domain.}
This domain considers a robotic arm (Sys) operating in presence of a human (Env), as described in \cite{muvvala2022regret}. The task for the robot is to build an arch with two boxes as support and one block (cylinder or box) on top as shown in Fig. \ref{fig: new_manip_exp}. The human can choose to either move a block back to the initial position or not intervene.
The robot's actions are \texttt{transit}, \texttt{grasp}, \texttt{transfer}, and \texttt{release} with unit cost for each action, except for transferring the cylinder, which costs two units. 

A winning strategy from the initial state does not exist as the human can always undo robot actions. However, an admissible strategy $\sigma_{adm}$ that solves Problem \ref{prob: problem_1} exists. Here, $\sigma_{adm}$ continually attempts to build the arch with either a box or the cylinder on top. Interestingly, a subset of $\sigma_{adm}$, $\wcoop$ strategies, commit to building an arch with the box on top as it ensures the least payoff among all $\wco$ strategies. Note that synthesizing $\wcoop$ is easy as it can be constructed from $\wco$ and $\coop$ (Line 2 of Alg. \ref{algo: pseudo_naive_adm}). Thus, $\sigma_{adm}$ allows the robot to operate beyond the winning region while more desirable (less costly) behaviors can also be extracted via $\wcoop$ strategies.

\textit{Comparison against Best-Effort (BE): } BE strategies build arch with either a box or cylinder on top, whereas $\wcoop$ strategies use a box on top, which is more desirable.

\section{Conclusion}

This paper relaxes the requirement of winning strategies in quantitative, reachability games using the notion of admissibility. While we show that admissible strategies are desirable in such settings, they are hard to synthesize due to their history-dependence. We show that such strategies can produce overly optimistic behaviors, and propose admissible winning strategies to mitigate them. We specifically show that admissible winning strategies are appropriate for robotics applications and their synthesis does not require more effort than synthesizing admissible strategies. Future work should explore classes of admissible winning strategies that are more risk-averse, e.g., regret-minimizing admissible strategies.

\section*{Acknowledgments}

This work was supported in part by the University of Colorado Boulder and by Strategic University Research Partnership (SURP) grants from the NASA Jet Propulsion Laboratory (JPL) (RSA 1688009 and 1704147). The authors thank the reviewers for their constructive feedback.

\bibliographystyle{named}
\bibliography{references}

\clearpage
\appendix
\section{Appendix - Additional Experiments}
\label{app: exp}
\begin{figure}[t!]
    \centering
    \begin{subfigure}[b]{0.48\columnwidth}
        \includegraphics[width=0.99\linewidth]{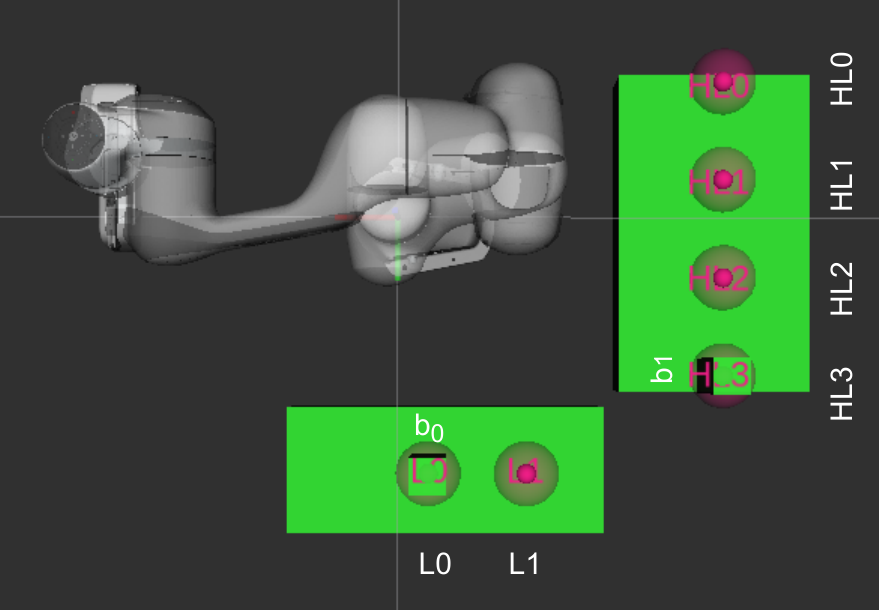}
        \label{fig: manip3}
    \end{subfigure}%
    ~
    \begin{subfigure}[b]{0.452\columnwidth}
        \includegraphics[width=0.99\linewidth]{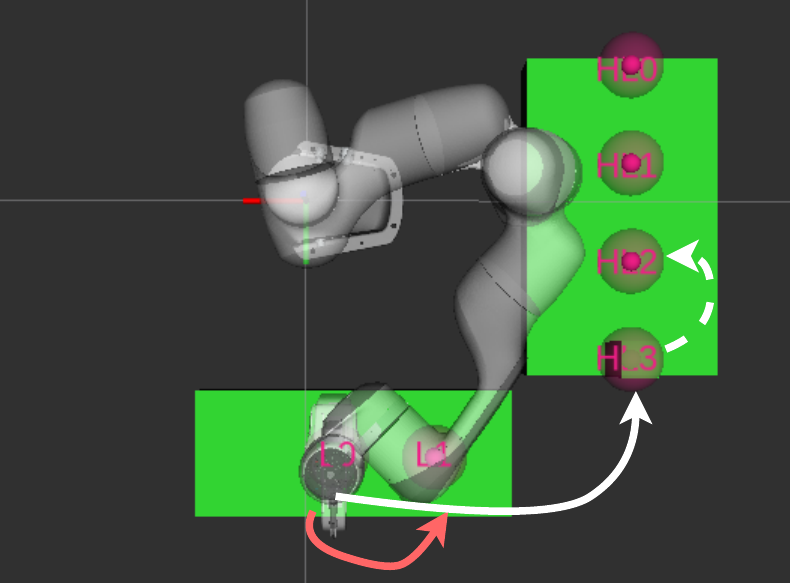}
        \label{fig: manip4}
    \end{subfigure}%
    \caption{
    Left: Initial setup.
    Right: Strategy $\sigma^{win}_{adm}$ (red) and $\sigma_{adm}$ (solid white). 
    $\phi = (\neg p_{09}) \Until (p_{01} \vee (p_{06} \wedge p_{17}))$, where $p_{ij}$ is box $i$ in loc $j$ where loc 9 is HL0; loc 6 \& 7 are HL3 \& HL2, respectively.}
    \label{fig: manip_exp}
\end{figure}

\paragraph{Manipulator Example 2.}

This domain again considers a robotic arm (Sys) operating in presence of a human (Env). 
Similar to the manipulator example considered in the Results section. In this example, we illustrate the difference in the emergent behavior when an admissible winning strategy as well as an admissible strategy exists.

In Fig. \ref{fig: manip_exp}, only Sys can manipulate boxes at location ``L\#" (bottom) while both players can moves boxes at ``HL\#" (right) locations. Fig. \ref{fig: manip_exp}-left shows the initial setup: $b_0$ at L0 and $b_1$ and HL3. Objective for Sys is: either $b_0$ at L1 or $b_0$ at HL3 and $b_1$ at HL2 and never $b_0$ at HL0. It is more expensive for Sys to operate at L\# locations than to operate at HL\# locations. Note that a winning strategy from the initial state exists. 

Under $\sigma^{win}_{adm}$, Sys player places $b_0$ at L1 as this always ensure task completion. We note that strategy that places $b_0$ at HL0 will never be admissible. Under $\sigma_{adm}$, Sys player transits to box $b_0$ and grasps it. As the human cooperates, Sys finishes the task by placing $b_0$ at HL3. There exists human actions that can force Sys player to violate the task once $b_0$ is in HL\#. Thus, $\sigma_{adm}$ generates optimistic behaviors.

\paragraph{Computation Results.} In Fig. \ref{fig: comp_time}, we report computation times that empirically shows PTIME complexity for Alg. \ref{algo: pseudo_naive_adm} for fixed $\B$ for the Manipulator case-study. We consider 4 different specifications:  $\phi_1 = F(p_{06})$; $\phi_2 = F(p_{06} \wedge F(p_{07}))$; $\phi_3 = F(p_{06} \wedge F(p_{07} \wedge F(p_{08})))$; $\phi_4 = F(p_{06} \wedge F(p_{07} \wedge F(p_{08} \wedge F(p_{09}))))$.

\section{Appendix - Proofs and Derivation}

\paragraph{Derivation of Eq. \eqref{eq: adm_eq}.}

We can simplify Eqs.~\eqref{eq: adm_eq_1} and \eqref{eq: adm_eq_2}. By definition, $\forall h \in \plays^{v}$, 
$$\aVal(h, \sigma) \geq \aVal(h) \geq \cVal(h, \sigma).$$

Given strategy $\sigma$, in Eq. \eqref{eq: adm_eq_2}, it needs to satisfy the condition $\aVal(h, \sigma) \leq \aVal(h)$ for $\sigma$ to be adversarial-cooperative. Thus, a strategy $\sigma$ that satisfies $\aVal(h, \sigma) > \aVal(h)$ is not a valid strategy. Hence, $\aVal(h, \sigma) = \aVal(h) > \cVal(h, \sigma)$ or $\aVal(h, \sigma) = \aVal(h) = \cVal(h, \sigma)$. The former is subsumed by Eq. \eqref{eq: adm_eq_1} while the latter is not. Thus by simplifying the equations, we get Eq.~\eqref{eq: adm_eq}.

\subsection{Necessary and Sufficient Conditions for weakly dominating strategy}

A strategy $\sigma$ is not admissible if there exists another strategy $\sigma'$ that weakly dominates it. Formally, $\sigma' \succ \sigma$ if, and only if, $\forall h \in \plays^{v_0}(\sigma)$, the following holds:
\begin{subequations}
    \begin{align}
        & \big( \cVal(h, \sigma) \geq \aVal(h) \wedge \aVal(h, \sigma) > \aVal(h) \big) \vee \label{eq: dom_eq_1} \\
        & \big( \aVal(h) = \aVal(h, \sigma) = \cVal(h, \sigma) \; \wedge \notag \\
        & \hspace{4 cm} \acVal(h) < \aVal(h) \big)
        \label{eq: dom_eq_2}
    \end{align}
    \label{eq: dom_eq}
\end{subequations}

\begin{lemma}
    Strategy $\sigma$ is weakly dominated if, and only if, negation of Eq.~\eqref{eq: adm_eq} holds.
    \label{lem: weak_dom}
\end{lemma}

For history $h$, Eq. \eqref{eq: dom_eq_1} implies there exists another strategy $\sigma'$ that is $\wco$ and thus has a lower adversarial value than $\aVal(h, \sigma)$. Further, the cooperative value under strategy $\sigma$ is either worse or equal to $\aVal(h)$. Thus, $\sigma'$ weakly dominates $\sigma$. In Eq. \eqref{eq: dom_eq_2}, $\sigma$ is not $\wcoop$ as there exists $\sigma'$ with lower cooperative optimal payoff while being worst-case optimal. Thus, $\sigma'$ weakly dominates $\sigma$.

\begin{figure}[t!]
    \centering
    \includegraphics[width=\linewidth]{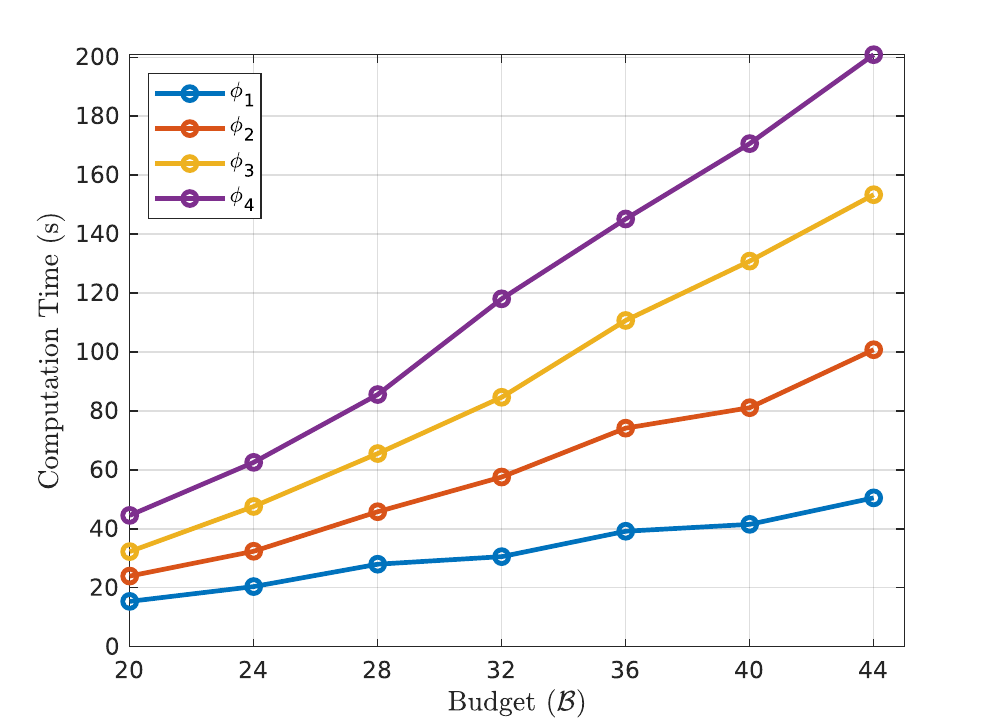}
    \caption{PTIME complexity for Alg. \ref{algo: pseudo_naive_adm} for fixed  Computation Time vs Budget for fixed $\G$ and varying $\phi$. See Section~\ref{app: exp} for details.}
    \label{fig: comp_time}
\end{figure}

\subsection{Proof of Lemma \ref{lem: quant_adm_not_val_pres}}

\begin{proof}
    For the proof, it is sufficient to show a counterexample.
    In Example \ref{ex: adm_qual_qaunt_diff}, for $\G$ in Fig. \ref{fig: cex_adm_val_not_preserving}, all states $v \in V\setminus\{v_2, v_3\}$ 
    belong to $V_{win}$ and hence $\sVal(v) = 1$. States $v_2$ and $v_3$ belong to $V_{pen}$ thus $\sVal(v_2) = \sVal(v_3) = 0$. Both, $\sigma_2$ and $\sigma_3$ are admissible strategies as shown in Example \ref{ex: adm_qual_qaunt_diff}.
    But, under $\sigma_3$, the possible plays are, $v_0(v_2 v_3)^{\omega}$, $v_0(v_2 v_3)^* v_6$, and $v_0 v_2 v_6$. Thus, under $\sigma_3$, there exists a play that starts in the winning region and does \emph{not} stay in the winning region. For plays induced by $\sigma_2$, all states in all the plays belong to the winning region. Hence, there exists an admissible strategy that is not value-preserving.
\end{proof}

\subsection{Proof of Lemma \ref{lem: scoop_proof}}

\begin{proof}
     \begin{figure}[h]
    \centering
    \resizebox{0.45\linewidth}{!}{%
        \centering
        \begin{tikzpicture}
            [->,>=stealth',shorten >=1pt,shorten <=1pt,auto,node distance=1cm,
            every loop/.style={looseness=6},
            initial text={},
            el/.style={font=\scriptsize},
            every fit/.style={draw,densely dotted,rectangle},
            inner sep=2mm,
            decoration = {snake,pre length=3pt,post length=7pt},
            loopright/.style={loop,looseness=6,out=-45, in=45},
            loopleft/.style={loop,looseness=6,out=135, in=225},
            loopabove/.style={loop,looseness=6,out=45, in=135},
            loopbelow/.style={loop,looseness=6,out=-135, in=-45},]
        \tikzstyle{every state}=[node distance=1.4cm,minimum size=7mm, inner sep=1pt];
        \node[state] at (0,0) (v0){$v_s$};
        \node[state, left of=v0, node distance=2.4cm,] (v4) {$v_0$};
        \node[state, rectangle, node distance=2.4cm, above right of=v0] (v1){$\sigma(h)$};
        \node[state, rectangle, node distance=2.4cm, below right of=v0] (v2) {$\sigma'(h)$};
        \path[-latex'] 
          (v4) edge[decorate] node[above] {$h$} (v0)
          (v0) edge[decorate] node {$\sigma$ is $\scoop$}(v1)
          (v0) edge[decorate] node[left] {$\sigma'$ is not $\scoop$}(v2);
        \end{tikzpicture}
        }
    \caption{$\scoop$ proof example}
    \label{fig: sc_proof}
\end{figure}

    Let $\sigma$ be $\scoop$. Assume there exists $\sigma' \neq \sigma$ that is compatible with history $h$, $\last(h) = v_s$, and ``splits" at $v_s$ as shown in Fig. \ref{fig: sc_proof}. Thus, $\sigma(h) \neq \sigma'(h)$. We note that only two cases for a strategy are possible, i.e., it is either $\scoop$ or not. 
    Further, let's assume that $\sigma'$ weakly dominates $\sigma$. We prove that $\sigma'$ can not weakly dominate $\sigma$. 
    
    \paragraph{Case I}$\cVal(h) < \aVal(h)$:  As $\sigma'$ is not $\scoop$ this implies $\cVal(h, \sigma') \geq \aVal(h)$ and as $\cVal(h, \sigma) \leq \aVal(h, \sigma)$ for any $\sigma \in \Sigma$, we get 
    $$\aVal(h, \sigma') \geq \cVal(h, \sigma') \geq \aVal(h) > \cVal(h, \sigma).$$
    
    On simplifying, we get $\aVal(h, \sigma') \geq \aVal(h) > \cVal(h, \sigma)$. This statement implies that there exists $\tau \in \Tau$ such that $\Val(h \cdot \play^{h}(\sigma', \tau)) > \Val(h \cdot \play^{h}(\sigma, \tau))$. Note that since, $\Val(h \cdot 
    \play^{h}(\sigma, \tau)) = \Val(h') + \Val(\play^{v_s}(\sigma, \tau))$ where $h' = v_0 \ldots v_{|h| - 2}$. Thus, we get $\Val(\play^{v_s}(\sigma', \tau)) > \Val(\play^{v_s}(\sigma, \tau))$. This contradicts the assumption that $\sigma'$ dominates $\sigma$ as $\sigma'$ should always have a payoff that is equal to or lower than $\sigma$. 
    
    \paragraph{Case II}$\cVal(h) = \aVal(h)$: For this case we get,
    $$\aVal(h, \sigma') \geq \cVal(h, \sigma') \geq \aVal(h) = \cVal(h) = \cVal(h, \sigma).$$ 
    
    This implies that $\cVal(h, \sigma') \geq \cVal(h)$. But, since $\sigma'$ dominates $\sigma$, there should exist a strategy $\tau \in \Tau$ under which $\sigma'$ does strictly better than $\sigma$. Since, $\cVal(h, \sigma') \geq \cVal(h) \forall \tau \in \Tau$, it implies that there does not exists a payoff $\Val(\play^{v_s}(\sigma', \tau))$ that has a payoff strictly less than $\cVal(h, \sigma) = \cVal(h)$. Thus, $\sigma'$ does not dominate $\sigma$. 
    
     We can repeat this for all histories $h$ in $\plays^{v_0}(\sigma)$.
    Hence, every $\sigma$ that is $\scoop$ is admissible.
\end{proof}

\subsection{Proof of Theorem \ref{thm: scoop_memory}}
\begin{proof}
    To prove this it is sufficient to show an example.
    Consider the game $\G$ in Fig. \ref{fig: local_conds_not_sufficient}. Let us consider strategies: $\sigma_1:(v_0 \to v_1), (v_3 \to v_6), (v_7 \to v_8)$, $\sigma_2:(v_0 \to v_1), (v_3 \to v_6),(v_7 \to v_9)$, $\sigma_3: (v_0 \to v_2), (v_4 \to v_6),(v_7 \to v_9)$, and $\sigma_4: (v_0 \to v_2), (v_4 \to v_6),(v_7 \to v_8)$. Note, $\sigma_1, \sigma_2$ and $\sigma_3, \sigma_4$ only differ at state $v_7$ and $v_0$. 
        
    For the play induced by $\sigma_1$, $\cVal(v_0, \sigma_1) < \aVal(v_0)$, $\cVal(v_0 v_1 v_3, \sigma_1) < \aVal(v_0 v_1 v_3)$, $\cVal(v_0 v_1 v_3 v_6 v_7, \sigma_1) < \aVal(v_0 v_1 v_3 v_6 v_7)$. Thus, $\cVal(h, \sigma_1) < \aVal(h)$ for all the histories $h$ compatible with $\sigma_1$. For the play induced by $\sigma_2$,  $\cVal(v_0, \sigma_2) < \aVal(v_0)$, but $\cVal(v_0 v_1 v_3, \sigma_2) > \aVal(v_0 v_1 v_3)$. Notice that $\cVal(v_0 v_1 v_3 v_6, \sigma_2) < \aVal(v_0 v_1 v_3 v_6)$ and thus we need to check admissibility for all histories $h$ compatible with $\sigma$. Hence, for history $h := v_0 v_1 v_3$, $\sigma_2$ is not $\scoop$. If $h = v_0 v_2 v_4 v_6$, then strategy $\sigma_3$ and $\sigma_4$ are both admissible. Thus, $\sigma_3$ and $\sigma_4$ are both $\scoop$ strategies. This proves that $\scoop$ strategies are history-dependent.
\end{proof}

\begin{figure}[t!]
    \centering
    \resizebox{1.2\linewidth}{!}{%
        \centering
        \begin{tikzpicture}
            [->,>=stealth',shorten >=1pt,shorten <=1pt,auto,node distance=1cm,
            every loop/.style={looseness=6},
            initial text={},
            el/.style={font=\scriptsize},
            every fit/.style={draw,densely dotted,rectangle},
            inner sep=2mm,
            loopright/.style={loop,looseness=6,out=-45, in=45},
            loopleft/.style={loop,looseness=6,out=135, in=225},
            loopabove/.style={loop,looseness=6,out=45, in=135},
            loopbelow/.style={loop,looseness=6,out=-135, in=-45},]
        \tikzstyle{every state}=[node distance=1.4cm,minimum size=7mm, inner sep=1pt];
        \node[state, initial] at (0,0) (v0){$v_0$};
        \node[state, rectangle, right of=v0] (v1){$v_1$};
        \node[state, rectangle, below of=v1] (v2) {$v_2$};
        \node[state, node distance=1.4cm, right of=v1] (v3) {$v_3$};
        \node[state, right of=v2] (v4) {$v_4$};
        \node[state, rectangle, node distance=1.4cm, right of=v3] (v6) {$v_6$};
        \node[state, rectangle, node distance=1.4cm, above right of=v3] (v5) {$v_5$};
        \node[state, node distance=1.4cm, right of=v6] (v7) {$v_7$};
        \node[state, rectangle, node distance=1.4cm, right of=v7] (v9) {$v_9$};
        \node[state, rectangle, node distance=1.4cm, above right of=v7] (v8) {$v_8$};
        \node[state, rectangle, node distance=2.4cm, above of=v0] (v10) {$v_{10}$};
        \node[state, node distance=1.4cm, above right of=v10] (v11) {$v_{11}$};
        \node[state, node distance=1.4cm, below of=v2] (v12) {$v_{12}$};
        \node[state, node distance=1.4cm, above of=v1] (v13) {$v_{13}$};
        \node[state, accepting, node distance=1.4cm, inner sep=1pt, above left of=v5](val3) {};%
        \node[state, accepting, node distance=1.4cm, inner sep=1pt, above right of=v5](val4) {}; %
        \node[state, accepting, inner sep=1pt, above of=v8](val5) {};%
        \node[state, accepting, inner sep=1pt, right of=v8](val6) {};%
        \node[state, accepting, inner sep=1pt, below of=v9](val7) {};%
        \node[state, accepting, inner sep=1pt, right of=v9](val8) {};%
        \node[above left of=v0](lab0) {$(\textcolor{blue}{5}, \textcolor{red}{\infty})$};
        \node[node distance=0.6cm, below of=v1](lab1) {$(\textcolor{blue}{4}, \textcolor{red}{\infty})$};
        \node[left of=v2](lab2) {$(\textcolor{blue}{4}, \textcolor{red}{\infty})$};
         \node[above left of=v10](lab9) {$(\textcolor{blue}{\infty}, \textcolor{red}{\infty})$};
        \node[node distance=0.7cm, above left of=v3](lab3) {$(\textcolor{blue}{4}, \textcolor{red}{6})$};
        \node[below right of=v4](lab4) {$(\textcolor{blue}{4}, \textcolor{red}{11})$};
        \node[node distance=0.8cm, above of=v5](lab5) {$(\textcolor{blue}{4}, \textcolor{red}{5})$};
        \node[node distance=0.8cm, below of=v6](lab6) {$(\textcolor{blue}{3}, \textcolor{red}{10})$};
        \node[node distance=0.8cm, below of=v7](lab7) {$(\textcolor{blue}{3}, \textcolor{red}{10})$};
        \node[above left of=v8](lab8) {$(\textcolor{blue}{2}, \textcolor{red}{9})$};
        \node[node distance=0.8cm, below right of=v9](lab9) {$(\textcolor{blue}{7}, \textcolor{red}{11})$};
        \node[right of=v11](lab11) {$(\textcolor{blue}{\infty}, \textcolor{red}{\infty})$};
        \node[node distance=0.6cm, below of=v12](lab12) {$(\textcolor{blue}{\infty}, \textcolor{red}{\infty})$};
        \node[node distance=0.6cm, above of=v13](lab13) {$(\textcolor{blue}{\infty}, \textcolor{red}{\infty})$};
        
        \path[-latex'] (v0) edge node {1} (v1)
        (v0) edge node[above] {\xmark} (v10)
        (v0) edge node[below] {1} (v2)
        (v1) edge (v3)
        (v2) edge (v4)
        (v3) edge node {1} (v5)
        (v3) edge node {1} (v6)
        (v4) edge node[below] {1} (v6)
        (v6) edge (v7)
        (v7) edge node {1} (v8)
        (v7) edge[dashed] node[below] {?} node[above] {1} (v9)
        (v5) edge node[below left] {4} (val3)
        (v5) edge node[below right] {5} (val4)
        (v8) edge node[right] {2} (val5)
        (v8) edge node {9} (val6)
        (v9) edge node[left] {11} (val7)
        (v9) edge node {7} (val8)
        (v10) edge[bend left] (v11)
        (v11) edge[bend left] node {1} (v10)
        (v2) edge[bend left] (v12)
        (v12) edge[bend left] node {1} (v2)
        (v1) edge[bend left] (v13)
        (v13) edge[bend left] node[left] {1} (v1)
        (val3) edge[loop] (val3)
        (val4) edge[loop] (val4)
        (val5) edge[loop] (val5)
        (val6) edge[loop] (val6)
        (val7) edge[loopbelow] (val7)
        (val8) edge[loopright] (val8);
        \end{tikzpicture}
        }
    \caption{Illustrative example $\G$: for all $h \in \plays^{v_0}$ we define $\Val < \infty$ if a play reaches a goal (double circle around the state) state. The cost of the Env action is zero, while the cost of the Sys actions are shown along the edges.
    The values in blue and red are $\cVal(v)$ and $\aVal(v)$.
    }
    \label{fig: local_conds_not_sufficient}
\end{figure}
\subsection{Proof of Lemma \ref{lem: wcoop_proof}}
\begin{proof}
    \begin{figure}[h]
    \centering
    \resizebox{0.45\linewidth}{!}{%
        \centering
        \begin{tikzpicture}
            [->,>=stealth',shorten >=1pt,shorten <=1pt,auto,node distance=1cm,
            every loop/.style={looseness=6},
            initial text={},
            el/.style={font=\scriptsize},
            every fit/.style={draw,densely dotted,rectangle},
            inner sep=2mm,
            decoration = {snake,pre length=3pt,post length=7pt},
            loopright/.style={loop,looseness=6,out=-45, in=45},
            loopleft/.style={loop,looseness=6,out=135, in=225},
            loopabove/.style={loop,looseness=6,out=45, in=135},
            loopbelow/.style={loop,looseness=6,out=-135, in=-45},]
        \tikzstyle{every state}=[node distance=1.4cm,minimum size=7mm, inner sep=1pt];
        \node[state] at (0,0) (v0){$v_s$};
        \node[state, left of=v0, node distance=2.4cm,] (v4) {$v_0$};
        \node[state, rectangle, node distance=2.4cm, above right of=v0] (v1){$\sigma(h)$};
        \node[state, rectangle, node distance=2.4cm, below right of=v0] (v2) {$\sigma'(h)$};
        \path[-latex'] 
          (v4) edge[decorate] node[above] {$h$} (v0)
          (v0) edge[decorate] node {$\sigma$ is $\wcoop$}(v1)
          (v0) edge[decorate] node[left] {$\sigma'$ is not $\wcoop$}(v2);
        \end{tikzpicture}
        }
    \caption{$\wcoop$ proof example}
    \label{fig: wcoop_proof}
\end{figure}

    Let $\sigma$ be a $\wcoop$. Assume there exists $\sigma' \neq \sigma$ that is compatible with history $h$, $\last(h) = v_s$, and ``splits" at $v_s$ as shown in Figure \ref{fig: wcoop_proof}. Thus, $\sigma(h) \neq \sigma'(h)$. We note that only two cases for a strategy are possible, i.e., it is either $\wcoop$ or not.
    Further, let's assume that $\sigma'$ weakly dominates $\sigma$. We prove that $\sigma'$ can not weakly dominate $\sigma$.
    
    We first note that, by definition, we have $\aVal(h, \sigma) \geq \aVal(h)$ for all $\sigma \in \Sigma$. If $\aVal(h, \sigma) = \aVal(h)$ then it is $\wco$ else it is not. For strategy $\sigma'$ compatible with $h$, 
    \begin{align*}
        & \neg \wcoop \implies \bigg( \aVal(h, \sigma') > \aVal(h) \vee \\
        & \hspace{0.5cm} \big( \aVal(h, \sigma') = \aVal(h) \wedge \cVal(h, \sigma') \neq \acVal(h)  \big) \bigg).
    \end{align*}
    
    \paragraph{Case I}$\sigma'$ is not $\wco$: This implies that $\aVal(h) \neq \aVal(h, \sigma')$. As $\sigma'$ is not $\wco$, it implies that there exists an adversarial strategy $\tau \in \Tau$ for which the payoff $\Val(h \cdot \play^{v_s}(\sigma', \tau)) > \Val(h \cdot \play^{v_s}(\sigma, \tau))$. Since, $\Val(h \cdot \play^{v_s}(\sigma, \tau)) = \Val(h') + \Val(\play^{v_s}(\sigma, \tau))$, where $h' = v_0 \dots v_{|h| - 2}$. Thus, we get $\Val(\play^{v_s}(\sigma', \tau)) > \Val(\play^{v_s}(\sigma, \tau))$. This contradicts our assumption that $\sigma' \succ \sigma$ as $\sigma'$ should never do worse than $\sigma$.
    
    \paragraph{Case II}$\sigma'$ is $\wco$ but $\cVal(h, \sigma') \neq \acVal(h)$: For this case we have $\aVal(h) =\aVal(h, \sigma') \geq \cVal(h, \sigma') \wedge \cVal(h, \sigma') > \cVal(h, \sigma) = \acVal(h)$. This implies that there exists a play under $\sigma$ such that payoff $\Val(\play^{v_s}(\sigma, \tau)) < \Val(\play^{v_s}(\sigma', \tau))$. This contradicts our statement as $\sigma \succ \sigma'$. 
    
    Thus, every strategy that is $\wcoop$ is admissible. 
\end{proof}
\definecolor{lightcyan}{RGB}{224,255,255}
\definecolor{lightgreen}{RGB}{144,238,144}
\definecolor{lavender}{RGB}{230,230,250}
\definecolor{salmon}{RGB}{250,128,114}
\definecolor{beige}{RGB}{245,245,220}

\begin{algorithm}[t!]
    \caption{Value Iteration}
    \label{algo: vi_algo}
    \SetKwInOut{Input}{Input}
    \SetKwInOut{Output}{Output}
    \SetKwComment{Comment}{/* }{ */}
    \Input{Game $\G$}
    \Output{Optimal $\sigma_{win}$, $V_{win}$}
    \Comment{$W$ - function that maps state to values in $\mathbb{R} \cup \infty$}
    $W_{win} \leftarrow \infty; \; W_{win}' \leftarrow \infty;$ \\
    $\sigma_{win} \leftarrow \emptyset;$ \\
    \lForAll{$v$ in $V_f$}{ $W_{win}(v) \gets 0$}
    \Comment{Compute Winning region and Winning strategy}
    \While{$W_{win}' \neq W_{win}$}{
        $W_{win}' = W_{win}$ \\
        \For{$v \in V \backslash V_f$}{
            \Comment{For $\aVal$ Value Iteration}
            \colorbox{lightgray}{$W_{win}(v) = \max(C(v, a) + W'_{win}(v')) \quad \text{if}\; v \in V_e $} \label{line: max_computation} \\
            \Comment{For $\cVal$ Value Iteration}
            \colorbox{lightcyan}{$W_{win}(v) = \min(C(v, a) + W'_{win}(v')) \quad \text{if}\; v \in V_e $} \label{line: min_computation} \\
            $W_{win}(v) = \min(C(v, a) + W'_{win}(v')) \quad \text{if}\; v \in V_s$ \\
            $\sigma_{win}(v) = \argmin_{a}(C(v, a) + W'_{win}(v')) \quad \text{if}\; v \in V_s$ \\
        }
    }
    $V_{win} \gets  \{v \; |\; W_{win}(v) \neq \infty \; \forall v \in V \}$ \\
    \KwRet{$\sigma_{win}, V_{win}$}
\end{algorithm}

\subsection{Proof of Theorem \ref{thm: wcoop_memoryless}}

\begin{proof}
    We begin by showing that memoryless strategies are sufficient for optimal $\cVal$ and $\aVal$.
    For $\coop$ strategies, the game can be viewed as a single-player as both players are playing cooperatively. Thus, synthesizing a witnessing strategy $\sigma$ for $\coop$ reduces to the classical shortest path problem in a weighted graph. This can be solved in polynomial time using Dijkstra's and Flyod Warshall's algorithms when the weights are non-negative and arbitrary \cite{floyd1962algorithm,Mehlhorn2008algo}. 
    
    In the adversarial setting, \citeauthor{khachiyan2008short} show that memoryless strategies are sufficient for two-player, non-negative weights scenario. 
    
    Next, we show that for a given history $h$, memoryless strategies are sufficient to be optimal adversarial-cooperative. We show this by proving that $\cVal(\last(h), \sigma) = \acVal(h)$.  Given history $h$, we define $\aVal(\G, h)$ and $\cVal(\G, h)$ to be the optimal adversarial and cooperative value in $\G$. We extend the definition to $\acVal(\G, h)$ accordingly. We denote by $\Sigma(\G)$ and $\Tau(\G)$ the set of all valid strategies in $\G$.
    
    Given $\G$ and history $h$, we define $\G'$ to be the subgame such that the initial state in $\G'$ is $\last(h)$ and every state in $\G'$ satisfies condition $\aVal(\G, v) \leq \aVal(\G, h)
    \; \forall v \in V$. We note that $\aVal(\G, v_f) = \cVal(\G, v_f) = 0$ and $\aVal(\G, v) \geq \cVal(\G, v) \neq 0 \; \forall v \notin V_f$. Thus, every subgame $\G'$ will include the goal states $v \in V_f$.
    As $\G'$ is the subgame of $\G$ the strategies in $\G'$ can be uniquely mapped to $\G$. Notice that $\Sigma(\G') \subseteq \Sigma(\G)$ where $\Sigma(\G')$ is the set of all valid  strategies in $\G'$. The weights of the edges in $\G'$ are the same as $\G$.
    
    First, we will show that every state that is reachable in $\G'$ has at least one outgoing edge. Next, we will show that the cooperative value of $\G'$ at the initial state $\last(h)$ is exactly the optimal adversarial-cooperative value of $h$ in $\G$, i.e., 
    $\cVal(\G', \last(h)) = \acVal(\G', \last(h)) = \acVal(\G, h)$. As memoryless strategies are sufficient for optimal cooperative value and strategies in $\G'$ can be uniquely mapped to $\G$, we will conclude that memoryless strategies are sufficient to be optimal $\wcoop$ strategy.

    \paragraph{Every reachable vertex in $\G'$ has at least one outgoing edge:} 
    By Thm. \ref{thm: adm_exist}, given $\G$, history $h$, there always exists a witnessing strategy $\sigma$ such that $\aVal(\G, h, \sigma) \leq \aVal(\G, h)$. By construction, we have that $\aVal(\G', v) \leq \aVal(\G, h)$. If $v$ is a Sys player state in $\G'$, then there always exists an action that corresponds to $\sigma$ such that $\sigma(v)$ is a valid edge in $\G'$. If $v$ is a state $\G'$ and belongs to $V_e$ then, by definition, for all the valid actions $a_e$,  $\aVal(\G', v') \leq \aVal(\G', v)$ where $v' = \delta(v, a_e)$. Thus, all edges from $v \in V_e$ will be present in $\G'$. Thus, any vertex reachable from $\last(h)$ has at least one outgoing edge in $\G'$.
    \paragraph{For any history $h$, $\cVal(\G',\last(h)) = \acVal(\G, h):$} We first observe that any state whose $\aVal(\G, v) > \aVal(\G, h)$, will not be present in $\G'$. Thus, for any prefix $h'$ of the plays starting from $\last(h)$ in $\G'$, will have their $\aVal(\G', \last(h')) \leq \aVal(\G, h)$. Further, all the plays that satisfy the $\acVal(\G, h)$ condition will be captured in $\G'$ too. Thus, $\acVal(\G', \last(h)) = \acVal(\G, h)$.
    Now, consider alternate strategy $\sigma' \neq \sigma$ where $\sigma', \sigma \in \Sigma(\G')$. We note that both $\sigma$ and $\sigma'$ are valid strategies in $\G$ too.
    Here $\sigma$ is a strategy such that $\aVal(\G', v, \sigma) \leq \aVal(\G, h)$.
    $\sigma'$ can be $\coop$ optimal in $\G$, i.e., $\cVal(\G, h, \sigma') < \acVal(\G, h)$, but $\aVal(\G, h, \sigma') > \aVal(\G, h)$ by definition of adversarial-cooperative value. As, $\aVal(\G, h, \sigma') > \aVal(\G, h)$, the state $\sigma'(h)$ will not be present in $\G'$. Inductively, for every history $h'$, starting from $\last(h)$ in $\G'$, we have that $\sigma'(h')$ will not be present in $\G'$. As every state that belongs $\sigma'(h')$ does not exists in $\G'$, $\cVal(\G', \last(h')) \not<  \acVal(\G, h)$. Further, $\cVal(\G', \last(h')) \not>  \acVal(\G, h)$ as $\cVal(\G', \last(h')) > \acVal(\G, h) \implies \aVal(\G', \last(h')) > \acVal(\G, h)$. This contradicts $\G'$ by construction. Thus, $\cVal(\G',\last(h)) = \acVal(\G, h)$. Since, $\acVal(\G', \last(h)) = \acVal(\G, h)$ we get $\cVal(\G',\last(h)) = \acVal(\G', \last(h)) = \acVal(\G, h)$. As memoryless strategies are sufficient for witnessing optimal cooperative value, this implies that memoryless strategies are sufficient for a strategy to be adversarial-cooperative optimal.
    
    Since memoryless strategies are sufficient for $\wco$ conditions and adversarial-cooperative optimal conditions, this implies that memoryless strategies are sufficient for $\wcoop$ strategies.
\end{proof}

\subsection{Proof of Theorem \ref{thm: adm_str}}
\begin{proof}
    
\end{proof}

\subsection{Proof of Lemma \ref{lem: weak_dom}}

\begin{proof}
 We need to prove that $\neg$ Eq. \eqref{eq: adm_eq} $\equiv$ Eq. \eqref{eq: dom_eq}. To prove this, we prove that $\neg$ Eq. \eqref{eq: dom_eq} $\equiv$ Eq. \eqref{eq: adm_eq}.

We first observe that,
$$\neg(\text{Eq. \eqref{eq: dom_eq_1}} \vee \text{Eq. \eqref{eq: dom_eq_2}}) = \neg \text{Eq. \eqref{eq: dom_eq_1}} \wedge \neg \text{Eq. \eqref{eq: dom_eq_2}}.$$
where,
\begin{equation*}
    \neg \text{Eq. \eqref{eq: dom_eq_1}} = \cVal(h, \sigma) < \aVal(h) \vee \aVal(h, \sigma) \leq \aVal(h)
\end{equation*}

By rewriting the equation, we get 
\begin{multline*}
\cVal(h, \sigma) < \aVal(h) \\
\vee \bigl( \cVal(h, \sigma) \geq \aVal(h) \wedge \aVal(h, \sigma) \leq \aVal(h) \bigr)     
\end{multline*}

Rearranging the term inside the bracket, we get, $\aVal(h, \sigma) \leq \aVal(h) \leq \cVal(h, \sigma)$. Since, $\aVal(h, \sigma) \geq \cVal(h, \sigma)$ for any $h$, we get $\aVal(h, \sigma) = \aVal(h) = \cVal(h, \sigma)$. Thus, $\neg \text{Eq. \eqref{eq: dom_eq_1}}$ is
\begin{equation}
      \cVal(h, \sigma) < \aVal(h) \vee \bigl(\aVal(h, \sigma) = \aVal(h) = \cVal(h, \sigma)\bigr)
      \label{eq: inter_step_1}
\end{equation}

Now, let us take the negation of Eq. \eqref{eq: dom_eq_2}. Thus,
\begin{multline*}
    \neg \left[ \aVal(h) = \aVal(h, \sigma) = \cVal(h, \sigma) \right] \\
    \vee \neg \left[\acVal(h) < \aVal(h) \right]
\end{multline*}

On simplifying  $\neg \text{Eq. \eqref{eq: dom_eq_2}}$, we get
\begin{align}
    \left[\aVal(h) = \aVal(h, \sigma) = \cVal(h, \sigma) \right] & \notag \\
    \implies & \neg \left[\acVal(h) < \aVal(h) \right] \notag \\
    \aVal(h) = \aVal(h, \sigma) = \cVal(h, \sigma) & \implies \acVal(h) \geq \aVal(h)
    \label{eq: inter_step_2}
\end{align}

By substituting values for $\neg \text{Eq. \eqref{eq: dom_eq_1}} \wedge \neg \text{Eq. \eqref{eq: dom_eq_2}}$, we get
\begin{align*}
     & \cVal(h, \sigma) < \aVal(h) \; \; \vee \\
     & \hspace{3 cm}
     \bigl(\aVal(h, \sigma) = \aVal(h) = \cVal(h, \sigma)\bigr) \; \; \wedge \\
     & \big(\aVal(h) = \aVal(h, \sigma) = \cVal(h, \sigma) \implies \\
     & \hspace{5cm} \acVal(h) \geq \aVal(h) \big)
\end{align*}
\begin{multline*}
    \cVal(h, \sigma) < \aVal(h) \vee \biggl( \aVal(h, \sigma) = \aVal(h) = \cVal(h, \sigma) \\
    \wedge \; \aVal(h) = \aVal(h, \sigma) = \cVal(h, \sigma) \\ \implies \acVal(h) \geq \aVal(h) \biggr)
\end{multline*}
On simplifying, 
\begin{multline*}
    \cVal(h, \sigma) < \aVal(h) \vee \biggl(\aVal(h) = \aVal(h, \sigma) = \cVal(h, \sigma) \\
    \implies \acVal(h) \geq \aVal(h) \biggr)
\end{multline*}
By definition, $\acVal(h) \leq \aVal(h)$. Thus, we can simplify the equation to be 
\begin{align*}
    \big(\cVal(h, \sigma) & < \aVal(h)\big) \;\; \vee\\
    \big(\aVal(h) = \aVal(h, \sigma) & = \cVal(h, \sigma) = \acVal(h)\big) 
\end{align*}
Hence, we get that $\neg \text{Eq. \eqref{eq: dom_eq}} = \text{Eq. \eqref{eq: adm_eq}}$.
\end{proof}

\begin{algorithm}[t!]
    \caption{Admissible Synthesis}
    \label{algo: naive_adm}
    \SetKwInOut{Input}{Input}\SetKwInOut{Output}{Output}
    \SetKwComment{Comment}{/* }{ */}
    \Input{Game $\G$, Budget $\B$}
    \Output{Strategy $\Sigma_{adm}$}
    \SetKwFunction{ValueIteration}{ValueIteration}
    \SetKwFunction{push}{push}
    \SetKwFunction{pop}{pop}
    \SetKwFunction{nxt}{next}
    \SetKwFunction{iter}{iter}
    \SetKwProg{try}{try}{:}{}
    \SetKwProg{catch}{catch}{:}{end}
    $\G' \gets$ Unroll $\mathcal{G}$ up until Payoff $\B$ \\
    \Comment{VI from Alg. \ref{algo: vi_algo}}
    $\aVal \gets$ min-max \ValueIteration($\G'$) \\
    $\cVal \gets$ min-min  \ValueIteration($\G'$) \\
    \For{$v \in V_s$ in $\G'$}{
        $\acVal(v) = \min \{\cVal(v') | \aVal(v') \leq \aVal(v) \; \text{where $v'$ is valid successor(s)}\}$ \\
    }
    \If{$\B <\cVal(v_0)$}
    {\KwRet{$\Sigma_{adm} := \Sigma$}}
    \Comment{\small Compute Admissible strategies-DFS}
    $h; \aValues \gets$ initialize empty stack \\
    $h.\push\big((v_0, \{\delta(v_0, a_s)\})\big)$ \\
    $\aValues.\push(\aVal(v_0))$ \\
    \While{$h \neq \emptyset$}{
        $v, \{v'\} \gets h[-1]$ \\
        \try{}{
            $v' \gets \nxt(\iter(\{v'\}))$ \\
             \If{$v \in V_s$}{
                 \lForAll{$s$ in $h$}{$hist := [s[0]]$}
                 \If{\colorbox{lightgray}{$\cVal(v') < \min\{\aValues\}$}}{
                     $h.\push\big((v', \{\delta(v', a_e)\})\big)$ \\
                     $\Sigma_{adm}: hist \to v'$ \\
                     $\aValues.\push(\aVal(v'))$ \\
                 }
                 \ElseIf{$\aVal(v) = \aVal(v') = \cVal(v') = \acVal(v)$}{
                 $h.\push\big((v', \{\delta(v', a_e)\})\big)$ \\
                 $\Sigma_{adm}: hist \to v'$ \\
                 $\aValues.\push(\aVal(v'))$ \\
                 }
             }
             \If{$v \in V_e$}{
             $h.\push\big((v', \{\delta(v', a_s)\})\big)$ \\
             $\aValues.\push(\aVal(v'))$ \\
             }
        }
        \catch{StopIteration}{
            $h.\pop()$ \\
            $\aValues.\pop()$ \\
        }
    }
    \KwRet{$\Sigma_{adm}$}
\end{algorithm}

\subsection{Proof of Theorem \ref{thm: adm_exist}}

\begin{proof}
    The result follows from Thm. \ref{thm: adm_str} that states that $\scoop$ or $\wcoop$ are necessary and sufficient conditions for a strategy to be admissible. Using the results of \citeauthor[Thm. 1]{brihaye2017pseudopolynomial}, we show that the value iteration algorithm on $\G$ will converge after finite iterations to a fixed point. Then, using \citeauthor[Corollary 18 and Proposition 19]{brihaye2017pseudopolynomial}, we prove that witnessing strategies for $\wco$ and $\coop$ always exists. 
     
    Further, from Thm. \ref{thm: wcoop_memoryless}, we know that for any history $h$, $\acVal(h)$ is same as the optimal cooperative value in the sub-game starting from $\last(h)$. Thus, witnessing strategies for $\wco$ and $\coop$ are sufficient conditions for $\wcoop$ to always exist. 
    
    Let us now define Strongly Cooperative Optimal strategy ($\sco$). See Def. \ref{def: sco} for a formal definition. A strategy that is $\sco$ is also $\scoop$ by definition. Thus, it suffices for us to prove that $\sco$ always exists to show that $\scoop$ always exists. For any history $h$, $\sco$ is a strategy such that if $\cVal(h) < \aVal(h)$ then $\cVal(h, \sigma) = \cVal(h)$ and if $\cVal(h) = \aVal(h)$ then $\aVal(h, \sigma) = \aVal(h)$. We observe that $\cVal(h, \sigma) = \cVal(h)$ is definition for $\coop$ and by definition $\cVal(h) \leq \aVal(h)$ and thus $\cVal(h, \sigma) \leq \aVal(h)$. As witnessing strategy for  $\coop$ always exists, $\sco$ strategies always exist. Thus. $\scoop$ always exists.
        
    As $\scoop$ and $\wcoop$ are necessary and sufficient (from Thm. \ref{thm: adm_exist}), we conclude that admissible strategies always exist.

\end{proof}

\subsection{Note on Theorem \ref{thm: adm_exist}}
\label{sec: note_on_witnessing_str}
\citeauthor[Corollary 18]{brihaye2017pseudopolynomial} is the proof for the existence of finite memory strategy for Sys player when the values of every state, after running the value iteration algorithm, are finite, i.e., not $\pm \infty$.
When the weights in $\G$ are non-negative, we can compute witnessing \emph{optimal memoryless} strategies for $\wco$ and $\coop$ \cite{khachiyan2008short}. This is because there does not exist Negative-Cycle (NC) plays (plays whose payoff is negative) in $\G$ that can ``fool" the Sys player to choose NC plays rather than reaching the goal state \cite{brihaye2017pseudopolynomial}. For states that belong to $V_{los}$, all valid actions from every state belong to an optimal strategy. If $v_0 \in V_{win}$, then an optimal memoryless strategy $\sigma$ exists which is $\wco$ such that $\aVal(v_0, \sigma) < \infty$. If $v \not\in V_{win}$, then Sys player is ``trapped," and there does not exist a strategy that enforces reaching a goal state. From such states, $\aVal(v, \sigma) = \infty \;\; \forall \sigma \in \Sigma$. If $v_0 \in V_{pen}$, then an optimal memoryless strategy $\sigma$  exists which is $\coop$ such that $\cVal(v_0, \sigma) < \infty$. If $v \in V_{los}$ then $\cVal(v, \sigma) = \aVal(v, \sigma) = \infty$ by definition and thus under no strategy for Sys and Env player can the play ever reach a goal state. Thus, a worst-case optimal strategy for Sys player picks an action indifferently when it is not in the winning region.  

\subsection{Proof of Lemma \ref{lem: pre_indep_adm}}

\begin{proof}
    We first show that $\Val$ is indeed history-independent on $\G'$. For all $h$ in $\G'$, such that $\last(h)$ is a not leaf node, as all the weights in $\G'$ are $0$, the payoff value for any prefix of history $h$ is exactly $0$. Informally, A history-independent $\Val$ is a function such that the payoff is independent of the sequence of states visited. Formally, $\Val(h_{\leq j}) = \Val(h) \; \forall j \in \mathbb{N}_{\leq |h| - 1}$. Here, $h_{\leq j}$ is a finite prefix of $h$ until the $(j+1)^{th}$ state, and $\mathbb{N}_{\leq k}$ is the set of natural numbers smaller or equal to $k$.
    
    For all histories $h$, in $\G'$, such that $\last(h)$ is a leaf node, by construction, the payoff is associated with the last state, i.e., $\Val(h) = \Val(\last(h))$. Specifically, $\Val(h) = \Val(\last(h)) \leq \B$ if the leaf node is a goal state for the Sys player else $\Val(h) = \Val(\last(h)) = \infty$. Thus, 
     \begin{align*}
         \cVal(h) = \min_{\sigma \in \Sigma} \min_{\tau \in \Tau} \Val(\last(\play^{v_0}(\sigma, \tau))) \\
         \aVal(h) = \min_{\sigma \in \Sigma} \max_{\tau \in \Tau} \Val(\last(\play^{v_0}(\sigma, \tau)))
     \end{align*}
    
    From above, we have that $\aVal(h) = \aVal(\last(h))$ and $\cVal(h) = \cVal(\last(h))$, and $\acVal(h) = \min\{\cVal(\last(h), \sigma) | \sigma \in \Sigma, \; \aVal(\last(h), \sigma) \leq \aVal(\last(h))\}$. 
    Thus, for $v':= \sigma(h)$, we can rewrite Eq. \eqref{eq: adm_eq_2} and get Eq. \eqref{eq: pre_indep_adm_eq_2}.
    
    While it is sufficient to look only at the last state along a history $h$ to compute $\aVal$,  for $\scoop$ condition, $\cVal(v') < \aVal(\last(h))$ does not imply that $\cVal(v') < \aVal(h_{\leq j}) \; \forall j \in \mathbb{N}_{|h| - 1}$. This is also evident from Figure \ref{fig: local_conds_not_sufficient} where the adversarial values for Sys player states along $h:= v_0 v_1 v_3 v_6 v_7$ and $v' := v_9$ is $\aVal(v_0) = \infty; \aVal(v_3) = 6; \aVal(v_7) = 10$. While $\cVal(v_9) < \aVal(v_7)$,  $\cVal(v_9) > \aVal(v_3)$. 
        
    Thus, we need to check if $\cVal(v') < \aVal(v) (\iff \cVal(v') < \min\{\aVal(v)\})$ for all Sys player states along $h$ to check for admissibility of strategy $\sigma$ compatable with $h$. As $\aVal$ is history-independent, there are finitely many adversarial values for a given graph $\G'$, and thus $\{\aVal(v)\}$ is finite in size.
        
    For $\wcoop$ condition, it is sufficient to evaluate $\acVal$, $\aVal$, and $\cVal$ at $\last(h)$ to compute the corresponding optimal values for payoff independent functions. Thus, to check if $\aVal(h) = \aVal(h, \sigma)$ it sufficies to check $\aVal(\last(h)) = \aVal(\last(h), \sigma)$. We can make the same argument for $\cVal(\last(h), \sigma) = \acVal(\last(h))$. Hence, we get the equations in Lemma \ref{lem: pre_indep_adm}. 
\end{proof}

\subsection{Proof of Theorem \ref{thm: sound_correct_naive_adm}}

\begin{proof}
    Let us define $\aVal(\G, v)$ and $\cVal(\G, v)$ to be the positional adversarial and cooperative value for state $v$ in $\G$. We extend the definition to $\acVal(\G, v)$ accordingly.
    We begin our proof by first noting that $\G'$ is a finite tree arena where the leaf nodes are partitioned into goal states for the Sys player and sink states. By construction, every state that is not a leaf node in $\G'$ has at least one outgoing edge. Only plays that reach a goal state in $\G'$ have a finite payoff value. Plays that fail to reach a goal state have a payoff of $+\infty$. Thus, using \citeauthor{khachiyan2008short}'s algorithm, we can compute optimal $\aVal(\G', v)$ and $\cVal(\G', v)$ and the corresponding witnessing strategies. Thus, game $\G'$ is well-formed, which implies $\scoop$ and $\wcoop$ always exists (from Thm. \ref{thm: adm_exist}). Hence,
    there always exists an admissible strategy in $\G'$. 
        
    Next, using Lemma \ref{lem: pre_indep_adm}, we see that the payoff function for $\G'$ is history-independent. Thus, checking for admissibility along a play reduces to checking for the adversarial and cooperative values at each state along a play $\G'$. If $\B < \cVal(\G, v_0)$, then there does not exist a play that reaches a goal state in $\G'$. Hence, $\aVal(\G', v) = \cVal(\G', v) = +\infty \; \forall v \in V$ and every valid action from every state is part of a strategy that is admissible. The Sys player may choose $\sigma$ indifferently in such cases. If $\B \geq \cVal(\G, v_0)$, there exists at least one play for which the payoff is finite. Hence, there exist states along such plays for which $\cVal(\G', v) \neq \infty$ and for all states $\aVal(\G', v) \leq \infty$ for every play. 
    
    In Thm. \ref{thm: scoop_memory}, we show that even for a tree-like arena whose payoff function is also history-independent, the admissibility of an action depends on the history. 
    As a consequence of Corollary \ref{cor: adm_memoryless}, we use Depth First Search (DFS) based approach to explore nodes along each play until all the nodes are explored. At each state, we check if the $\scoop$ or $\wcoop$ conditions are satisfied. If yes, then we add that history and the corresponding successor state as an admissible strategy to the set of all admissible strategies. Alg. \ref{algo: pseudo_naive_adm} returns such a strategy, and hence, the strategy returned is correct.
\end{proof}
\begin{algorithm}[t!]
    \caption{Admissbile Winning Synthesis}
    \label{algo: adm_win}
    \SetKwInOut{Input}{Input}\SetKwInOut{Output}{Output}
    \SetKwComment{Comment}{/* }{ */}
    \Input{Game $\G$, Budget $\B$}
    \Output{Strategy $\Sigma^{win}_{adm}$}
    \SetKwFunction{ValueIteration}{ValueIteration}
    \SetKwFunction{push}{push}
    \SetKwFunction{pop}{pop}
    \SetKwFunction{nxt}{next}
    \SetKwFunction{iter}{iter}
    \SetKwProg{try}{try}{:}{}
    \SetKwProg{catch}{catch}{:}{end}
     $\G' \gets$ Unroll $\mathcal{G}$ up until Payoff $\B$ \\
    \Comment{VI from Alg. \ref{algo: vi_algo}}
    $\aVal; V_{win} \gets$ min-max \ValueIteration($\G'$) \\
    $\cVal \gets$ min-min  \ValueIteration($\G'$) \\
    \For{$v \in V_s$ in $\G'$}{
        $\acVal(v) = \min \{\cVal(v') | \aVal(v') \leq \aVal(v) \; \text{where $v'$ is valid successor(s)}\}$ \\
    }
    \If{$\B <\cVal(v_0)$}
    {\KwRet{$\Sigma^{win}_{adm} := \Sigma$}}
    \Comment{\small Compute Admissible strategies-DFS}
    $h, \aValues \gets$ initialize empty stack \\
    $h.\push\big((v_0, \{\delta(v_0, a_s)\})\big)$ \\
    $\aValues.\push(\aVal(v_0))$ \\
    \While{$h \neq \emptyset$}{
        $v, \{v'\} \gets h[-1]$ \\
        \try{}{
            $v' \gets \nxt(\iter(\{v'\}))$ \\
             \If{$v \in V_s$}{
                 \lForAll{$s$ in $h$}{$hist := [s[0]]$}
                 \If{\colorbox{lightgray}{$\cVal(v') < \min\{\aValues\}$}}{
                    \If{\colorbox{lightgray}{$\neg (v \in V_{win}) \vee (v' \in V_{win})$}}{
                     $h.\push\big((v', \{\delta(v', a_e)\})\big)$ \\
                     $\Sigma^{win}_{adm}: hist \to v'$ \\
                     $\aValues.\push(\aVal(v'))$ \\
                     }
                 }
                 \ElseIf{$\aVal(v) = \aVal(v') = \cVal(v') = \acVal(v)$}{
                 $h.\push\big((v', \{\delta(v', a_e)\})\big)$ \\
                 $\Sigma^{win}_{adm}: hist \to v'$ \\
                 $\aValues.\push(\aVal(v'))$ \\
                 }
             }
             \If{$v \in V_e$}{
             $h.\push\big((v', \{\delta(v', a_s)\})\big)$ \\
             $\aValues.\push(\aVal(v'))$ \\
             }
        }
        \catch{StopIteration}{
            $h.\pop()$ \\
            $\aValues.\pop()$ \\
        }
    }
    \KwRet{$\Sigma^{win}_{adm}$}
\end{algorithm}

\begin{table}[tbh!]
    \centering
    \begin{tabular}{c c c}
        \toprule
         Qualitative & \multicolumn{2}{c}{Quantitative} \\
         \cmidrule{2-3}
          & $\scoop$ & $\wcoop$ \\
          \midrule
          $-1 \to -1$ & $ -1 \to -1$ & $ -1 \to -1$\\
          $0 \to 0, 1$ &  $0 \to 0, 1$ & $0 \to 0, 1$  \\
          $1 \to 1$ &  $1 \to 1, \textcolor{red}{0}$ & $1 \to 1$ \\
        \bottomrule
    \end{tabular}
    \caption{Value preservation: In qualitative settings, an admissible strategy never decreases its value. In quantitative settings, an admissible strategy $\scoop$ can decrease its state value (in red).}
    \label{tab: val_pre_tab}
\end{table}
    
\subsection{Proof of Lemma \ref{lem: optmistic_str}}
    \begin{proof}
    Given $\G$, the initial state of the game can belong to three regions. Thus, $v_0$ is either part of the losing, pending, or winning region. From Thm. \ref{thm: adm_str}, a strategy $\sigma$ is admissible if, and only if, it is $\scoop$ or $\wcoop$.
    
    \paragraph{Case I: $v_0 \in V_{los}$} By definition, there does not exist a play that reaches the goal state from $v_0$. Thus, for all histories $h$, $\cVal(h) = \aVal(h) = \infty$. If $\sigma$ is $\scoop$ then $\cVal(h, \sigma) = \aVal(h)$ which is true for every strategy in $\G$.    
    Thus, if $v_0 \in V_{los}$ then every strategy is $\scoop$ strategy. If $\sigma$ is $\wcoop$, then $\aVal(h) = \aVal(h, \sigma)$ and $\cVal(h, \sigma) = \acVal(h)$. Since, for all histories $h$, $\aVal(h) = \cVal(h) = \infty$ we have that $\acVal(h, \sigma) = \infty$. Thus, every action from every state in $\G$ belongs to an admissible strategy $\sigma$ that is $\wcoop$.
    
    \paragraph{Case II: $v_0 \in V_{pen}$} If $v_0 \in V_{pen}$, then there exists a play that reaches the goal state in $\G$. If $\sigma$ is $\scoop$ then it must satisfy, either $\cVal(h, \sigma) < \aVal(h)$ or $\cVal(h, \sigma) = \aVal(h, \sigma) = \aVal(h)$. We note that $\cVal(v) < \infty$ and $\aVal(v) = \infty$ for all $v \in V_{pen}$.
    Let $\last(h) \in V_{pen}$, then $\cVal(h) < \aVal(h)$. If $\sigma(h) \in V_{los}$, then $\cVal(h, \sigma) = \infty \not< \aVal(h)$ thus $\sigma$ is not $\scoop$ strategy. For any prefix $h$ that ends in Sys player state in the pending region, if $\sigma$ is $\scoop$, then $\sigma(h) \in V_{win} \cup V_{pen}$. Thus, $\sVal(v) \in \{0, 1\}$ for all states in plays induced by strategy $\sigma$ which is $\scoop$. 
    
    Let $\sigma$ be $\wcoop$ strategy and $h \in \plays^{v_0}(\sigma)$ such that $\last(h) \in V_{pen}$. Then, by definition, if $\sigma(h) \in V_{los}$ then $\cVal(h, \sigma) \neq \acVal(h)$ as $\cVal(h, \sigma) = \infty$ for all prefixes $h$. Thus, if $\sigma$ is $\wcoop$ strategy then $\sigma(h) \in V_{pen} \cup V_{win}$. Thus, $\sVal(v) \in \{0, 1\}$ for all states in plays induced by strategy $\sigma$ which is $\wcoop$. Hence, every strategy that is $\scoop$ or $\wcoop$ is value preserving if $\last(h) \in V_{pen}$.
    
    \paragraph{Case III: $v_0 \in V_{win}$} For this case, there exists a winning strategy $\sigma_{win}$ from the initial state. Thus, $\sVal(v) = 1$, for all the states in plays induced by the winning strategy as they always stay in the winning region. As every optimal winning strategy satisfies $\aVal(h) = \aVal(h, \sigma)$, it is worst-case optimal. Let $\Sigma_{win}$ and $\Sigma_{\wcoop}$ be the set of all winning and worst-case cooperative optimal strategies. Then, by definition, 
    we have $\Sigma_{\wcoop} \subseteq \Sigma_{win}$ in the winning region and every $\wcoop$ strategy is optimal winning strategy. Hence, $\wcoop$ strategies are value-preserving.
    
    Now, let $\sigma$ be $\scoop$ such that $\last(h) \in V_{win}$. Then if $\sigma(h) \in V_{los}$ then $\cVal(h, \sigma) > \aVal(h)$. If $\sigma(h) \in \{V_{pen} \cup V_{win}\}$ then $\cVal(h, \sigma)$ can be less than $\aVal(h)$. Thus, for all prefixes $h$ such that $\last(h) \in V_{win}$, if $\sigma$ is $\scoop$ then it is not value preserving. 
    
    From above, we have that, for any history $h$, $\wcoop$ are value preserving strategies and hence are not optimistic strategies. For histories $h$ where $\last(h) \in V_{win}$, if $\sigma$ is $\scoop$ then it \emph{may not} be value-preserving. Thus, if an optimistic strategy exists, then it is a $\scoop$ strategy. Table \ref{tab: val_pre_tab} summarizes the value preservation property for admissible strategies in different regions in the game.
\end{proof}

\subsection{Proof of Theorem \ref{thm: adm_win_str}}

In order to prove this, we will first show that every strategy that is either $\modscoop$ or $\wcoop$ is admissible winning. 

\begin{lemma}
    All $\modscoop$ strategies are admissible winning
    \label{lem: modscoop_adm_win}
\end{lemma}
\begin{proof}
    By definition, a strategy that is $\modscoop$ is (i) $\scoop$ and (ii) value-preserving. From Lemma \ref{lem: optmistic_str}, if there exists an optimistic strategy $\sigma$, then a strategy exists that is $\scoop$ are \emph{not} value-preserving. Specifically, for any history $h$, if $\last(h) \in V_{win}$, then $\sigma(h) \in V_{win} \cup V_{pen}$ and thus $\sigma$ is not value-preserving.

    Further, for history $h$, if $\last(h) \in V_{pen} \cup V_{los}$ then a strategy that is $\modscoop$ is also $\scoop$. Thus, we only look at the case where $\last(h) \in V_{win}$. If $\last(h) \in V_{win}$ then a strategy $\sigma$ is $\modscoop$, if $\cVal(h) < \aVal(h)$ then $\cVal(h, \sigma) < \aVal(h) \wedge \aVal(\sigma(h)) \neq \infty$ else if $\cVal(h) = \aVal(h)$ then $\aVal(h, \sigma) = \cVal(h, \sigma) = \aVal(h)$.

    Thus, we need to prove that every play in $\plays^{h}(\sigma)$ reaches a goal state with a payoff less than or equal to $\B$, i.e., $\aVal(h, \sigma) \leq \B$ or $\aVal(h) \leq \B$. For the second condition, if $\cVal(h) = \aVal(h)$ then $\sigma$ is worst-case optimal as $\aVal(h, \sigma) = \aVal(h)$. By definition, all worst-case optimal strategies ensure reaching the goal state such that $\aVal(h) \leq \B$. Thus, strategy $\sigma$ that is $\wco$ and $\modscoop$ will reach a goal state with a payoff less than or equal to $\B$.

    For the second condition, if $\cVal(h) < \aVal(h)$ then strategy $\sigma$ is $\modscoop$ if  $\cVal(h, \sigma) < \aVal(h) \wedge \aVal(\sigma(h))~\neq~\infty$, that is, $\sigma$ is $\scoop$ and never takes an action from $\last(h)$ in winning region that evolves to the pending or losing region (as such actions do not satisfy ($\sVal(\last(h)) = 1 \implies \sVal(\delta(\last(h), \sigma(h))) = 1$). Thus, every play under strategy $\sigma$ will remain in the winning region. As from every state in the winning region, there exists a winning strategy such that $\aVal(h) < \infty$, the plays under a value-preserving $\modscoop$ strategy always remain in the winning region, and thus value-preserving $\modscoop$ is a winning strategy. 
    Thus, every strategy that $\modscoop$ is admissible winning.  
    
\end{proof}

\begin{lemma}
    \label{lem: wcoop_adm_win}
    All $\wcoop$ strategies are admissible winning
\end{lemma}
\begin{proof}
    From Lemma \ref{lem: wcoop_proof}, we have that every strategy $\sigma$ that is $\wcoop$ is admissible. For a strategy to be admissible winning, it needs to satisfy an additional requirement that:  $\forall h \in \plays^{v_0}(\sigma)$, $\aVal(h) < \infty \implies \Val(h \cdot \play^{h}(\sigma, \tau)) \leq \B$ for all Env player strategy. 

    Informally, if there exists an optimal winning strategy $\sigma$ from history $h$, then all plays under strategy $\sigma$ should reach a goal state with payoff less than or equal to $\B$, i.e., $\aVal(h, \sigma) \leq \B$. As $\aVal(h) = \aVal(h, \sigma)$, we have $\aVal(h) \leq \B$. 

    By definition, a $\wcoop$ strategy $\sigma'$ is a strategy that is worst-case optimal i.e., $\aVal(h, \sigma') = \aVal(h)$. Thus, every strategy that is $\wcoop$ will also be a $\wco$ strategy. Thus, $\sigma'$ ensures $\aVal(h, \sigma') = \aVal(h) = \aVal(h, \sigma)$. Thus, $\aVal(h, \sigma') \leq \B$.

    Hence, every strategy that is $\wcoop$ is admissible winning.
\end{proof}

Now we use the same techniques from Thm. \ref{thm: adm_str} to prove that $\modscoop$ and $\wcoop$ are necessary and sufficient conditions for a strategy to be admissible winning. 

\begin{proof}[Proof of Thm. \ref{thm: adm_win_str}]
    We first start by noting the following. From Lemma \ref{lem: weak_dom}, we have that negation of Eqs.~\eqref{eq: adm_win_eq_1} and \eqref{eq: adm_win_eq_2} is as follows:
    \begin{subequations}
        \begin{align}
            & \Big( \big( \cVal(h, \sigma) \geq \aVal(h) \wedge \aVal(h, \sigma) > \aVal(h) \big) \vee \label{eq: adm_win_dom_eq_1}\\
            & \hspace{5cm} \neg(p \implies q) \Big) \label{eq: adm_win_dom_eq_2} \\
            & \big( \aVal(h) = \aVal(h, \sigma) = \cVal(h, \sigma) \; \wedge \notag \\
            & \hspace{4 cm} \acVal(h) < \aVal(h) \big) \label{eq: adm_win_dom_eq_3}
        \end{align}
    \end{subequations}

    Here is $p$ and $q$ are logical propositions that correspond to $\sVal(\last(h)) = 1$ and $\sVal(\delta(\last(h), \sigma(h))) = 1$, respectively. Thus, by tautology, $\neg(p \implies q)$ is true if and only if $\sVal(\last(h)) = 1$ is true and $\sVal(\delta(\last(h), \sigma(h))) = 1$ is false. Informally, $\neg(p \implies q)$ is true if $\last(h)$ belongs to the winning region but $\sigma(h)$ does not belong to the winning region. Further, Eq.~\eqref{eq: adm_win_dom_eq_1} is the same as Eq.~\eqref{eq: dom_eq_1} and Eq.~\eqref{eq: adm_win_dom_eq_3} is the same as Eq.~\eqref{eq: dom_eq_2}. 
    Given an admissible winning strategy $\sigma$, and prefix $h \in \plays^{v_0}(\sigma)$, we want to prove that $\sigma \iff \modscoop \vee \wcoop$, i.e., $\sigma$ is the witnessing strategies that satisfies $\modscoop$ or $\wcoop$ conditions. We use the following tautology to prove iff condition. 
    \begin{equation*}
        (A \iff B) \iff \big[(B \implies A) \wedge (\neg B \implies \neg A)\big]   
    \end{equation*}
    
    Here A is the logical statement that $\sigma$ is admissible winning as per Def. \ref{def: adm_win} and $B$ is the logical statement that implies either Eq. \eqref{eq: adm_win_eq_1} or Eq. \eqref{eq: adm_win_eq_2} is true.

    \paragraph{Case I: $(\neg B \implies \neg A)$} For all $h \in \plays^{v_0}$, if $\last(h) \in V_{pen} \cup V_{los}$ then $\modscoop$
    strategies are the same as the $\scoop$ strategy. Thus, from Thm. \ref{thm: adm_exist}'s Case I we have that any strategy $\sigma$ that satisfies Eq.~\eqref{eq: adm_win_dom_eq_1} is weakly dominated by another strategy that is $\wco$. Thus, $\sigma$ is not admissible, i.e., $\neg A$ holds. if $\last(h) \in V_{win}$, then $\sigma$ either satisfies Eq.~\eqref{eq: adm_win_dom_eq_1} or Eq.~\eqref{eq: adm_win_dom_eq_2}. We note that Eqs.~\eqref{eq: adm_win_dom_eq_1} and \eqref{eq: adm_win_dom_eq_2} can be rewritten as:
    \begin{equation*}
        \cVal(h, \sigma) < \aVal(h) \implies \neg (p \implies q)
    \end{equation*}

    Thus, if $\sigma$ satisfies $\cVal(h, \sigma) < \aVal(h)$ (it is $\scoop$) then $\neg (p \implies q)$ is true, i.e., $\last(h) \in V_{win}$ and $\sigma(h) \notin V_{win}$. 
    From Lemma \ref{lem: optmistic_str}, we have that for every play under strategy $\sigma$ with $\last(h) \in V_{win}$, $\sigma(h) \in V_{win} \cup V_{pen}$. As $\sigma(h) \notin V_{win} \implies \sigma(h) \in V_{pen}$.  
    If $\sigma(h) \in V_{pen}$ then by definition, there exists an Env strategy such that $\Val(h \cdot \play^{h}(\sigma, \tau)) = \infty$. Thus, $\sigma$ is not a winning strategy. Hence, $\sigma$ is not admissible winning and $\neg A$ holds. 
    
    For any history $h \in \plays^{v_0}$, if strategy $\sigma$ satisfies Eq.~\eqref{eq: adm_win_dom_eq_3} then from Thm. \ref{thm: adm_exist}'s Case I, we have that there exists $\sigma'$ that is $\wcoop$ that does as good as $\sigma$ and there exists a play under $\sigma'$ with a payoff strictly lower than any payoff under strategy $\sigma$. Thus, $\sigma$ is not admissible, and $\neg A$ holds. 

    \paragraph{Case II: $(B \implies A)$}  Assume that for all prefixes $h$ of $\plays^{v_0}(\sigma)$ we have that Eqs.~\eqref{eq: adm_win_eq_1} and \eqref{eq: adm_win_eq_2}  hold such that $\sigma$ is the witnessing strategy for either $\modscoop$ or $\wcoop$. From Lemma \ref{lem: modscoop_adm_win} and \ref{lem: wcoop_adm_win}, we have a strategy that is either $\modscoop$ or $\wcoop$ is always admissible winning. 

    Since we have shown both $(\neg B \implies \neg A)$ and $(B \implies A)$, hence, for all prefixes starting from $v_0$, Eqs.~\eqref{eq: adm_win_eq_1} and \eqref{eq: adm_win_eq_2} are necessary and sufficient conditions for strategy $\sigma$ to be admissible winning. 
\end{proof}

\subsection{Proof of Lemma \ref{lem: adm_win_exist}}
\begin{proof}
    From Thm. \ref{thm: adm_win_str} we have that $\wcoop$ is a necessary and sufficient condition for a strategy to be admissible winning. From Thm.~\ref{thm: wcoop_memoryless}, we have that $\wcoop$ always exists. Thus, we conclude that admissible winning strategies always exist. 

    Further, we can show that $\modscoop$ strategies always exist. To prove this, we observe that $\modscoop$ is a subset of $\scoop$. Specifically, $\forall h \in \plays^{h}(\sigma)$ such that $\aVal(h) < \infty$, $\modscoop$ are exactly the strategies that are $\scoop$ and value-preserving. From Lemma~\ref{lem: modscoop_adm_win}, we have that $\modscoop$ are winning strategies for all $h$ such that $\aVal(h) < \infty$. Thus, witnessing strategies for worst-case strategies from the winning region is a sufficient condition for $\modscoop$ to always exist. From \cite{brihaye2017pseudopolynomial} and \ref{sec: note_on_witnessing_str}, we can compute witnessing strategies for all worst-case values for every history $h$ in $\G$. Thus, $\modscoop$ strategies always exist.
\end{proof}

\subsection{Proof of Theorem \ref{thm: adm_win_memory}}
\begin{proof}
    By definition a strategy $\sigma$ is admissible winning if it $\scoop$ and $\sVal(\last(h))~=~1 \implies \sVal(\delta(\last(h), \sigma(h)))~=~1$. We observe that the value-preserving constraint is Markovian, i.e., it only depends on $\last(h)$ (last state along a history) and $\delta(\last(h), \sigma(h))$ (the next state under strategy $\sigma$). Thus, it is sufficient to check for the value-preserving constraint memoryless-ly, i.e., independent of the sequence of states in $h$.  

    From Thm. \ref{thm: scoop_memory}, we know that we need to keep track of states along a history to check if $\cVal(h, \sigma) < \aVal(h)$ ($\scoop$ condition) is satisfied. Thus, memoryless strategies are not sufficient for admissible winning strategies.
\end{proof}

\subsection{Proof of Theorem \ref{thm: sound_correct_adm_win}}
\begin{proof}
From Sec. \ref{sec: adm_win_str}, we have that Alg. \ref{algo: adm_win} is same as Alg. \ref{algo: naive_adm} except for the admissibility checking criteria. 

The only difference in Eqs.~\eqref{eq: adm_win_eq_1} and \eqref{eq: adm_win_eq_2} and Eq.~\eqref{eq: adm_eq} is the value preserving constraint. We observe that $\sVal(\last(h)) = 1 \implies \sVal(\delta(\last(h), \sigma(h))) = 1$ constraint is Markovian, i.e., it only depends on $\last(h)$ (last state along a history) and $\delta(\last(h), \sigma(h))$ (the next state under strategy $\sigma$).
 By tautology, we express the above statement as 
 $$\neg \left(\sVal(\last(h)) = 1\right) \vee \left(\sVal(\delta(\last(h), \sigma(h))) = 1\right).$$
 As $\sVal(v) = 1$ only in the winning region, we have the following equivalent statement $\neg (v \in V_{win}) \vee (v' \in V_{win})$ where $v:= \last(h)$ and $v':= \delta(\last(h), \sigma(h))$. Thus, modifing the admissibility checking criteria in Line 9 of Alg. \ref{algo: pseudo_naive_adm} to $\cVal(v') < \min\{\aValues\}$ and $\neg (v \in V_{win}) \vee (v' \in V_{win})$ is sufficient to capture the admissibility criteria for $\modscoop$. Thus, from Thms.~\ref{thm: sound_correct_naive_adm} and \ref{thm: adm_win_str}, we have that  Alg. \ref{algo: adm_win} returns the set of all admissible winning strategies, is sound and correct. 
 \end{proof}

\subsection{Note on Alg. \ref{algo: vi_algo}}
\label{sec: vi_algo_note}
Alg. \ref{algo: vi_algo} is the Value Iteration algorithm for computing optimal $\Val$s and optimal strategy from every state in $\G$ \cite{brihaye2017pseudopolynomial} (Adapted and Modified from Alg. 1). Line $\ref{line: max_computation}$ and Line $\ref{line: min_computation}$ correspond to computation when the Env player is playing adversarially (maximally) and cooperatively (minimally), respectively. When Env is playing adversarially, Alg. \ref{algo: vi_algo} returns optimal $\aVal$ and worst-case optimal strategy $\wco$. For $\aVal$ computation, $\sigma_{win}$ is the optimal winning strategy from every state in $V_{win}$.

When Env is playing cooperatively, Alg. \ref{algo: vi_algo} returns optimal $\cVal$ and cooperative optimal strategy $\coop$. For $\cVal$ computation, $\sigma_{win} = \sigma_{coop}$ and $V_{win}$ is the set of states from which $\cVal(v) \neq \infty$. We can observe that the algorithm is based on fixed point computation, and thus, the algorithm runs in polynomial time.

\subsection{Note on Alg. \ref{algo: naive_adm} and Alg. \ref{algo: adm_win}}

Alg. \ref{algo: naive_adm} and Alg. \ref{algo: adm_win} compute admissible and admissible winning strategies, respectively. The key components are unrolling the graph, running value iteration, and an algorithm for Depth-First Search (DFS). The DFS algorithm searches over $\G'$ iteratively rather than recursively in order to avoid Python's low recursion limit. Specifically, the algorithm explores states in preorder fashion. Alg. \ref{algo: vi_algo} gives the outline to compute optimal $\aVal$ and $\cVal$ and the corresponding witness strategies. See \ref{sec: vi_algo_note} for more details. After we unroll the graph, we compute $\aVal$, $\cVal$, and $\acVal$ for every state and finally run the DFS algorithm for strategy synthesis. 

We initialize a stack $h$ to keep track of states visited so far, and their corresponding adversarial values are added to $\aValues$'s stack. Note, $h$ is a stack of tuples where the first element is the states visited, and the second element is an iterator (Line 14 in Algs.~\ref{algo: naive_adm} and \ref{algo: adm_win}) over the successors of $v$ in $\G'$. We start the search by adding the initial state to $h$ and iterate over its successors $v'$ using the \texttt{next} operator. After exploring a successor $v'$, we break out of the \texttt{try} block and explore the successors of $v'$. Once, we explore all the successors of a state, we pop that state from  $h$ and $\aValues$ stacks and continue till every state has been popped from $h$. The highlighted section in Algs.~\ref{algo: naive_adm} and \ref{algo: adm_win} show the difference in admissibility checking in both algorithms. 

\section{Note on $\sco$ vs $\scoop$}

\citeauthor{brenguier2016admissibility} define strongly-cooperative optimal strategies in order to show that every $\sco$ is admissible. Then, they proceed to shows that $\sco$ strategies always exists in determined (well-formed as per their terminology) games, i.e., you can always find a strategy $\sigma$ such that $\aVal(h) = \aVal(h, \sigma)$ and $\cVal(h) = \cVal(h, \sigma)$ for every $h \in \plays^{v}(\sigma)$.

In our case, we define $\scoop$ strategy that is a superset of $\sco$. $\sco$ are maximal in dominance order, but so are $\scoop$.

\begin{definition}[$\sco$]
    \label{def: sco}
    For all $h \in \plays^{v_0}(\sigma)$, strategy $\sigma$ is Strongly Cooperative Optimal $(\sco)$, 
    if $\cVal(h) < \aVal(h)$ then 
    $\cVal(h, \sigma) = \cVal(h)$; if $\cVal(h) = \aVal(h)$ then $\aVal(h, \sigma) = \aVal(h)$.
\end{definition}

\begin{lemma}
    All $\sco$ strategies are admissible.
    \label{lem: sco_proof}
\end{lemma}
\begin{proof}
    \begin{figure}[h]
    \centering
    \resizebox{0.45\linewidth}{!}{%
        \centering
        \begin{tikzpicture}
            [->,>=stealth',shorten >=1pt,shorten <=1pt,auto,node distance=1cm,
            every loop/.style={looseness=6},
            initial text={},
            el/.style={font=\scriptsize},
            every fit/.style={draw,densely dotted,rectangle},
            inner sep=2mm,
            decoration = {snake,pre length=3pt,post length=7pt},
            loopright/.style={loop,looseness=6,out=-45, in=45},
            loopleft/.style={loop,looseness=6,out=135, in=225},
            loopabove/.style={loop,looseness=6,out=45, in=135},
            loopbelow/.style={loop,looseness=6,out=-135, in=-45},]
        \tikzstyle{every state}=[node distance=1.4cm,minimum size=7mm, inner sep=1pt];
        \node[state] at (0,0) (v0){$v_s$};
        \node[state, left of=v0, node distance=2.4cm,] (v4) {$v_0$};
        \node[state, rectangle, node distance=2.4cm, above right of=v0] (v1){$\sigma(h)$};
        \node[state, rectangle, node distance=2.4cm, below right of=v0] (v2) {$\sigma'(h)$};
        \path[-latex'] 
          (v4) edge[decorate] node[above] {$h$} (v0)
          (v0) edge[decorate] node {$\sigma$ is $\sco$}(v1)
          (v0) edge[decorate] node[left] {$\sigma'$ is not $\sco$}(v2);
        \end{tikzpicture}
        }
    \caption{$\sco$ proof example}
    \label{fig: sco_proof}
\end{figure}

    Let $\sigma$ be $\sco$. Assume there exists $\sigma' \neq \sigma$ that is compatible with history $h$, $\last(h) = v_s$, and ``splits" at $v_s$ as shown in Figure \ref{fig: sco_proof}. Thus, $\sigma(h) \neq \sigma'(h)$. We note that only two cases for a strategy are possible, i.e., it is either $\sco$ or not. 
    Further, let's say that $\sigma'$ weakly dominates $\sigma$. We show that $\sigma'$ does not weakly dominate $\sigma$. 

    \paragraph{Case I}$\cVal(h) < \aVal(h)$: As $\sigma'$ is not $\sco$ this implies $\cVal(h, \sigma') \geq \aVal(h)$ and as $\cVal(h, \sigma) \leq \aVal(h, \sigma)$ for any $\sigma \in \Sigma$, we get 
    $$\aVal(h, \sigma') \geq \cVal(h, \sigma') \geq \aVal(h) > \cVal(h) = \cVal(h, \sigma).$$ 
    
    On simplifying, we get $\aVal(h, \sigma') \geq \aVal(h) > \cVal(h, \sigma)$. This statement implies that there exists $\tau \in \Tau$ such that $\Val(h \cdot \play^{v_s}(\sigma', \tau)) > \Val(h \cdot \play^{v_s}(\sigma, \tau))$. Note that since, $\Val(h \cdot 
    \play^{v_s}(\sigma, \tau)) = \Val(h) + \Val(\play^{v_s}(\sigma, \tau))$ we get $\Val(\play^{v_s}(\sigma', \tau)) > \Val(\play^{v_s}(\sigma, \tau))$. This contradicts the assumption that $\sigma'$ dominates $\sigma$ as $\sigma'$ should always have a payoff that is equal to or lower than $\sigma$. 

    \paragraph{Case II}$\cVal(h) = \aVal(h)$: For this case we get,
    $$\aVal(h, \sigma') \geq \cVal(h, \sigma') \geq \aVal(h) = \cVal(h) = \cVal(h, \sigma).$$ 

    This implies that $\cVal(h, \sigma') \geq \cVal(h)$. But, since $\sigma'$ dominates $\sigma$, there should exist a strategy $\tau \in \Tau$ under which $\sigma'$ does strictly better than $\sigma$. Since, $\cVal(h, \sigma') \geq \cVal(h) \forall \tau \in \Tau$, it implies that there does not exists a payoff $\Val(\play^{v_s}(\sigma', \tau))$ that has a payoff strictly less than $\cVal(h, \sigma) = \cVal(h)$. Thus, $\sigma'$ does not dominate $\sigma$. 

    We can repeat this for all histories $h$ in $\plays^{v_0}(\sigma)$.
    Hence, every $\sigma$ that is $\sco$ is admissible. 
\end{proof}

Intuitively, $\sco$ strategies, just like $\scoop$, let the Sys player take a ``riskier" action as long there exists a lower payoff along that play. This makes the strategy optimistic.

\end{document}